\newif\ifreport\reporttrue
\documentclass[journal]{IEEEtran}
\usepackage{amsmath,amssymb,amsfonts,amsthm}
\usepackage{subcaption}
\usepackage{color}
\usepackage{algorithm}
\usepackage{algpseudocode}

\newtheorem{theorem}{Theorem}
\newtheorem{corollary}{Corollary}
\newtheorem{definition}{Definition}
\newtheorem{lemma}{Lemma}

\def\blue{\color{blue}}
\def\red{\color{red}}
\def\violet{\color{violet}}
\usepackage{tcolorbox}
\usepackage{lipsum}
\usepackage{hyperref}
\usepackage[noadjust]{cite}
 \begin{document}

\title{How Does Data Freshness Affect Real-time Supervised Learning?}
        
\author{Md Kamran Chowdhury Shisher,~\IEEEmembership{Student~Member,~IEEE,}
        Yin~Sun,~\IEEEmembership{Senior~Member,~IEEE}
\IEEEcompsocitemizethanks{\IEEEcompsocthanksitem M.K.C. Shisher and Y. Sun are with the Department of Electrical and Computer Engineering, Auburn University, Auburn, AL, 36849. This paper is accepted in part at ACM MobiHoc 2022 \cite{ShisherMobihoc}. This work was supported in part by the NSF grant CCF-1813078 and the ARO grant W911NF-21-1-0244.}
}

\newcommand{\ignore}[1]{{}}
\pagestyle{plain}
\def\blue{\color{blue}}
\maketitle

\begin{abstract}
In this paper, we analyze the impact of data freshness on real-time supervised learning, where a neural network is trained to infer a time-varying target (e.g., the position of the vehicle in front) based on features (e.g., video frames) observed at a sensing node (e.g., camera or lidar). One might expect that the performance of real-time supervised learning degrades monotonically as the feature becomes stale. Using an information-theoretic analysis, we show that this is true if the feature and target data sequence can be closely approximated as a Markov chain; it is not true if the data sequence is far from Markovian. Hence, the prediction error of real-time supervised learning is a function of the Age of Information (AoI), where the function could be non-monotonic. Several experiments are conducted to illustrate the monotonic and non-monotonic behaviors of the prediction error. To minimize the inference error in real-time, we propose a new ``selection-from-buffer'' model for sending the features, which is more general than the ``generate-at-will'' model used in earlier studies. By using Gittins and Whittle indices, low-complexity scheduling strategies are developed to minimize the inference error, where a new connection between the Gittins index theory and Age of Information (AoI) minimization is discovered. These scheduling results hold (i) for minimizing general AoI functions (monotonic or non-monotonic) and (ii) for general feature transmission time distributions. Data-driven evaluations are presented to illustrate the benefits of the proposed scheduling algorithms.
\end{abstract}

\begin{IEEEkeywords}
Age of Information, supervised learning, scheduling, Markov chain, buffer management.
\end{IEEEkeywords}

\section{Introduction}
\IEEEPARstart{I}{n} recent years, the proliferation of networked control and cyber-physical systems such as autonomous vehicle, UAV navigation, remote surgery, industrial control system has significantly boosted the need for real-time prediction. For example, an autonomous vehicle infers the trajectories of nearby vehicles and the intention of pedestrians based on lidars and cameras installed on the vehicle \cite{mozaffari2020deep}. In remote surgery, the movement of a surgical robot is predicted in real-time. These prediction problems can be solved by real-time supervised learning, where a neural network is trained to predict a time varying target based on feature observations that are collected from a sensing node. Due to data processing time, transmission errors, and  queueing delay, the features delivered to the neural predictor may not be fresh. The performance of networked intelligent systems depends heavily on the accuracy of real-time  prediction. Hence, it is important to understand how data freshness affects the performance of real-time supervised learning. 

To evaluate data freshness, a metric \emph{Age of information} (AoI) was introduced in \cite{kaul2012real}. Let $U_t$ be the generation time of the freshest feature received by the neural predictor at time $t$. Then, the AoI of the features, as a function of time $t$, is defined as $\Delta(t)=t-U_t$, which is the time difference between the current time $t$ and the generation time $U_t$ of the freshest received feature. The age of information concept has gained a lot of attention from the research communities. Analysis and optimization of AoI were studied in various networked systems, including remote estimation, control system, and edge computing. In these studies, it is commonly assumed that the system performance degrades monotonically as the AoI grows. Nonetheless, this is not always true in real-time supervised learning. For example, it was observed that the predictor error of day-ahead solar power forecasting is not a monotonic function of the AoI, because there exists an inherent daily periodic changing  
pattern in the solar power time-series data \cite{shisher2021age}. 

In this study, we carry out several experiments and present an information-theoretic analysis to interpret the impact of data freshness in real-time supervised learning. In addition, we design buffer management and transmission scheduling strategies to improve the accuracy of real-time supervised learning. The key contributions of this paper are summarized as follows:
\begin{itemize}
\item We develop an  information-theoretic approach to  analyze how the AoI affects the performance of  real-time supervised learning. It is shown that the prediction errors (training error and inference error) are functions of AoI, whereas they could be non-monotonic AoI functions --- this is a key difference from previous studies on AoI functions, e.g., \cite{kosta2017age,SunNonlinear2019, sun2017update,Tripathi2019}. When the target and feature data sequence can be closely approximated as a Markov chain, the prediction errors are non-decreasing functions of the AoI. When the target and feature data sequence is far from Markovian, the prediction errors could be non-monotonic in the AoI (see Sections 2-3).

\item We conduct several experiments and observe that, due to long-range dependence, response delay, and/or communication delay, the target and feature data sequence can be far from Markovian and the corresponding prediction errors are  non-monotonic AoI functions. In certain scenarios, even a fresh feature (AoI=0) may generate larger prediction errors than stale features (AoI $>$ 0), i.e., the freshest feature may not be the best feature; \ifreport see Figs. 2-3 for an illustration. \else see Fig. 2 for an illustration. \fi

\begin{figure*}[h]
  \centering
  
  \begin{subfigure}[t]{0.45\textwidth}
\includegraphics[width=\textwidth]{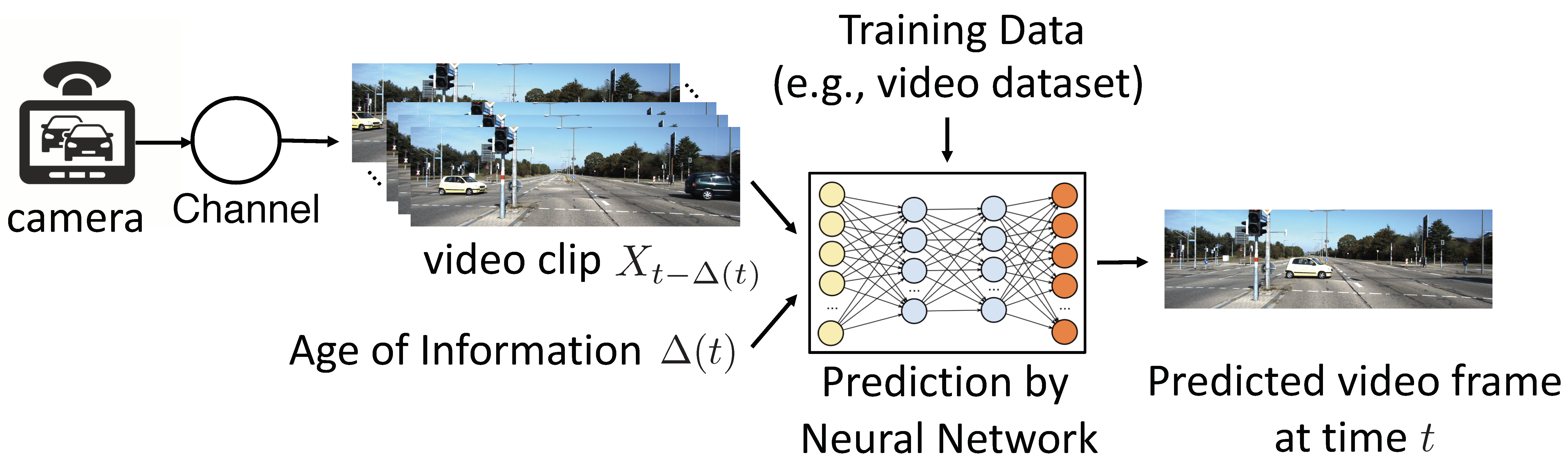}
  \subcaption{Video prediction Task}
\end{subfigure}
%
\hspace{3mm} 
\begin{subfigure}[t]{0.20\textwidth}
\includegraphics[width=\textwidth]{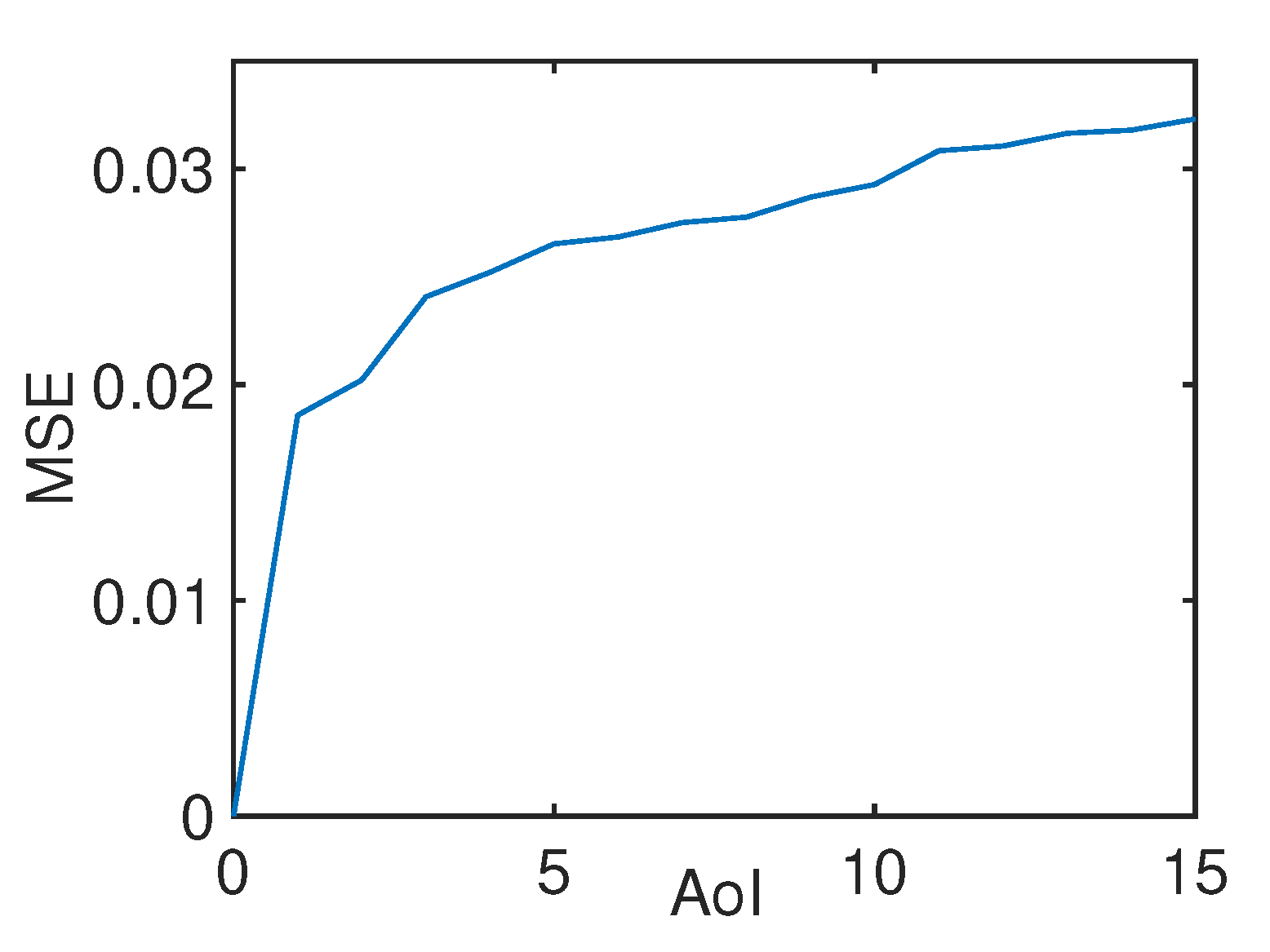}
\subcaption{Training Error vs. AoI}
\end{subfigure}
%
\hspace{3mm}
\begin{subfigure}[t]{0.20\textwidth}
\includegraphics[width=\textwidth]{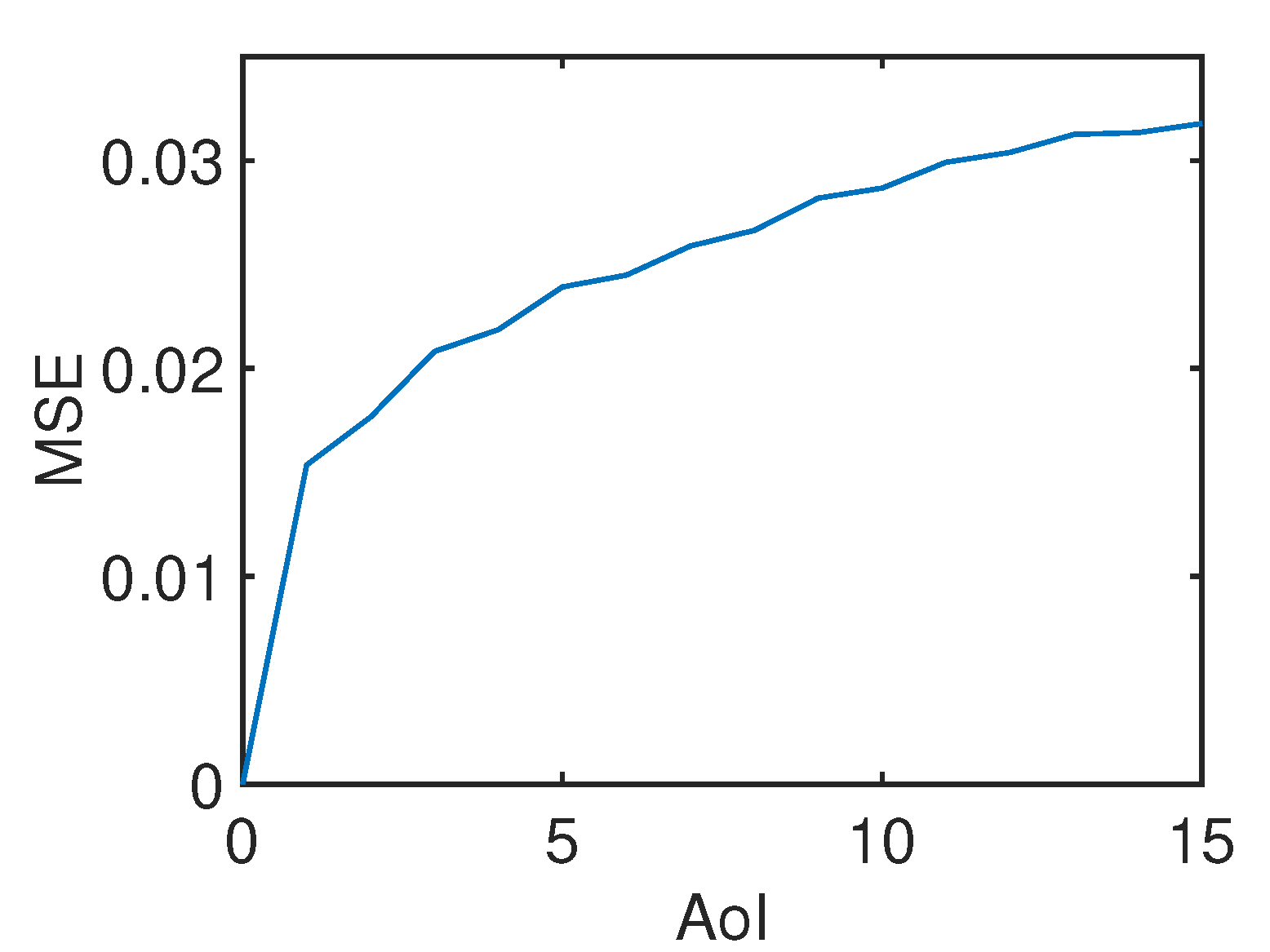}
\subcaption{Inference Error vs. AoI}
\end{subfigure}
\caption{\small Performance of supervised learning based video prediction. The experimental results in (b) and (c) are regenerated from \cite{lee2018stochastic}. The training and inference errors are non-decreasing functions of the AoI. \label{fig:learning}}
\end{figure*}

\item 
We propose  buffer management and transmission scheduling strategies to minimize the inference error. Because the inference error could be a non-monotonic AoI function, we introduce a novel ``selection-from-buffer'' model for feature transmissions, which is more general than the ``generate-at-will'' model used in many earlier studies, e.g., \cite{sun2017update, SunNonlinear2019, yates2015lazy}. If the AoI function is non-decreasing, the ``selection-from-buffer'' model achieves same performance as the ``generate-at-will'' model; if the AoI function is non-monotonic, the ``selection-from-buffer'' model can potentially achieve better performance. 

\item In the single-source case, an optimal scheduling policy is devised to minimize the long-term average inference error.
By exploiting a new connection with the Gittins index theory \cite{gittins2011multi}, the optimal scheduling policy is proven to be a threshold policy on the Gittins index (Theorems \ref{theorem5}-\ref{theorem6}), where the threshold can be computed by using a low complexity algorithm like bisection search. This  scheduling policy is more general than the scheduling policies proposed in \cite{SunNonlinear2019, sun2017update}.
\item In the multi-source case, a Whittle index scheduling policy is designed to reduce the weighted sum of the inference errors of the sources. By using the Gittins index obtained in the single-source case, a semi-analytical expression of the Whittle index is obtained (Theorems \ref{theorem7}-\ref{theorem8}), which is more general than the Whittle index formula in \cite[Equation (7)] {Tripathi2019}. 

\item The above scheduling results hold (i) for minimizing general AoI functions (monotonic or non-monotonic) and (ii) for general feature transmission time distributions. Data driven evaluations show that ``selection-from-buffer” with optimal scheduler achieves up to $3$ times smaller inference error compared to ``generate-at-will,” and $8$ times smaller inference error compared to periodic feature updating (see Fig. \ref{fig:singlesourcedifferentsigma}). Whittle index policy achieves up to $2$ times performance gain compared to maximum age first (MAF) policy (see Fig. \ref{fig:multisourceweight}).    
\ignore{\item Whittle index approach is used to develop a multiple-source scheduling policy for minimizing the weighted sum of inference errors of $m$ sources. A semi-analytical expression for finding the Whittle index policy is provided by using the Gittins index.}
\end{itemize}
\ignore{In this work, we analyze the impact of data freshness on real-time supervised learning. We also develop an optimal feature buffering and transmission strategy that minimizes the prediction error and has low complexity.}

\ignore{\begin{figure}[t]
\centering
\includegraphics[width=0.20\textwidth]{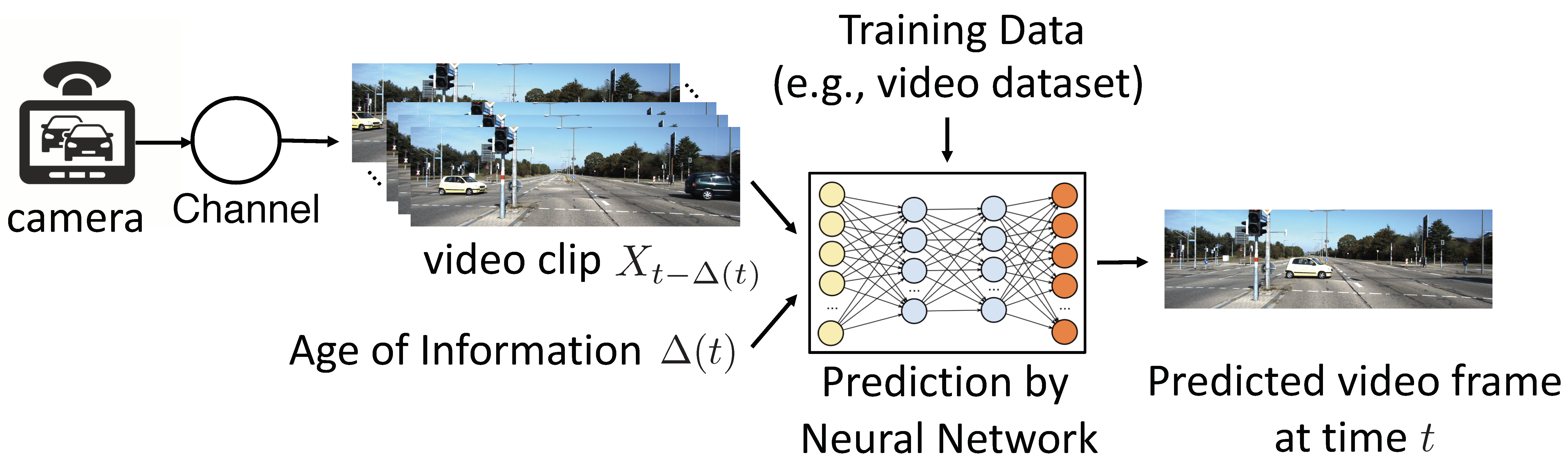}
\caption{\small A real-time supervised learning system. \label{fig:learning}
}
\end{figure}}

\subsection{Related Works}
In recent years, AoI has become a popular research topic \cite{yates2021age}. Average AoI and average peak AoI are studied in many queueing systems \cite{kaul2012real, sun2017update, yates2015lazy}. As surveyed in \cite{SunNonlinear2019}, there exist a number of applications of non-linear AoI functions, such as auto-correlation function \cite{kosta2017age}, estimation error \cite{SunTIT2020, orneeTON2021, klugel2019aoi}, and Shannon's mutual information and conditional entropy \cite{SunNonlinear2019}. In existing studies on AoI, it was usually assumed that the observed data sequence is Markovian and the performance degradation caused by information aging was modeled as a monotonic AoI function. However, practical data sequence may not be Markovian \cite{guo2019credibility, SunNonlinear2019, wang2022framework}. In the present paper, theoretical results
and experimental studies are provided to analyze the performance of real-time supervised learning for both Markovian and non-Markovian time-series data. In \cite{zhang2020taming}, impact of peak-AoI on the convergence speed of online training was analyzed. Unlike online training in \cite{zhang2020taming}, our work considers offline training and online inference.

Moreover, there are significant research efforts on the optimization of AoI functions by designing sampling and scheduling policies. Previous studies \cite{sun2017update, SunNonlinear2019, orneeTON2021, bedewy2021optimal, Tripathi2019, Kadota2018} focused on non-decreasing AoI functions. Recently, a Whittle index based multi-source scheduling policy was derived in \cite{chen2021uncertainty} to minimize Shannon's conditional entropy that could be a non-monotonic function of the AoI. The Whittle index policy in \cite{chen2021uncertainty} requires that (i) the state of each source evolves as binary Markov process, (ii) the AoI function is concave with respect to the belief state of the Markov process, and (iii) the packet transmission time is constant. The results in \cite{SunNonlinear2019, orneeTON2021, sun2017update, bedewy2021optimal, Tripathi2019, Kadota2018, chen2021uncertainty} are not appropriate for minimizing  general (potentially non-monotonic) AoI functions, as considered in the present paper. 
\section{Information-theoretic Measures for Real-time Supervised Learning}
\subsection{Freshness-aware Learning Model}
Consider the real-time supervised learning system illustrated in Fig. \ref{fig:learning}, where the goal is to predict a label $Y_t \in \mathcal Y$ (e.g., the location of the car in front) at each time $t$ based on a feature $X_{t-\Delta(t)}$ (e.g., a video clip) that was generated $\Delta(t)$ seconds ago.~The feature, $X_{t-\Delta(t)}=(V_{t-\Delta(t)}, \ldots, V_{t-\Delta(t)-u+1})$ is a time sequence with length $u$ (e.g., each video clip consisting of $u$ consecutive video frames). We consider a class of popular supervised learning algorithms called \emph{Empirical Risk Minimization (ERM)} \cite{goodfellow2016deep}. \ignore{\cite{vapnik2013nature, goodfellow2016deep, mohri2018foundations}}In freshness-aware ERM algorithms, a neural network is trained to construct an action $a = \phi(X_{t-\Delta(t)},\Delta(t)) \in \mathcal A$ where $\phi: \mathcal X \times \mathcal D \mapsto \mathcal A$ is a function of  feature $X_{t-\Delta(t)}\in\mathcal X$ and its AoI $\Delta(t) \in \mathcal D$. The performance of learning is measured by a loss function $L: \mathcal Y \times \mathcal A \mapsto \mathbb R$, where $L(y,a)$ is the incurred loss if action $a$ is chosen by the neural network when $Y_t=y$. We assume that $\mathcal Y$, $\mathcal X$, and $\mathcal D$ are discrete and finite sets. The loss function $L$ is determined by the \emph{targeted application} of the system.~For example, in neural network based estimation, the loss function is usually chosen as the square estimation error $L_2(\mathbf y,\hat{\mathbf y}) =\|\mathbf y - \hat{\mathbf y}\|^2$, where the action $a = \hat y$ is an estimate of $Y_t = y$.  In softmax regression (i.e., neural network based maximum likelihood
classification), the action $a = Q_Y$ is a distribution of $Y_t$ and the loss function $L_{\text{log}}(y, Q_Y ) = - \text{log}~Q_Y (y)$ is the negative log-likelihood of the label value $Y_t = y$. Therefore, the loss function $L$ characterizes the goal and purpose of a specific application.  

\subsection{Offline Training Error}
The real-time supervised learning system that we consider consists of two phases: \emph{offline training} and \emph{online inference}.~In the offline training phase, the neural network is trained using a training dataset. Let $P_{\tilde Y_0, \tilde X_{- \Theta}, \Theta}$ denote the empirical distribution of the label $\tilde Y_0$, feature $\tilde X_{-\Theta}$, and AoI $\Theta$ in the training dataset, where the AoI $\Theta \geq 0$ of the feature $\tilde X_{-\Theta}$ is the time difference between $\tilde Y_0$ and $\tilde X_{-\Theta}$. In ERM algorithms, the  training problem is formulated as  
\begin{align}\label{eq_trainingerror}
\mathrm{err}_{\mathrm{training}} = \min_{\phi\in \Lambda} \mathbb{E}_{Y,X,\Theta\sim P_{\tilde Y_0, \tilde X_{-\Theta},\Theta}}[L(Y,\phi(X,\Theta))],
\end{align} 
where $\Lambda$ is the set of functions that can be constructed by the neural
network, and $\mathrm{err}_{\mathrm{training}}$ is the minimum training error. The optimal solution to \eqref{eq_trainingerror} is denoted by $\phi^*_{P_{\tilde Y_0, \tilde X_{-\Theta},\Theta}}$.

Let $\Phi=\{f : \mathcal X \times \mathcal D \mapsto \mathcal A\}$ be the set of all functions mapping
from $\mathcal X \times \mathcal D$ to $\mathcal A$. Any action $\phi(x,\theta)$ constructed by the neural network belongs to $\Phi$, whereas the neural network cannot produce some functions in $\Phi$. Hence, $\Lambda \subset \Phi$. By relaxing the feasible set $\Lambda$ in \eqref{eq_trainingerror} as $\Phi$, we obtain a lower bound of $\mathrm{err}_{\mathrm{training}}$, i.e.,
\begin{align}\label{eq_TrainingErrorLB}
H_L(\tilde Y_0| \tilde X_{-\Theta},\Theta)=\min_{\phi\in \Phi} \mathbb{E}_{Y,X,\Theta\sim P_{\tilde Y_0, \tilde X_{-\Theta},\Theta}}[L(Y,\phi(X,\Theta))], 
\end{align} 
where $H_L(\tilde Y_0| \tilde X_{-\Theta},\Theta)$ is a generalized conditional entropy of $\tilde Y_0$ given $(\tilde X_{-\Theta},\Theta)$ \cite{Dawid2004, Dawid1998,farnia2016minimax}. 
Compared to $\mathrm{err}_{\mathrm{training}}$, its information-theoretic lower bound $H_L(\tilde Y_0| \tilde X_{-\Theta},\Theta)$ is mathematically more convenient to analyze. The gap between $\mathrm{err}_{\mathrm{training}}$ and the lower bound $H_L(\tilde Y_0| \tilde X_{-\Theta},\Theta)$ was studied recently in \cite{shisher2022local}, where the gap is small if the function spaces $\Lambda$ and $\Phi$ are close to each other, e.g., when the neural network is sufficiently wide and deep \cite{goodfellow2016deep}. \ignore{The optimal solution to \eqref{eq_TrainingErrorLB} is denoted by $\hat \phi_{P_{\tilde Y_0, \tilde X_{-\Theta},\Theta}}$.} 
 
\ignore{\Phi$ is the set of all function mappings from $\mathcal X \times\mathcal D$ to $\mathcal A$. Because $\Lambda \subset \Phi, H_L(\tilde Y_0| \tilde X_{-\Theta},\Theta) \leq \mathrm{err}_{\mathrm{training}}$. For {\blue notational convenience, we refer to $H_L(\tilde Y_0| \tilde X_{-\Theta},\Theta)$ as an \emph{$L$-conditional entropy}}, because it is associated to a loss function $L$. By choosing different $L$, a broad class of $L$-conditional entropies is obtained.~In particular, Shannon's conditional entropy is derived if $L$ is the logarithmic loss function $L_{\log}(y,P_Y) = - \log P_Y(y)$. {\blue The optimal solution to \eqref{eq_TrainingErrorLB} is denoted as $\hat \phi_{P_{\tilde Y_t, \tilde X_{t-\Theta},\Theta}}$.}

{\blue Because} $\Phi$ contains all functions from $\mathcal{X}\times \mathcal D$ to $\mathcal{A}$, \eqref{eq_TrainingErrorLB} can be decomposed into a sequence of separated optimization problems, each optimizing an action $\phi(x, \theta)$ for given $(x, \theta) \in \mathcal X \times \mathcal D$ \cite{Dawid2004, farnia2016minimax}:
\begin{align} \label{L_condEntropy}
\!\!\!\!&H_L(\tilde Y_t| \tilde X_{t-\Theta},\Theta) \nonumber\\
\!\!\!\!=&\!\!\min_{\substack{\phi(x, \theta)\in \mathcal A
} } \!\!\sum_{\substack{x \in \mathcal X, \theta \in \mathcal D}} \!\! P_{\tilde X_{t-\Theta},\Theta}(x, \theta)\mathbb E_{Y\sim P_{\tilde Y_t|\tilde X_{t-\Theta}=x, \Theta=\theta}}[L(Y,\phi(x, \theta))] \nonumber\\
\!\!\!\!=&\!\!\!\sum_{\substack{x \in \mathcal X, \theta \in \mathcal D}} \!\!\!\! P_{X_{t-\Theta}, \Theta}(x, \theta) \!\!\min_{\phi(x, \theta)\in\mathcal A}\!\! E_{Y\sim P_{\tilde Y_t|\tilde X_{t-\Theta}=x, \Theta=\theta}}[L(Y,\phi(x, \theta))].\!\!\!\!
\end{align}
We note that problem \eqref{eq_trainingerror} cannot be decomposed in this way, because its function space $\Lambda$ is smaller than $\Phi$.} 

For notational convenience, we refer to $H_L(\tilde Y_0| \tilde X_{-\Theta},\Theta)$ as an \emph{L-conditional entropy}, because it is associated with a loss function $L$. The \emph{$L$-entropy} of a random variable $Y$ is defined as \cite{Dawid2004, farnia2016minimax}
\begin{align}\label{eq_Lentropy}
H_L(Y) = \min_{a\in\mathcal A} \mathbb{E}_{Y \sim P_{Y}}[L(Y,a)].
\end{align} 
Let $a_{P_Y}$ denote an optimal solution to \eqref{eq_Lentropy}, which is called a \emph{Bayes action} \cite{Dawid2004}.
The $L$-conditional entropy of $Y$ given $X=x$ is 
\begin{align}\label{given_L_condentropy}
H_L(Y| X=x)= \min_{a \in\mathcal A}\! \mathbb E_{Y\sim P_{Y| X=x}} [L(Y, a)].
\end{align}
Using \eqref{given_L_condentropy}, we can get the $L$-conditional entropy of $Y$ given $X$ \cite{Dawid2004, farnia2016minimax}
\begin{align}\label{eq_cond_entropy1}
H_L(Y|X)=\sum_{x \in \mathcal X} P_X(x) H_L(Y| X=x).
\end{align}
Similar to \eqref{eq_cond_entropy1}, \eqref{eq_TrainingErrorLB} can be decomposed as\ignore{into a sequence of separated optimization problems in \eqref{eq_TrainingErrorLB1}, each optimizing an action $\phi(x, \theta)$ given $(x, \theta)$:} 
\label{eq_TrainingErrorLB1}
\begin{align}\label{eq_TrainingErrorLB1}
\!\!\!&H_L(\tilde Y_0| \tilde X_{-\Theta},\Theta) \nonumber\\
=&\sum_{\substack{x \in \mathcal X, \theta \in \mathcal D}} \!\!\!\! P_{\tilde X_{-\Theta}, \Theta}(x, \theta) H_L(\tilde Y_0|\tilde X_{-\theta}=x, \Theta=\theta).\!\!\!
\end{align}
We assume that in the training dataset, the AoI $\Theta$ is independent of the label $\tilde Y_0$ and feature $\tilde X_{-\mu}$ for all $\mu\geq 0$. By this assumption and \eqref{eq_TrainingErrorLB1}, one can get \ifreport(see Appendix \ref{pfreshness_aware_cond} for its proof) \else (see  our technical report \cite{technical_report} for its proof)\fi
\begin{align}\label{freshness_aware_cond}
\!\! H_L(\tilde Y_0| \tilde X_{-\Theta},\Theta)\!=\sum_{\theta \in \mathcal D} P_{\Theta}(\theta)~H_L(\tilde Y_0| \tilde X_{-\theta}).
\end{align}
\ifreport
\else

By choosing different $L$ in \eqref{eq_Lentropy}, a broad class of $L$-entropies is obtained. In particular, Shannon’s entropy is derived if $L$ is the logarithmic loss function $L_{\log}(y,Q_{Y}) = - \log Q_{Y}(y)$. More examples of the loss function $L$, the definitions of $L$-divergence $D_L(P_Y || Q_Y)$, $L$-mutual information $I_L(Y; X)$, and $L$-conditional mutual information $I_L(Y; X| Z)$ are provided in \cite{technical_report}. In general, $I_L( X ; Y)  \neq I_L(Y ; X)$, which is different from $f$-mutual information. Moreover, a comparison among the $L$-divergence, Bregman divergence \cite{dhillon2008matrix}, and the $f$-divergence \cite{csiszar2004information} is provided in \cite{technical_report}.
\fi
\ignore{\begin{align}\label{freshness_aware_cond}
\!\! H_L(\tilde Y_t| \tilde X_{t-\Theta},\Theta)\!=&\!\!\sum_{\substack{x \in \mathcal X, \theta \in \mathcal D}} \!\!\!\! P_{X_{t-\Theta}}(x) P_{\Theta}(\theta) H_L(\tilde Y_t| \tilde X_{t-\Theta}=x,\Theta=\theta)\!\!\! \nonumber\\
=&\sum_{\theta \in \mathcal D} \!\!\!\! P_{\Theta}(\theta)~H_L(\tilde Y_t| \tilde X_{t-\theta}).\!\!\! 
\end{align}}

\ifreport
The $L$-\emph{divergence} $D_L(P_{Y} || P_{\tilde Y})$ of $P_{Y}$ from $P_{\tilde Y}$ can be expressed as \cite{Dawid2004, farnia2016minimax}
\begin{align}\label{divergence}
   \!\!\!\!\!\! D_L(P_{Y} || P_{\tilde Y})\!=\!\mathbb E_{Y \sim P_{\tilde Y}}\left[L\left(Y, a_{P_{Y}}\right)\right]-\mathbb E_{Y \sim P_{\tilde Y}}\left[L\left(Y, a_{P_{\tilde Y}}\right)\right] \geq 0.
\end{align}
The \emph{$L$-mutual information} $I_L(Y;X)$ is defined as \cite{Dawid2004, farnia2016minimax}
\begin{align}\label{MI}
I_L(Y; X)=& \mathbb E_{X \sim P_{X}}\left[D_L\left(P_{Y|X}||P_{Y}\right)\right]\nonumber\\
=&H_L(Y)-H_L(Y|X) \geq 0,
\end{align}
which measures the performance gain in predicting $Y$ by observing $X$. In general, $I_L(Y;X)$ $\neq$ $I_L(X;Y)$. The $L$-conditional mutual information $I_L(Y; X | Z)$ is given by 
\begin{align}\label{CMI}
I_L(Y; X|Z)=& \mathbb E_{X, Z \sim P_{X, Z}}\left[D_L\left(P_{Y|X, Z}||P_{Y | Z}\right)\right]\nonumber\\
=&H_L(Y | Z)-H_L(Y|X, Z).
\end{align}
\fi

\ifreport
The relationship among $L$-divergence, Bregman divergence \cite{dhillon2008matrix}, and $f$-divergence \cite{csiszar2004information} is discussed in Appendix \ref{InformationTheory2}. We note that any Bregman divergence is an $L$-divergence, and an $L$-divergence is a Bregman divergence only if $H_L(Y_t)$ is continuously differentiable and strictly concave in $\mathcal P_{Y_t}$ \cite{Dawid2004}. Examples of loss function $L$, $L$-entropy, and $L$-cross entropy are provided in Appendix \ref{InformationTheory1}.
\else
\ignore{The relationship among $L$-divergence, Bregman divergence \cite{Dawid2004, Amari}, and $f$-divergence \cite{csiszar2004information} is discussed in our technical report \cite{technical_report}. We note that any Bregman divergence is an $L$-divergence, and an $L$-divergence is a Bregman divergence only if $H_L(Y_t)$ is continuously differentiable and strictly concave in $\mathcal P_{Y_t}$ \cite{Dawid2004}. Examples of loss function $L$, $L$-entropy, and $L$-cross entropy are provided in \cite{technical_report}.}
\fi 

\subsection{Online Inference Error}
In the online inference phase, the neural predictor trained by \eqref{eq_trainingerror} is used to predict the target in real-time. We assume that $\{(Y_t , X_t ), t \in \mathbb Z\}$ is a stationary process that is independent of the AoI process $\{\Delta(t), t \in \mathbb Z\}$. Using this assumption, the time-average expected inference error during the time slots $t=0,1,\ldots, T-1$ is given by
\begin{align}\label{eq_inferenceerror}
\mathrm{err}_{\mathrm{inference}}(T)= \frac{1}{T} \mathbb E \left [ \sum_{t=0}^{T-1} p(\Delta(t))\right],
\end{align}
where \ignore{$p(\Delta(t))$ is the inference error in time slot $t$, $\Delta(t)$ is the inference AoI, and $p(\cdot)$ is defined by}
\begin{align}\label{instantaneous_err1} 
p(\delta)=\mathbb E_{Y, X \sim P_{Y_t, X_{t-\delta}}}\left[L\left(Y,\phi^*_{P_{\tilde Y_0, \tilde X_{-\Theta},\Theta}}(X,\delta)\right)\right],
\end{align}
$p(\Delta(t))$ is the expected inference error in time slot $t$, and $\Delta(t)$ is the inference AoI at time $t$, i.e., the time difference between label $Y_t$ and feature $X_{t-\Delta(t)}$. \ifreport The proof of \eqref{eq_inferenceerror} is provided in Appendix \ref{peq_inferenceerror}. \else The proof of \eqref{eq_inferenceerror} is provided in \cite{technical_report}. \fi
\ignore{where $P_{Y_t, X_{t- \delta}}$ is the empirical distribution of the label $Y_t$ and the feature $X_{t- \delta}$ and $\delta$ is the inference AoI.}
\ignore{\begin{align}\label{eq_inferenceerror}
\!\!\mathrm{err}_{\mathrm{inference}} \!\!=\! \frac{1}{T} \!\sum_{t=0}^{T-1} \!\mathbb E_{Y, X \sim P_{Y_t, X_{t-\Delta(t)}}}\!\!\!\left[\!L(Y,\phi^*_{P_{\tilde Y_t, \tilde X_{t-\Theta},\Theta}}\!(X,\Delta(t)))\right]\!,
\end{align}} 

Let us define \emph{$L$-cross entropy} between $Y$ and $\tilde Y$ as 
\begin{align} \label{cross-entropy}
H_L(Y; \tilde Y)= \mathbb{E}_{Y \sim P_{Y}}\left[L\left(Y, a_{P_{\tilde Y}}\right)\right],
\end{align}
and \emph{$L$-conditional cross entropy} between $Y$ and $\tilde Y$ given $X$ as
\begin{align} \label{cond-cross-entropy}
H_L(Y; \tilde Y | X)= \sum_{x \in \mathcal X} P_X(x) \mathbb{E}_{Y \sim P_{Y|X=x}}\left[L\left(Y, a_{P_{\tilde Y|\tilde X=x}}\right)\right],
\end{align}
where $a_{P_{\tilde Y}}$ and $a_{P_{\tilde Y|\tilde X=x}}$ are the Bayes actions associated with $P_{\tilde Y}$ and $P_{\tilde Y|\tilde X=x}$, respectively. If the neural predictor in \eqref{instantaneous_err1} is replaced by the Bayes action $a_{\tilde Y_0|\tilde X_{-\delta}=x}$, i.e., the optimal solution to \eqref{eq_TrainingErrorLB}, then $p(\delta)$ becomes an $L$-conditional cross entropy
\begin{align}\label{L-CondCrossEntropy}
\!\!\!\!\!\!&H_L(Y_{t}; \tilde Y_0 | X_{t-\delta})\!\! \nonumber\\
=&\!\!\!\sum_{x \in \mathcal X} \!\!P_{X_{t-\delta}}(x)\mathbb{E}_{Y \sim P_{Y_t| X_{t- \delta}=x}}\!\!\left[ L\left(Y,a_{\tilde Y_0|\tilde X_{-\delta}=x}\right)\right]\!.\!\!\!\!
\end{align} 
If the function spaces $\Lambda$ and $\Phi$ are close to each other, the difference between $p(\delta)$ and $H_L(Y_{t}; \tilde Y_0 | X_{t-\delta})$ is small.
\ignore{We assume that $\{(Y_t , X_t ), t \in \mathbb Z\}$ is a stationary process that is independent of $\{\Delta(t), t \in \mathbb Z\}$. If the neural predictor in \eqref{eq_inferenceerror} is replaced by the optimal solution to \eqref{eq_TrainingErrorLB1}, then $\mathrm{err}_{\mathrm{inference}}$ becomes an $L$-conditional cross-entropy
\begin{align}\label{L-CondCrossEntropy}
&H_L(Y_{t}; \tilde Y_{t} | X_{t-\Delta}, \Delta) \nonumber\\
=&\mathbb{E}_{Y,X,\Delta\sim P_{Y_t, X_{t- \Delta},\Delta}}\left[L\left(Y,\hat \phi_{P_{\tilde Y_t, \tilde X_{t-\Theta},\Theta}}(X,\Delta)\right)\right],
\end{align} 
where $\hat \phi_{P_{\tilde Y_t, \tilde X_{t-\Theta},\Theta}}$ is the optimal solution to \eqref{eq_TrainingErrorLB1} and $\Delta$ is a random variable that follows the empirical distribution of $\{\Delta(t), t \in \mathbb Z\}$.}
\ignore{If the function space $\Lambda$ is sufficiently large, the difference between $\mathrm{err}_{\mathrm{inference}}$ and $H_L(\tilde Y_{t}; Y_{t} | X_{t-\Delta}, \Delta)$ is small.} 
\ignore{\ifreport
Since $(Y_0, X_{\Theta})$ is independent of $\Theta$ and $\{(\tilde Y_t, \tilde X_t), t \in \mathbb Z\}$ is independent of $\Delta$, the $L$-conditional cross entropy in \eqref{L-CondCrossEntropy} can be decomposed as 
(See Appendix \ref{pDecomposed_Cross_entropy} for its proof)
\begin{align}\label{Decomposed_Cross_entropy}
H_L(\tilde Y_t;Y_t| \tilde X_{t-\Delta}, \Delta) = \sum_{\delta \in \mathcal D} P_\Delta(\delta)~H_L(\tilde Y_t;Y_t| \tilde X_{t-\delta}).
\end{align}
{\violet Similar to \eqref{divergence}, we can get
\begin{align}\label{L-ConCrossEntropy1}
\!\!\!\!\!H_L(\tilde Y_{t}; Y_{t} | \tilde X_{t-\Delta}, \Delta) =&H_L(\tilde Y_{t} | \tilde X_{t-\Delta}, \Delta)+\sum_{x \in \mathcal X, \delta \in \mathcal D} P_{\tilde X_{t- \Delta},\Delta}(x, \delta)\nonumber\\
&\!\!\times D_L\left(P_{\tilde Y_t | \tilde X_{t-\Delta}=x, \Delta=\delta} || P_{Y_t | X_{t-\Theta}=x, \Theta=\delta}\right).\!\!\!
\end{align}
Because the $L$-divergence in \eqref{L-ConCrossEntropy1} is non-negative, we have}
\begin{align}\label{lowerbound_inference}
H_L(\tilde Y_{t}; Y_{t} | \tilde X_{t-\Delta}, \Delta) \geq H_L(\tilde Y_{t} | \tilde X_{t-\Delta}, \Delta).
\end{align}
\fi}

\ignore{Consider the real-time forecasting system illustrated in Fig. \ref{fig:learning}, where the goal is to predict a fresh label $Y_t$ (e.g., location of the car in front) of time $t$ based on {\blue an earlier}  observation $X_{t-\Delta(t)}$ that was generated $\Delta(t)$ seconds ago.~The observation, a.k.a., feature, $X_{t-\Delta(t)}$$=$$(s_{t-\Delta(t)}$, $\ldots, s_{t-\Delta(t)-u+1})$ is a {\blue time sequence of length $u$ (e.g., $u$ consecutive video frames). We consider} a class of supervised learning algorithms called Empirical Risk Minimization (ERM), which is a standard approach for supervised learning \cite{vapnik2013nature, goodfellow2016deep, mohri2018foundations}. In ERM, the decision-maker predicts {\blue the label} $Y_t\in\mathcal Y$ by taking an action $a$ $=$ $\phi(X_{t-\Delta(t)},\Delta(t))$$\in \mathcal A$ based on  the observation $X_{t-\Delta(t)}\in\mathcal X$ and {\blue its AoI $\Delta(t) \in \mathcal D$, where we assume that $\mathcal Y$, $\mathcal X$, and $\mathcal D$ are discrete sets.}~The learning performance is measured by a loss function $L$, where $L(y,a)$ is the incurred loss if action $a$ is chosen when $Y_t=y$. The loss function $L$ is specified by the ERM supervised learning algorithm.~For example, {\blue $L$ is a quadratic function $L_2(y,\hat y) = (y - \hat y)^2$ in linear regression and a logarithmic function $L_{\log}(y,P_Y) = - \log P_Y(y)$ in logistic regression, where $P_Y$ denotes the distribution of $Y$.}

Supervised learning based real-time forecasting consists of two phases: \emph{offline training}
and \emph{online inference}.~In the offline training phase, a training dataset is collected and is used to train a neural network. Let $P_{Y_t, X_{t- \Theta}, \Theta}$ denote the empirical distribution of the training data $(Y_t,X_{t-\Theta})$ and training AoI $\Theta$, {\blue and $(Y_t,X_{t-\Theta},\Theta)$ are random variables following this empirical distribution. Here, the training AoI $\Theta \geq 0$ is the time difference between the observation $X_{t-\Theta}$ and the label $Y_t$. We assume that the training label and feature $(Y_t ,X_t )$ are stationary over time and are independent of the training AoI $\Theta$.} The objective of training in ERM-based real-time forecasting is to solve the following problem:  
\begin{align}\label{eq_trainingerror}
\mathrm{err}_{\mathrm{training}} = \min_{\phi\in \Lambda} \mathbb{E}_{Y,X,\Theta\sim P_{Y_t,X_{t-\Theta},\Theta}}[L(Y,\phi(X,\Theta))],
\end{align} 
where $\phi: \mathcal X \times \mathcal D \mapsto \mathcal A$ is selected from a family of decision functions $\Lambda$ that can be implemented by a neural network, $$\mathbb{E}_{Y,X,\Theta\sim P_{Y_t,X_{t-\Theta},\Theta}}[L(Y,\phi(X,\Theta))]$$ is {\blue the expected loss over the empirical distribution of training data and training AoI, and} $\mathrm{err}_{\mathrm{training}}$ is {\blue called} the \emph{minimum training error} \footnote{In this paper, we focus on \emph{freshness-aware inference}, in which the AoI is fed as an input of the neural predictor (see Fig. \ref{fig:learning}) and is used to predict the label. {\violet The complimentary case of \emph{freshness-agnostic inference}, where the AoI is unknown at the neural predictor, is out of the scope of this paper. The difference between \emph{freshness-aware inference} and \emph{freshness-agnostic inference} was briefly discussed in \cite{shisher2021age} and will be further studied in our future work.}}.~The optimal solution to \eqref{eq_trainingerror} is denoted by $\phi^*_{P_{Y_t, X_{t-\Theta},\Theta}}$. 

Within the {online inference} phase, {\blue the trained neural predictor $\phi^*_{P_{Y_t, X_{t-\Theta},\Theta}}$ is used to predict the target in real-time.} The \emph{inference error} is the expected loss on the inference data and inference AoI using the trained predictor, i.e.,
\begin{align}\label{eq_inferenceerror}
\mathrm{err}_{\mathrm{inference}} = \mathbb{E}_{Y,X,\Delta\sim P_{\tilde Y_t,\tilde X_{t- \Delta},\Delta}}[L(Y,\phi^*_{P_{Y_t, X_{t-\Theta},\Theta}}(X,\Delta))],
\end{align} 
where $P_{\tilde Y_t,\tilde X_{t- \Delta},\Delta}$ is the {\violet empirical distribution} of the inference data $(\tilde Y_t,\tilde X_{t- \Delta})$ and inference AoI $\Delta$, and $(\tilde Y_t,\tilde X_{t- \Delta},\Delta)$ are random variables following this distribution. 

During online inference, the observations (e.g., video frames for prediction) are sent to the trained neural predictor in real-time. As a result, the inference AoI is a random process $\{\Delta(t), t \in \mathbb Z \}$ governed by the communications from a data source (e.g., a sensor or camera) to the neural predictor, and $\Delta$ follows the empirical distribution of the AoI process $\{\Delta(t),t\in \mathbb Z\}$. {\violet We assume that the inference label and feature process $\{(\tilde Y_t , \tilde X_t), t\in \mathbb Z \}$ is stationary over time and is independent of the inference AoI $\Delta(t)$ and $\Delta$. On the other hand, the training AoI $\Theta$ (i.e., the time difference between  $X_{t-\Theta}$ and  $Y_t$) can be arbitrarily chosen ahead of time because the training dataset is prepared offline.} In Section \ref{InformationAnalysis}, we will discuss how to choose the training AoI $\Theta$. In Section \ref{Scheduling}, we will study how to optimally control the inference AoI process $\Delta(t)$ by scheduling the transmissions of {\violet observation features} to the neural predictor.}

\ifreport
\begin{figure*}[h]
  \centering
  \begin{subfigure}[t]{0.20\textwidth}
\includegraphics[width=\textwidth]{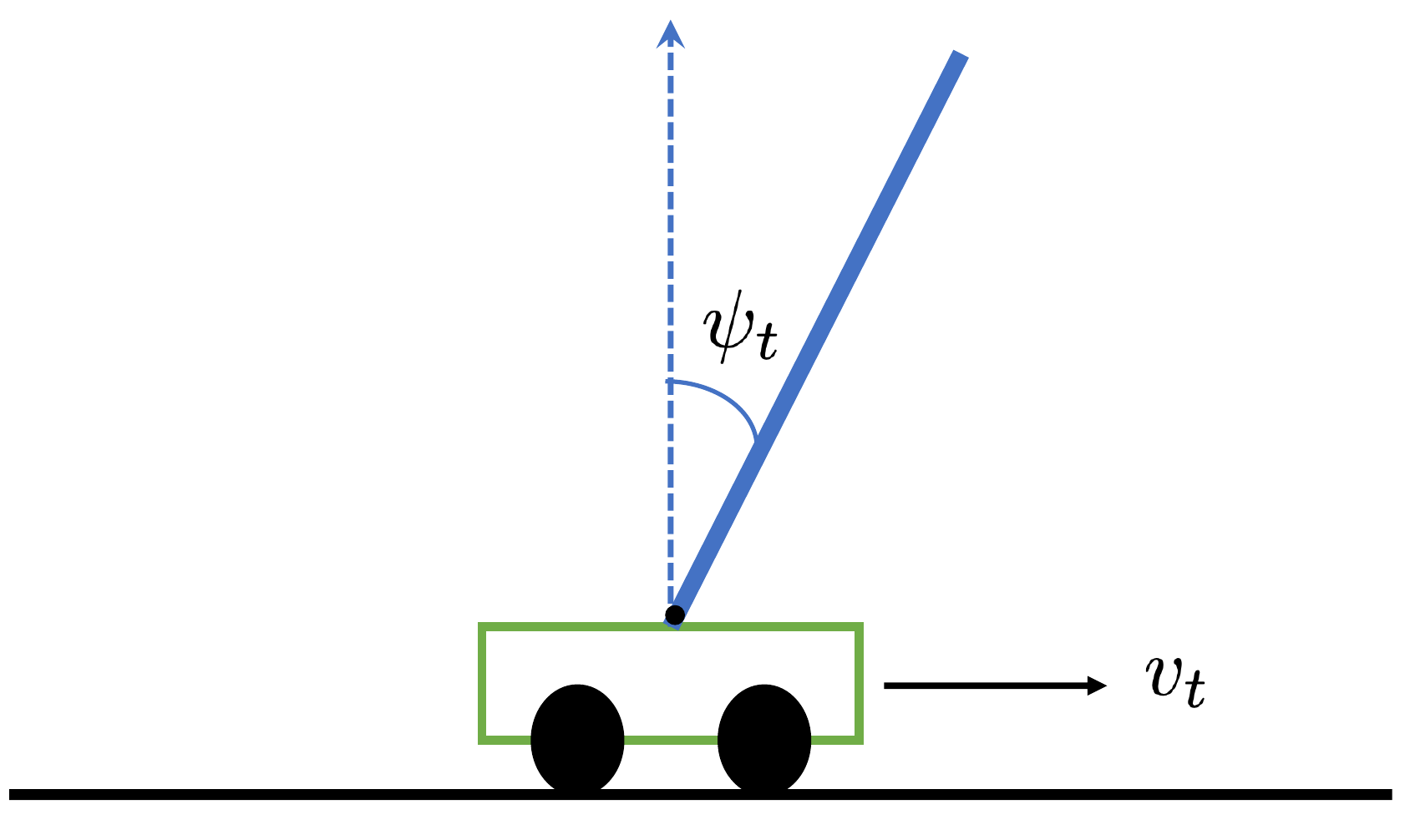}
  \subcaption{OpenAI Cart Pole  Task}
\end{subfigure}
%
\hspace{0mm}
\begin{subfigure}[t]{0.20\textwidth}
\includegraphics[width=\textwidth]{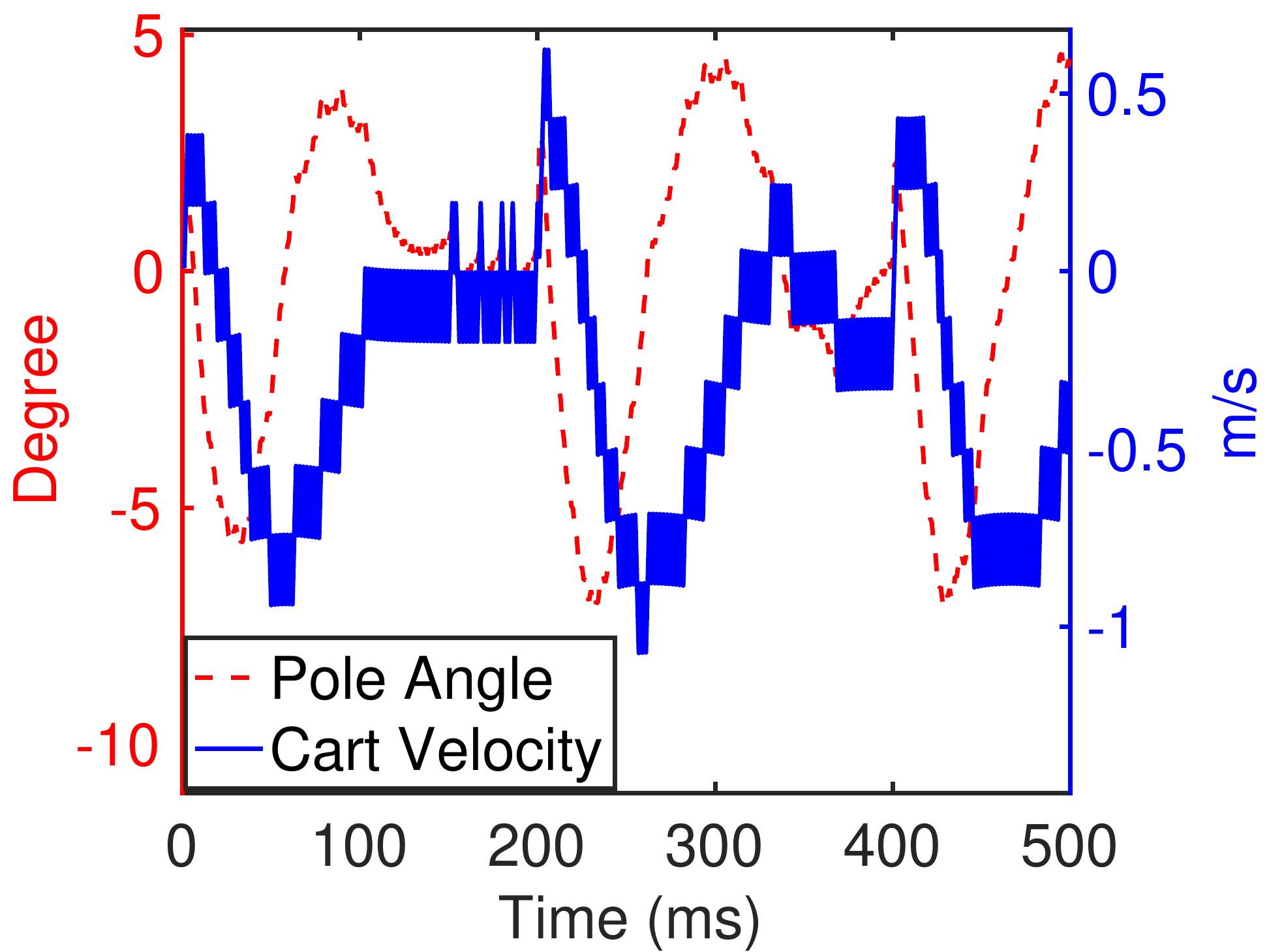}
  \subcaption{Data Traces}
\end{subfigure}
%
  \hspace*{0mm} 
\begin{subfigure}[t]{0.20\textwidth}
\includegraphics[width=\textwidth]{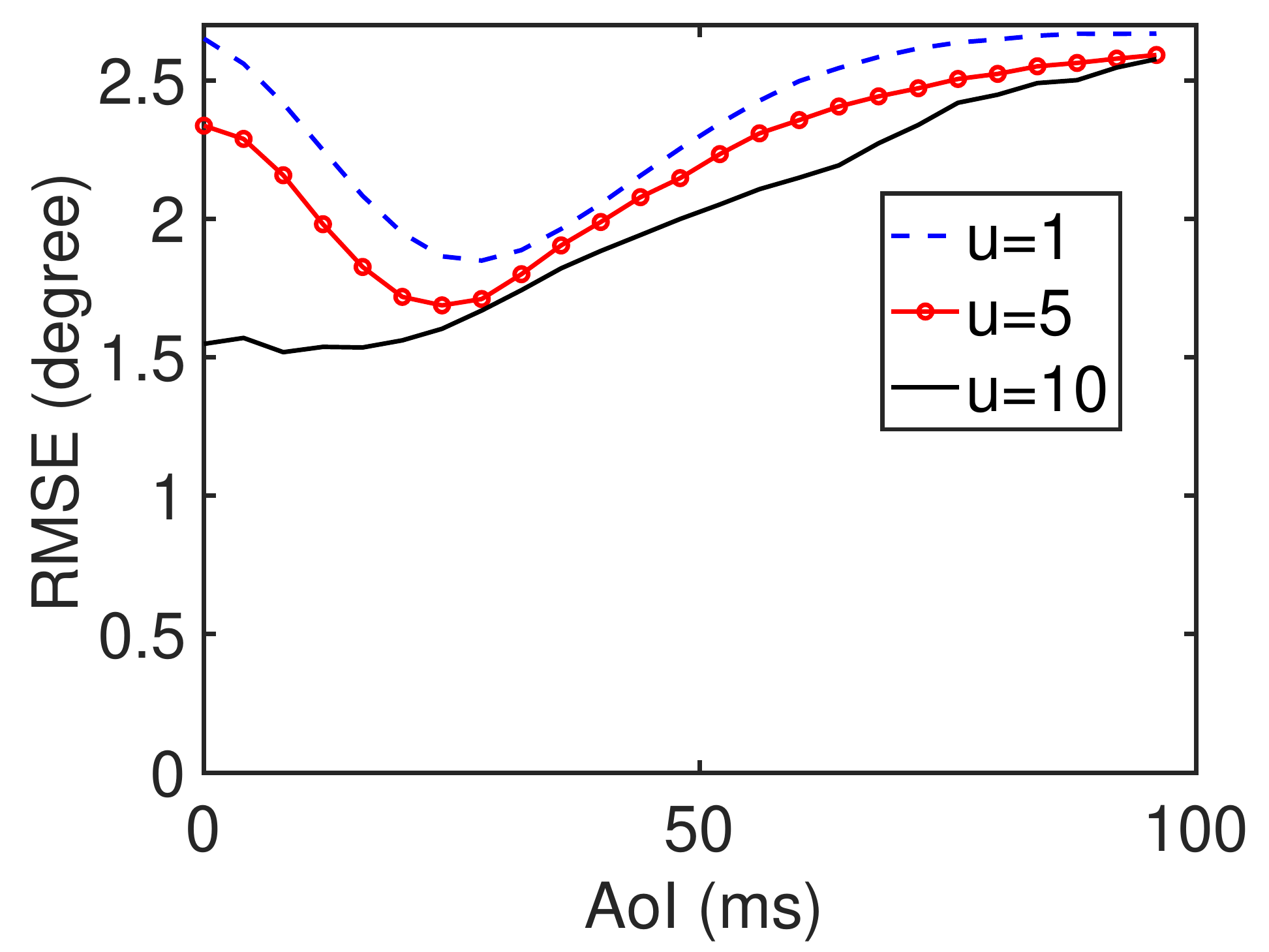}
  \subcaption{Training Error vs. AoI}
\end{subfigure}
%
\hspace{0mm}
\begin{subfigure}[t]{0.20\textwidth}
\includegraphics[width=\textwidth]{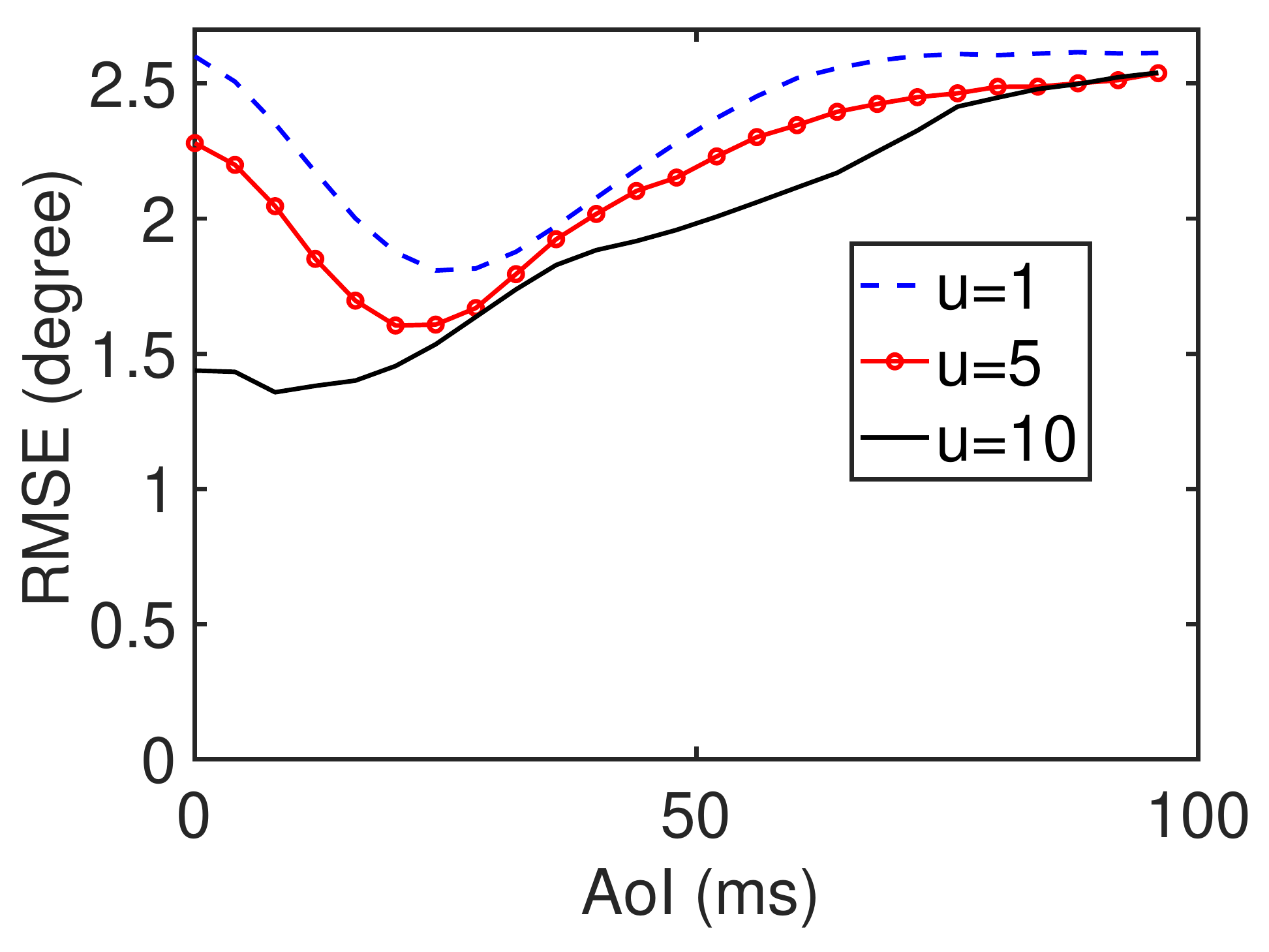}

  \subcaption{Inference Error vs. AoI}
\end{subfigure}
\caption{Performance of actuator state prediction under mechanical response delay. In the OpenAI CartPole-v1 task \cite{brockman2016openai}, the pole angle $\psi_t$ is predicted by using the cart velocity $v_{t-\delta}$ with an AoI $\delta$. Because of the mechanical response delay between cart velocity and pole angle, the training error and inference error are non-monotonic in the AoI.}
\label{fig:TrainingCartVelocity}
\end{figure*}
\else
\fi

\section{Interpretation of Freshness in Real-time Supervised Learning}\label{InformationAnalysis}
In this section, we study how the training AoI $\Theta$ and the inference
AoI $\Delta(t)$ affect the performance of real-time supervised learning.

\subsection{Training Error vs. Training AoI}\label{SecMinTrainingError}
We first consider the case of deterministic training AoI $\Theta=\theta$. Given $\Theta=\theta$, $H_L(\tilde Y_0| \tilde X_{-\Theta},\Theta)$ in \eqref{freshness_aware_cond} becomes simply $H_L(\tilde Y_0| \tilde X_{-\theta})$, which is a function of $\theta$. One may expect that the training error would grow with the AoI $\theta$. If $\tilde Y_0 \leftrightarrow \tilde X_{-\mu}  \leftrightarrow \tilde X_{-\mu-\nu}$ is a Markov chain for all $\mu,\nu\geq 0$, by the data processing inequality for $L$-conditional entropy \cite[Lemma 12.1] {Dawid1998}, one can show that $H_L(\tilde Y_0| \tilde X_{-\theta})$ is a non-decreasing function of $\theta$. \ifreport Nevertheless, the experimental results in Figs. \ref{fig:learning}-\ref{fig:Training} and \cite{shisher2021age} \else Nevertheless, the experimental results in Figs. \ref{fig:learning}-\ref{fig:DelayedNetworkedControlled} and \cite{technical_report, shisher2021age} \fi show that the training error is a growing function of the training AoI $\theta$ in some applications (e.g., video prediction), whereas it is a non-monotonic function of $\theta$ in other applications (e.g., temperature prediction and actuator state prediction with delay). As we will explain below, a fundamental reason behind these phenomena is that practical time-series data could be either Markovian or non-Markovian. For non-Markovian $(\tilde Y_0, \tilde X_{-\mu}, \tilde X_{-\mu -\nu})$, $H_L(\tilde Y_0| \tilde X_{-\theta})$ is not necessarily monotonic in $\theta$.  

Next, we develop an $\epsilon$-data processing inequality to analyze information freshness for both Markovian and non-Markovian time-series data. To that end, the following relaxation of the standard Markov chain model is needed, which is motivated by \cite{huang2019universal}: 

\ignore{{\blue Let us first consider the case of deterministic training AoI $\Theta=\theta$. In this case, $H_L(Y_t| X_{t-\Theta},\Theta=\theta)$ can be simply written as $H_L(Y_t| X_{t-\theta})$. Because $\{(Y_t, X_t)\}_{t \in \mathbb Z}$ is stationary, $H_L(Y_t| X_{t-\theta})$ is a function of $\theta$.} If $Y_t \leftrightarrow X_{t-\mu}  \leftrightarrow X_{t-\mu-\nu}$ is a Markov chain for all $\mu,\nu\geq 0$, then the data processing inequality \cite[Lemma 12.1] {Dawid1998} implies that $H_L(Y_{t} | X_{t-\theta})$ is a {\blue non-decreasing function of $\theta$.} However, our experimental results in Fig. \ref{fig:Training} show that {\blue the training error $\mathrm{err}_{\mathrm{training}}$ is not always monotonic in the training AoI $\theta$. This implies that the training data may not satisfy the Markov property. In fact, practical time-series data is usually non-Markovian \cite{Kampen1998Non-Markov, hanggi1977time, guo2019credibility, wang2021framework}, which hinders the use of data processing inequality.} Hence, novel analytical tools for interpreting information freshness in non-Markovian models are in great need. {\blue To resolve this challenge, we propose a new $\epsilon$-\emph{Markov chain model} that generalizes the standard Markov chain, and develop an $\epsilon$-\emph{data processing inequality} to characterize the relationship between training/inference errors and AoI.}}

\ignore{\subsubsection{$\epsilon$-Markov Chain Model and $\epsilon$-Data Processing Inequality}\label{Def_eMarkov}
We develop a unified framework {\blue to analyze information freshness for both} Markovian and non-Markovian time-series data. Towards that end, {\blue we introduce the following relaxation of the standard Markov chain model: }}

\begin{definition}[\textbf{$\epsilon$-Markov Chain}]
Given $\epsilon \geq 0$, a sequence of three random variables $Z, X,$ and $Y$ is said to be an \emph{$\epsilon$-Markov chain}, denoted as $Z \overset{\epsilon} \rightarrow X \overset{\epsilon} \rightarrow Y$, if
\begin{align}\label{epsilon-Markov-def}
I_{\chi^2}(Y;Z|X)=\mathbb E_{X, Z \sim P_{X, Z}} \left [ D_{\chi^2}\left(P_{Y|X,Z} || P_{Y|X} \right)\right] \leq \epsilon^2,
\end{align}
where\ifreport \footnote{In \eqref{epsilon-Markov-def}, if $P_{Y|X=x}(y) = 0$, then $P_{Y|X=x,Z=z}(y) = 0$ which leads to a term $\frac{0^2}{0}$ in the $\chi^2$-divergence $D_{\chi^2} (P_{Y|X=x,Z=z} || P_{Y|X=x})$. We adopt the convention in information theory \cite{polyanskiy2014lecture} to define $\frac{0^2}{0}= \lim_{a\rightarrow 0^+,b\rightarrow 0^+} \frac{(a-b)^2}{b} = 0$.}\else \fi
\begin{align}\label{chi-divergence-def}
D_{\chi^2}(P_Y ||Q_Y)=\sum_{y \in \mathcal{Y}} \frac{(P_Y(y) - Q_Y(y))^2}{Q_Y(y)}
\end{align}
is Neyman's $\chi^2$-divergence and $I_{\chi^2}(Y;Z|X)$ is $\chi^2$-conditional mutual information.
\end{definition}
A Markov chain is an $\epsilon$-Markov chain with $\epsilon= 0$. If $Z \rightarrow X \rightarrow Y$ is a Markov chain, then $Y \rightarrow X \rightarrow Z$ is also a Markov chain \cite[p. 34]{cover1999elements}. A similar property holds for the $\epsilon$-Markov chain.
\begin{lemma}\label{Symmetric}
 If $Z \overset{\epsilon}  \rightarrow X \overset{\epsilon} \rightarrow Y$, then $Y \overset{\epsilon}  \rightarrow X \overset{\epsilon} \rightarrow Z$.
\end{lemma}
\ifreport
\begin{proof}
See Appendix \ref{PSymmetric}.
\end{proof}
\else
Due to space limitation, all the proofs are relegated to our technical report \cite{technical_report}.
\fi
By Lemma \ref{Symmetric}, the $\epsilon$-Markov chain can be denoted as $Y \overset{\epsilon} \leftrightarrow X \overset{\epsilon} \leftrightarrow Z$.
In the following lemma, we provide a relaxation of the data processing inequality for $\epsilon$-Markov chain, which is called an \emph{$\epsilon$-data processing inequality}.
\begin{lemma}[\textbf{$\epsilon$-data processing inequality}] \label{Lemma_CMI}
If $Y \overset{\epsilon}\leftrightarrow X \overset{\epsilon}\leftrightarrow Z$ is an $\epsilon$-Markov chain, then 
\begin{align}
H_L(Y|X) \leq H_L(Y|Z)+O(\epsilon).
\end{align}
If, in addition, $H_L(Y)$ is twice differentiable in $P_Y$, then
\begin{align}
H_L(Y|X) \leq H_L(Y|Z)+O(\epsilon^2).
\end{align}
\end{lemma}
\ifreport
\begin{proof}
Lemma \ref{Lemma_CMI} is proven by using a local information geometric analysis; see Appendix \ref{PLemma_CMI} for the details.
\end{proof}
\else
\fi
Lemma \ref{Lemma_CMI}(b) was mentioned in \cite{shisher2021age} without proof. Lemma \ref{Lemma_CMI}(a) is new to the best of our knowledge. Now, we are ready to characterize how $H(\tilde Y_0 | \tilde X_{-\theta})$ varies with the AoI $\theta$.

\begin{theorem}\label{theorem1}
The $L$-conditional entropy
\begin{align}\label{eMarkov}
H_L(\tilde Y_0|\tilde X_{-\theta})= g_1(\theta)-g_2(\theta)
\end{align}
is a function of $\theta$, where $g_1(\theta)$ and $g_2(\theta)$ are two non-decreasing functions of $\theta$, given by
\begin{align}\label{g12function}
\!\!g_1(\theta)=&H_L(\tilde Y_0 | \tilde X_0) + \sum_{k=0}^{\theta-1}~I_L(\tilde Y_0; \tilde X_{-k}  | \tilde X_{-k-1}),~\nonumber\\
g_2(\theta)=&\sum_{k=0}^{\theta-1} I_L(\tilde Y_0; \tilde X_{-k-1} | \tilde X_{-k}).\!\!\!
\end{align}
If $\tilde Y_0 \overset{\epsilon}\leftrightarrow \tilde X_{-\mu} \overset{\epsilon}\leftrightarrow \tilde X_{-\mu-\nu}$ is an $\epsilon$-Markov chain for every $\mu, \nu \geq 0$, then $g_2(\theta) = O(\epsilon)$ and 
\begin{align}\label{eMarkov1}
H_L(\tilde Y_{0}|\tilde X_{-\theta})= g_1(\theta)+O(\epsilon).
\end{align}
\ignore{{\blue If, in addition,} $H_L(Y_0)$ is twice differentiable in $P_{Y_0}$, then $g_2(\theta) = O(\epsilon^2)$ and
\begin{align}\label{eMarkov2}
H_L(Y_0|X_{-\theta})= g_1(\theta)+O(\epsilon^2).
\end{align}}
\end{theorem}
\ifreport
\begin{proof}
See Appendix \ref{Ptheorem1}.
\end{proof}
\else
\fi

According to Theorem \ref{theorem1}, the monotonicity of $H_L(\tilde Y_0|\tilde X_{-\theta})$ in $\theta$ is characterized by the parameter $\epsilon \geq 0$ in the $\epsilon$-Markov chain model. If $\epsilon$ is small, then $\tilde Y_0 \overset{\epsilon}\leftrightarrow \tilde X_{-\mu} \overset{\epsilon}\leftrightarrow \tilde X_{-\mu-\nu}$ is close to a Markov chain, and $H_L(\tilde Y_0|\tilde X_{-\theta})$ is nearly non-decreasing in $\theta$. If $\epsilon$ is large, then $\tilde Y_0 \overset{\epsilon}\leftrightarrow \tilde X_{-\mu} \overset{\epsilon}\leftrightarrow \tilde X_{-\mu-\nu}$ is far from a Markov chain, and $H_L(\tilde Y_0|\tilde X_{-\theta})$ could be non-monotonic in $\theta$. Theorem  \ref{theorem1} can be readily extended to random AoI $\Theta$ by using stochastic orders \cite{stochasticOrder}.

\begin{figure*}[ht]
  \centering
  \begin{subfigure}[t]{0.25\textwidth}
\includegraphics[width=\textwidth]{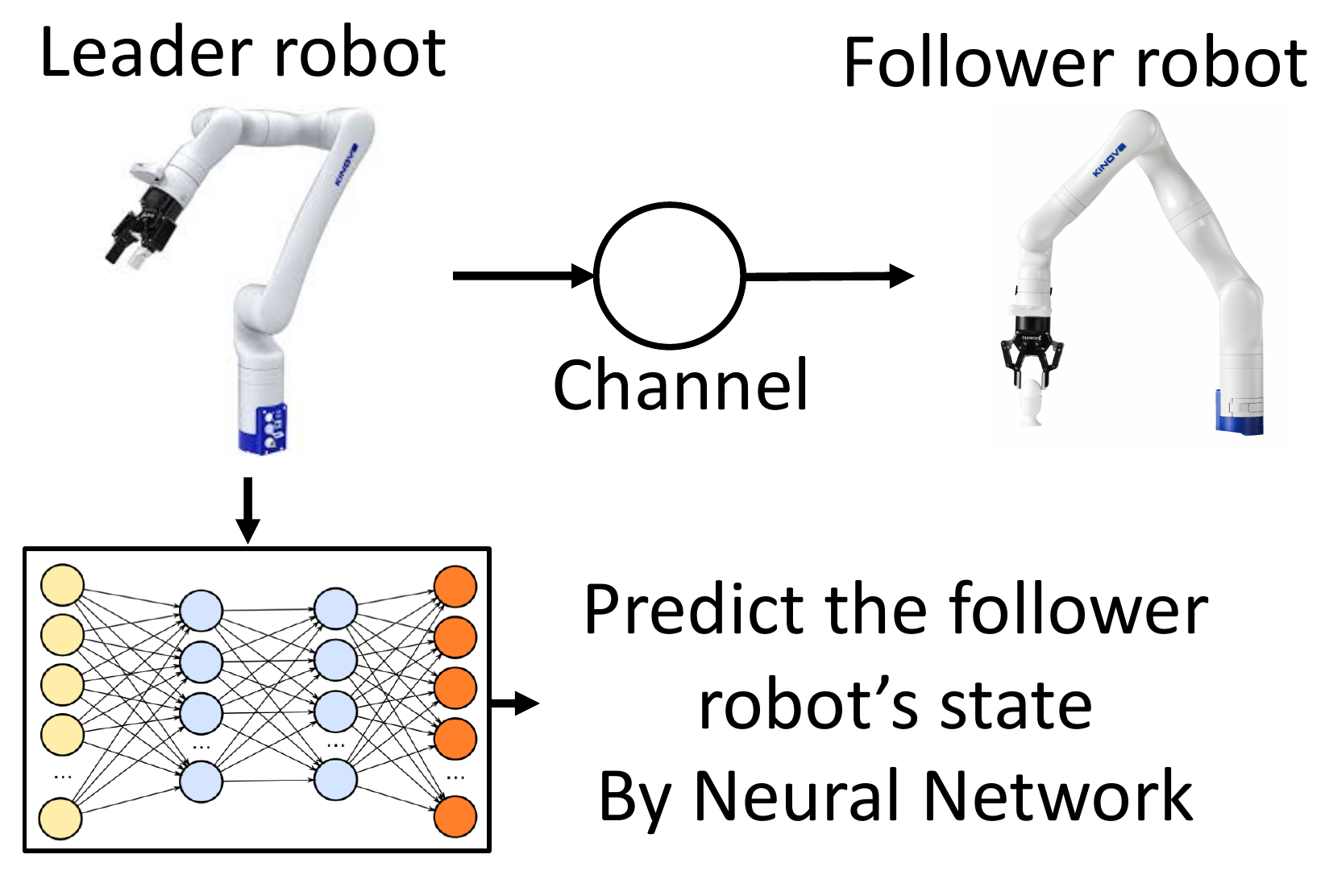}
  \subcaption{Prediction of Follower Robot}
\end{subfigure}
  \hspace{3mm}
\begin{subfigure}[t]{0.20\textwidth}
\includegraphics[width=\textwidth]{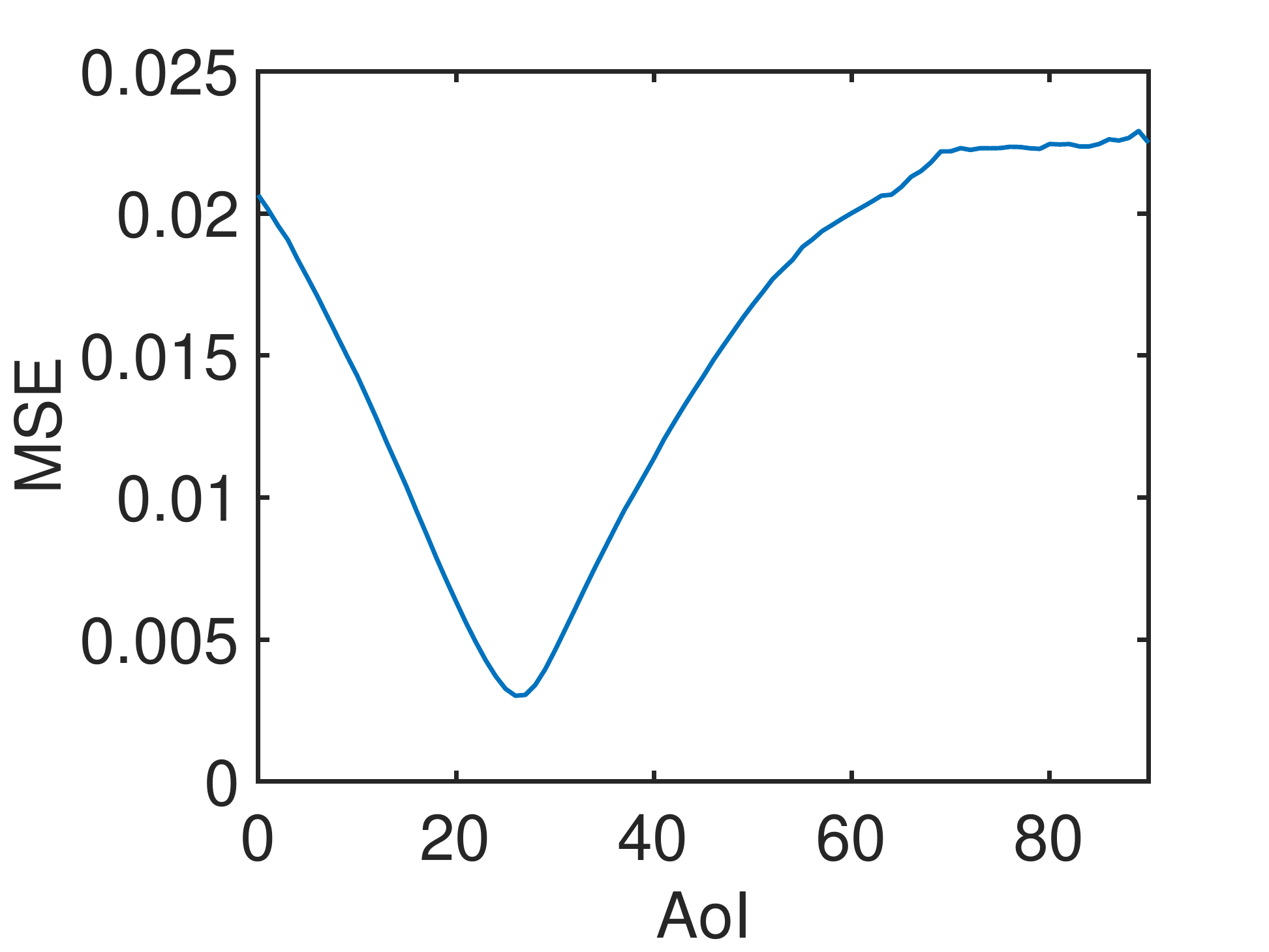}
  \subcaption{Training Error vs. AoI}
\end{subfigure}
  \hspace{3mm}
\begin{subfigure}[t]{0.20\textwidth}
\includegraphics[width=\textwidth]{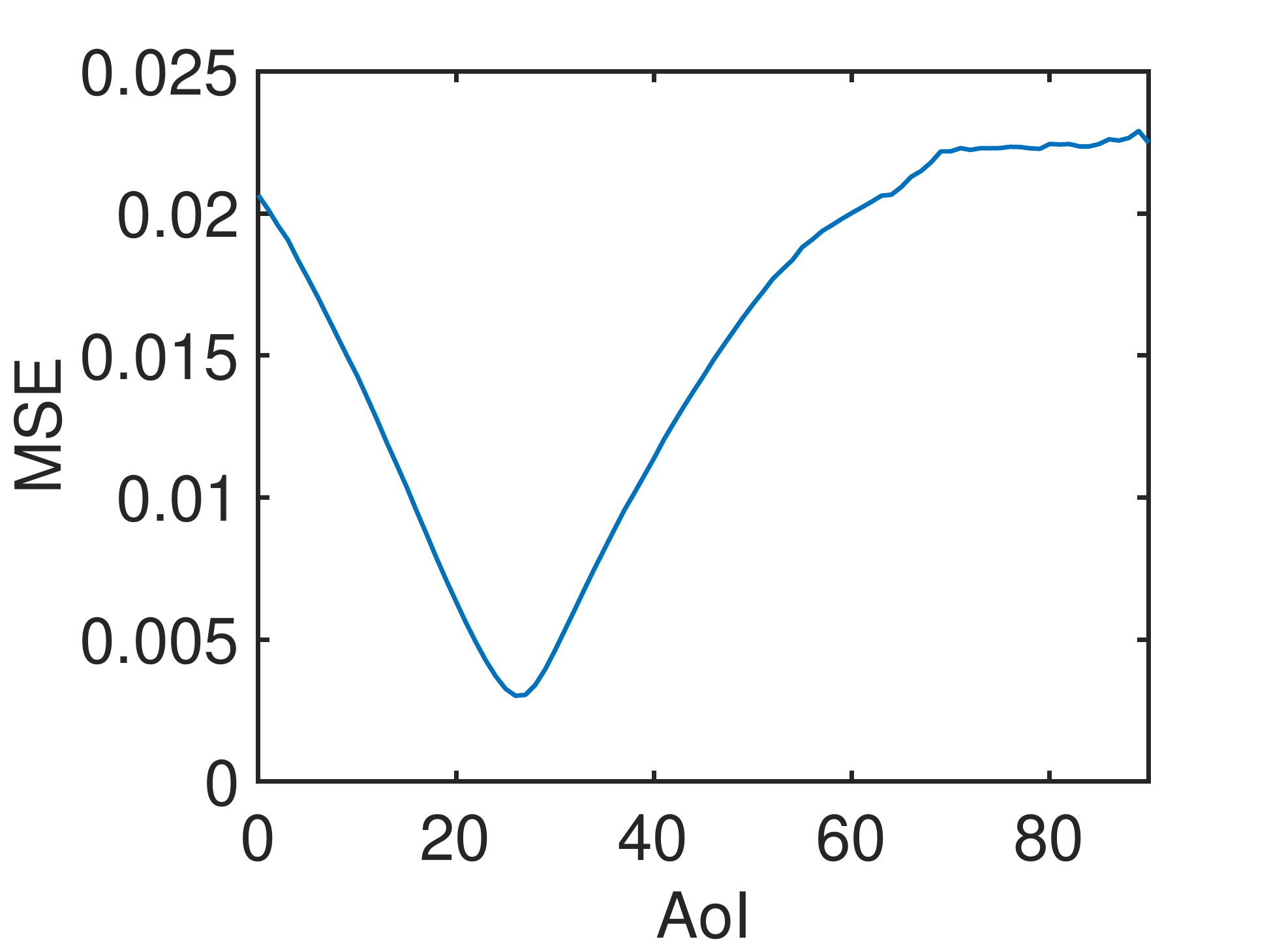}
  \subcaption{Inference Error vs. AoI}
\end{subfigure}
\caption{Robot state prediction in a leader-follower robotic system. The leader robot uses a neural network to predict the follower robot's state. The training and inference errors decrease in the AoI $\leq 25$ and increase when AoI $\geq 25$.}
\label{fig:DelayedNetworkedControlled}
\end{figure*}

\ignore{Theorem \ref{theorem1} {\blue tells us that} $H_L(Y_{t}|X_{t-\theta})$ is a function of the training AoI $\theta$, which is not necessarily monotonic.~The monotonicity of  $H_L(Y_{t}|X_{t-\theta})$ {\blue versus $\theta$} is characterized by the parameter $\epsilon \geq 0$ in the $\epsilon$-Markov chain model.~If $\epsilon$ is small, then the time-series data is close to a Markov chain and  $H_L(Y_{t}|X_{t-\theta})$ is nearly non-decreasing in the training AoI $\theta$. {\blue If $\epsilon$ is large, then the time-series data is far from a Markov chain and $H_L(Y_{t}|X_{t-\theta})$ is non-monotonic in $\theta$. ~Theorem \ref{theorem1} provides an interpretation of the} experimental results in Figure \ref{fig:Training}:~Recall that $u$ is the length of the observation sequence $X_{t-\theta}=(s_{t-\theta}, s_{t-\theta-1},\ldots, s_{t-\theta-u+1})$.~According to Shannon's interpretation of Markov sources in his seminal work \cite{Shannon1948}, the larger $u$, {\blue the closer $(Y_t, X_{t-\mu}, X_{t-\mu-\nu})$ tends to a Markov chain. The training error may be non-monotonic in the AoI $\theta$ for small $u$, but it will progressively become a growing function of the AoI as $u$ increases, which agrees with the results in Figure \ref{fig:Training}.}

\subsubsection{Training Error under Random Training AoI}
{\blue Theorem \ref{theorem1} can be extended to cover random training AoI $\Theta$ by using the stochastic ordering techniques \cite{stochasticOrder}. }}
\begin{definition}[\textbf{Univariate Stochastic Ordering}]\cite{stochasticOrder} 
A random variable $X$ is said to be stochastically smaller than another random variable $Z$, denoted as $X \leq_{st} Z$, if
\begin{align}
    P(X>x) \leq P(Z>x), \ \ \forall x \in \mathbb R.
\end{align}
\end{definition}
\begin{theorem}\label{theorem2}
If $\tilde Y_0 \overset{\epsilon}\leftrightarrow \tilde X_{-\mu} \overset{\epsilon}\leftrightarrow \tilde X_{-\mu-\nu}$ is an $\epsilon$-Markov chain for all $\mu, \nu \geq 0$, and the training AoIs in two experiments $1$ and $2$ satisfy $\Theta_{1} \leq_{st} \Theta_{2}$, then  
\begin{align}\label{dynamicsoln}
    H_L(\tilde Y_0|\tilde X_{-\Theta_1}, \Theta_1) \leq H_L(\tilde Y_0|\tilde X_{-\Theta_2}, \Theta_2)+O(\epsilon).
\end{align}
\ignore{\violet provided that the $L$-conditional entropies in \eqref{dynamicsoln} are finite. 
\item[(b)] If, in addition, $H_L(Y_t)$ is twice differentiable in $P_{Y_t}$, then
\begin{align}\label{dynamic_eqn}
    H_L(Y_{t}|X_{t-\Theta_1}, \Theta_1) \leq H_L(Y_{t}|X_{t-\Theta_2}, \Theta_2)+O(\epsilon^2).
\end{align}}
\end{theorem}
\ifreport
\begin{proof}
See Appendix \ref{Ptheorem2}.
\end{proof}
\else
\fi
According to Theorem \ref{theorem2}, if $\Theta_{1}$ is stochastically smaller than $\Theta_{2}$, then the training error in Experiment 1 is approximately
smaller than that in Experiment 2. If, in addition to the conditions in Theorems 3.4 and 3.6, $H_L(\tilde Y_0)$ is twice differentiable in $P_{\tilde Y_0}$, then the last term $O(\epsilon)$ in \eqref{eMarkov1} and \eqref{dynamicsoln} becomes $O(\epsilon^2)$. 

\subsection{Inference Error vs. Inference AoI}\label{SecInferenceError}
\ignore{{\blue Next, we analyze the relationship between inference error and inference AoI.}
\subsubsection{Inference Error under Deterministic Inference AoI}}
According to \eqref{given_L_condentropy}, \eqref{eq_cond_entropy1}, and \eqref{cond-cross-entropy}, $H_L(Y_{t}; \tilde Y_0 | X_{t-\delta})$ is lower bounded by $H_L(Y_{t} | X_{t-\delta})$. In addition, $H_L(Y_{t}; \tilde Y_0 | X_{t-\delta})$ is close to its lower bound $H_L(Y_{t} | X_{t-\delta})$, if the conditional distributions $P_{Y_t|X_{t-\delta}}$ and $P_{\tilde Y_0|\tilde X_{-\delta}}$ are close to each other, as shown by the following lemma. 

\ignore{{\blue In the case of deterministic inference AoI $\Delta=\delta$,  \ignore{the inference error $\mathrm{err}_{\mathrm{inference}}$ in \eqref{eq_inferenceerror} becomes 
$$\mathbb E_{Y, X \sim P_{\tilde Y_t, \tilde X_{t-\delta}}}\left[L\left(Y,\phi^*_{P_{Y_t, X_{t-\Theta},\Theta}}(X,\delta)\right)\right],$$ 
which is a function of inference AoI $\delta$ since $\{(\tilde Y_t, \tilde X_t),t \in \mathbb Z\}$ are assumed to be a stationary process. Moreover,} the $L$-conditional cross entropy $H_L(\tilde Y_{t}; Y_{t} | \tilde X_{t-\Delta}, \Delta=\delta)$ can be simplified as $H_L(\tilde Y_{t}; Y_{t} | \tilde X_{t-\delta})$. According to  \eqref{lowerbound_inference}, $H_L(\tilde Y_{t}; Y_{t} | \tilde X_{t-\delta})$ is lower bounded by} $H_L(\tilde Y_{t} | \tilde X_{t-\delta})$.
{\blue If the conditional distribution $P_{\tilde Y_t| \tilde X_{t-\delta}}$ of the inference data is close to the conditional distribution $P_{ Y_t| X_{t-\delta}}$ of the training data, then $H_L(\tilde Y_{t}; Y_{t} | \tilde X_{t-\delta})$ is close to its lower bound $H_L(\tilde Y_{t} | \tilde X_{t-\delta})$. This result is asserted in the following lemma. }}
\begin{lemma}\label{lemma_inference}
If for all $x \in \mathcal X$
    \begin{align}\label{T3condition2}
    D_{\chi^2}\left(P_{Y_{t}|X_{t-\delta}=x} || P_{\tilde Y_0| \tilde X_{-\delta}=x}\right) \!\! \leq \beta^2,
    \end{align}
    then  
    \begin{align}\label{Eq_Theorem3a}
        H_L(Y_{t}; \tilde Y_0 | X_{-\delta})=&H_L(Y_{t} | X_{t-\delta})+O(\beta).
    \end{align}
\end{lemma}
\ifreport
\begin{proof}
See Appendix \ref{Plemma_inference}.
\end{proof}
If \eqref{T3condition2} is replaced by the condition
 \begin{align}\label{T3condition2_diff}
    \sum_{x \in \mathcal X} P_{X_{t-\delta}}(x)~~D_{\chi^2}\left(P_{Y_{t}| X_{t-\delta}=x} || P_{\tilde Y_0| \tilde X_{-\delta}=x}\right) \!\!~\leq \beta^2,
    \end{align}
then Lemma \ref{lemma_inference} still holds.
\else
\fi 
By combining Theorem \ref{theorem1} and Lemma \ref{lemma_inference}, the monotonicity of $H_L(Y_{t}; \tilde Y_0 | X_{t-\delta})$ versus $\delta$ is characterized in the next theorem.
\begin{theorem}\label{theorem3}
The following assertions are true:
\begin{itemize}
\item[(a)] If $\{(Y_t, X_t),t \in \mathbb Z\}$ is a stationary process, then $H_L(Y_{t}; \tilde Y_0 | X_{t-\delta})$ is a function of the inference AoI $\delta$. 
\item[(b)] If, in addition, $Y_t \overset{\epsilon}\leftrightarrow X_{t-\mu} \overset{\epsilon}\leftrightarrow X_{t-\mu-\nu}$ is an $\epsilon$-Markov chain for all $\mu, \nu \geq 0$ and \eqref{T3condition2} holds for all $x \in \mathcal X$ and $\delta \in \mathcal D$, then for all $0 \leq \delta_1\leq \delta_2$
\begin{align}\label{eq_theorem3}
 \!\!H_L(Y_{t}; \tilde Y_0 | X_{t-\delta_1}) \leq H_L(Y_{t}; \tilde Y_0 | X_{t-\delta_2})+O\big(\max\{\epsilon, \beta\}\big).
\end{align}
 \end{itemize}
\end{theorem}
 \ifreport
\begin{proof}
See Appendix \ref{Ptheorem3}.
\end{proof}
\else
\fi
According to Theorem \ref{theorem3}, $H_L(Y_{t}; \tilde Y_0 | X_{t-\delta})$ is a function of the inference AoI $\delta$. If $\epsilon$ and $\beta$ are close to zero, $H_L(Y_{t}; \tilde Y_0 | X_{t-\delta})$ is nearly a non-decreasing function of $\delta$; otherwise, $H_L(Y_{t}; \tilde Y_0 | X_{t-\delta})$ can be far from a monotonic function of $\delta$.  

\ifreport
The $\epsilon$-Markov chain model that we propose can be viewed as a measure of conditional dependence. Different from earlier studies on conditional dependence measures \cite{fukumizu2007kernel, azadkia2019simple, reddi2013scale, joe1989relative}, we use a local information geometric approach to characterize how the non-Markov property of the data affects the relationship between AoI and the performance of real-time forecasting.
\else
\fi

\subsection{Interpretation of Experimental Results}\label{Experimentation}
We  conduct several experiments to study how the training and inference errors of real-time supervised learning vary with the AoI. The code of these experiments is provided in an open-source Github repository.\footnote{\url{https://github.com/Kamran0153/Impact-of-Data-Freshness-in-Learning}.}

Fig. \ref{fig:learning} illustrates the experimental results of supervised learning based video prediction, which are regenerated from \cite{lee2018stochastic}. In this experiment, the video   frame $V_t$ at time $t$ is predicted based on a feature $X_{t-\delta}=(V_{t-\delta}, V_{t-\delta-1})$ that is composed of two consecutive video frames, where $\Delta(t)=\delta$ is the AoI. A pre-trained neural network model called ``SAVP" \cite{lee2018stochastic} is used to evaluate on $256$ samples of ``BAIR" dataset \cite{ebert17sna}, which contains video frames of a randomly moving robotic arm. The pre-trained neural network model can be downloaded from the Github repository of \cite{lee2018stochastic}. One can observe from Fig. \ref{fig:learning}(b)-(c) that the training and inference errors are non-decreasing functions of the AoI, because the video clips $V_t$ are approximately a Markov chain.

\ifreport
Fig. \ref{fig:TrainingCartVelocity} plots the performance of actuator state prediction under mechanical response delay. We consider the OpenAI CartPole-v1 task \cite{brockman2016openai}, where a DQN reinforcement learning algorithm \cite{mnih2015human} is used to control the force on a cart and keep the pole attached to the cart from falling over. By simulating $10^4$ episodes of the OpenAI CartPole-v1 environment, a time-series dataset is collected that contains the pole angle $\psi_t$ and the velocity $V_{t}$ of the cart. The pole angle $\psi_t$ at time $t$ is predicted based on a feature $X_{t-\delta}=(V_{t-\delta}, \ldots, V_{t-\delta-u+1})$, i.e., a vector of cart velocity with length $u$, where $V_t$ is the cart velocity at time $t$ and $\Delta(t)=\delta$ is the AoI. The predictor in this experiment is an LSTM neural network that consists of one input layer, one hidden layer with 64 LSTM cells, and a fully connected output layer. First $72\%$ of the dataset is used for training and the rest of the dataset is used for inference. From the data trace in Fig. \ref{fig:TrainingCartVelocity}(b), one can observe a response (or reaction) delay of $25$-$30$ ms between cart velocity and pole angle. Such response delay exists broadly in  mechanical, circuit, biological, economic, and physical systems that are modeled by differential equations. Due to the response delay, $\psi_t$ is strongly correlated with $V_{t-25}$, but quite different from $V_t$. Hence, $(\psi_t, V_t, V_t-25)$ is far from a Markov chain. This agrees with Fig. \ref{fig:TrainingCartVelocity}(c)-(d), where the training error and inference error are non-monotonic in the AoI for $u=1$. 

According to Shannon's interpretation of Markov sources in his seminal work \cite{Shannon1948}, \ignore{the data sequence}$(\psi_t, X_{t-\mu}, X_{t-\mu-\nu})$ becomes closer to a Markov chain, as the size $u$ of feature vector $X_{t-\delta}=(V_{t-\delta}, \ldots,$\ $V_{t-\delta-u+1})$ increases. In fact, $(\psi_t, X_{t-\mu}, X_{t-\mu-\nu})$ is precisely a Markov chain if $u=\infty$. One can observe from Fig. \ref{fig:TrainingCartVelocity}(c)-(d) that, as $u$ grows, the training and inference errors get close to non-decreasing functions of the AoI. This is because $(\psi_t, X_{t-\mu}, X_{t-\mu-\nu})$ tends to be Markovian as $u$ increases, i.e., the parameter $\epsilon$ of the $\epsilon$-Markov chain $\psi_t \overset{\epsilon}\leftrightarrow X_{t-\mu} \overset{\epsilon}\leftrightarrow X_{t-\mu-\nu}$ reduces to zero as $u$ grows. We note that one disadvantage of large feature size $u$ is that it increases the channel capacity needed for transmitting the features. 

\else
\fi
Fig. \ref{fig:DelayedNetworkedControlled} depicts the performance of robot state prediction in a leader-follower robotic system. As illustrated in a Youtube video \footnote{\url{https://youtu.be/_z4FHuu3-ag}.}, the leader robot sends its state (joint angles) $X_t$ to the follower robot through a channel. One packet for updating the leader robot's state is sent periodically to the follower robot every $20$ time-slots. The transmission time of each updating packet is $20$ time-slots. The follower robot moves towards the leader's most recent state and locally controls its robotic fingers to grab an object. We constructed a robot simulation environment using the Robotics System Toolbox in MATLAB.  In each episode, a can is randomly generated on a table in front of the follower robot. The leader robot observes the position of the can and illustrates to the follower robot how to grab the can and place it on another table, without colliding with other objects in the environment. The rapidly-exploring random tree (RRT) algorithm is used to control the leader robot. Collision avoidance algorithm and trajectory generation algorithm are used for local control of the follower robot. The leader robot uses a neural network to predict the follower robot's state $Y_t$. The neural network consists of one input layer, one hidden layer with $256$ ReLU activation nodes, and one fully connected (dense) output layer. The dataset contains the leader and follower robots' states in 300 episodes of continue operation. The first $80\%$ of the dataset is used for the training and the other $20\%$ of the dataset is used for the inference. In Fig. \ref{fig:DelayedNetworkedControlled}, the training and the inference error decreases in AoI, when AoI $\leq 25$ and increases in AoI when AoI $\geq 25$. In this case, even a fresh feature with AoI=0 is not good for prediction. In this experiment, $(Y_t, X_{t-\mu}, X_{t-\mu-\nu})$ is not a Markov chain for all $\mu, \nu \geq 0$. Hence, the training and the inference error are not non-decreasing functions of AoI. 

To facilitate understanding the experimental results in Fig. \ref{fig:DelayedNetworkedControlled}, we provide a toy example to interpret it: Let $X_t$ be a Markov chain and $Y_t = f(X_{t-d})$. One can view $X_t$ as the input of a causal system with delay $d \geq 0$, and $Y_t$ as the
system output. Because $Y_t = f(X_{t-d})$, a stale system input $X_{t-d}$ at time $t-d$ is informative for inferring the current output $Y_t$ at time $t$. If the training and inference datasets have similar empirical distributions, \ifreport by using Lemma \ref{ToyExampleLemma1} from Appendix \ref{ToyExample}, we get \else by the data processing inequality we can show \fi $H_L( \tilde Y_0 | \tilde X_{\delta})$ and $H_L(Y_t;\tilde Y_0 | X_{t-\delta})$ decrease with $\delta$ when $0 \leq \delta \leq d$ and increase with $\delta$ when $\delta \geq d$, which is similar to Fig. \ref{fig:DelayedNetworkedControlled}. Moreover, $H_L(\tilde Y_0 | \tilde X_d)$ is close to zero if the function space $\Lambda$ is sufficiently large. It is equal to zero if  $\Lambda = \Phi$. The leader-follower robotic system in Fig. \ref{fig:DelayedNetworkedControlled} can be viewed as a causal system, where the system input is the leader robot’s state, and the system output is the follower robot’s state. Non-monotonicity occurs in Fig. \ref{fig:DelayedNetworkedControlled} because the input of a causal system is used to predict the system output in this experiment, which is similar to the toy example. However, the relationship between the system input and output in Fig. \ref{fig:DelayedNetworkedControlled} is more complicated than the toy example, due to the control algorithms used by the follower robot.


\ifreport
In Fig. \ref{fig:Training}, we plot the performance of temperature prediction. In this experiment, the temperature $Y_t$ at time $t$ is predicted based on a feature $X_{t-\delta}=\{s_{t-\delta}, \ldots, s_{t-\delta-u+1}\}$, where $s_t$ is a $7$-dimensional vector consisting of the temperature, pressure, saturation vapor pressure, vapor pressure deficit, specific humidity, airtight, and wind speed at time $t$. Similar to \cite{kerasexample}, we have used an LSTM neural network and Jena climate dataset recorded by Max Planck Institute for Biogeochemistry. In this experiment, time unit of the sequence is $1$ hour. Due to the long-range dependence of weather data, if $u=1, 6,$ or $12$, $(Y_{t}, X_{t-\mu}, X_{t -\mu-\nu})$ is not a Markov chain. If $u=24$, then $Y_t \leftrightarrow X_{t-\mu} \leftrightarrow X_{t-\mu-\nu}$ is close to a Markov chain. Hence, when $u=1, 6,$ or $12$, the training error and the inference error are non-monotonic in AoI and when $u=24$, the training error and the inference error are close to a non-decreasing function of AoI. 

Fig. \ref{fig:Trainingcsi} illustrates the performance of channel state information (CSI) prediction. The CSI $h_t$ at time $t$ is predicted based on a feature $X_{t-\delta}=\{h_{t-\delta}, \ldots, h_{t-\delta-u+1}\}$. The dataset for CSI is generated by using Jakes model \cite{baddour2005autoregressive}. Due to long-range dependence of CSI, the training error and the inference error are non-monotonic in AoI. However, they become non-decreasing functions of AoI as $u$ grows. The phenomenon of long-range dependence is also observed in solar power prediction \cite{shisher2021age}.
\else
Besides these experiments, if there exists response delay, long-range dependence, and periodic patterns in the target and feature data sequence, the training and inference errors could also be non-monotonic functions of the AoI. This phenomenon is observed in actuator state prediction, temperature prediction, wireless channel state information prediction, and solar power prediction \cite{technical_report, shisher2021age}.
\fi

\ignore{\subsubsection{Pole Angle Prediction in a Cart Pole Balancing Task}
{\blue The objective of the cart pole balancing task is to keep a pole attached to a cart from falling over by applying forces on the cart. In our experiment, we predict the pole angle $Y_t=\psi_t$ at time $t$ by observing a time sequence of cart velocity $X_{t-\delta}=\{v_{t-\delta}, \ldots, v_{t-\delta-u+1}\}$ with length $u$, where $v_t$ is the cart velocity at time $t$ and $\delta$ is the AoI. Data traces in Fig \ref{fig:TrainingCartVelocity}(b) portray a phase delay between the pole angle $\psi_t$ and the cart velocity $v_t$. The phase delay occurs due to the response delay between the cart and the pole. The response delay is also observed in many other mechanical systems \cite{du2007h}. Due to the phase delay, the data sequence $Y_t \leftrightarrow X_{t-\mu} \leftrightarrow X_{t-\mu-\nu}$ becomes non-Markov. However, according to Shannon’s interpretation of Markov sources in his seminal work \cite{Shannon1948}, the larger $u$, the closer $Y_t \leftrightarrow X_{t-\mu} \leftrightarrow X_{t-\mu-\nu}$ tends to a Markov chain. Now, by using Theorem \ref{theorem1} and Theorem \ref{theorem2}, we explain experimental results in Fig. \ref{fig:TrainingCartVelocity}(c)-(d). When $u=1$ and $u=5$, the data sequence is non-Markov. Hence, the training error and the inference error are non-monotonic functions of AoI. In this case, even the freshest feature $(\delta=0)$ is not good for prediction. As the $u$ is increased to $10$, the data sequence gets close to Markov and the learning error curves get close to monotonic functions. }}

\ignore{\subsubsection{Car States Prediction in a Mountain Car Control Task} 
Fig. \ref{fig:DelayedNetworkedControlled}(a) illustrates a mountain car control task, where the goal is to drive an underpowered car up a steep mountain road \cite[pp. 214-215] {sutton2018reinforcement}. As depicted in Fig. \ref{fig:DelayedNetworkedControlled}(b), in our experiment, the car is controlled by a remote controller, where the control action is applied to the car after a delay $D$ and the states of the car are fed back to the controller with a delay $D$. The controller sends  the available information $X_t$ (observed states and its action) to the predictor. The predictor infer the current state of the car $Y_t$ by observing features $X_{t-\delta}$. This model embodies many networked controlled systems \cite{yang2006networked, gupta2009networked, peng2020switching}. In Fig. \ref{fig:DelayedNetworkedControlled}(c)-(d), we plot the training error and the inference error versus the AoI $\delta$ for different values of communication delay $D$. When $D>0$, the data sequence becomes non-Markov. However, when $D=0$, the data sequence becomes exactly Markov. For this reason, as explained in Theorem \ref{theorem1} and Theorem \ref{theorem2}, the training error and the inference error are non-monotonic functions of AoI when $D=5$ and $D=10$. When $D=0$, they become monotonic functions of AoI.
\ignore{{\blue The mountain car control task is to drive an underpowered car up a steep mountain road. As depicted in Fig. \ref{fig:DelayedNetworkedControlled}(b), in our experiment, the car is controlled by a remote controller, where the control action is applied to the car after a delay $D$ and the states of the car are fed back to the controller with a delay $D$. This model embodies many networked controlled systems \cite{yang2006networked, gupta2009networked, peng2020switching}. We predict states of the car $Y_t=s_t^{car}$ by observing features $X_{t-\delta}=\{s_{t-\delta}^{\text{cont}}, a_{t-\delta}^{\text{cont}}\}$, where $s_{t}^{\text{cont}}$ is the observed state at the controller side and $a_{t}^{\text{cont}}$ is the action taken by the controller at time $t$. Due to the communication delay $D>0$, the data sequence becomes non-Markov. However, when $D=0$, the data sequence becomes exactly Markov. As illustrated in Fig. \ref{fig:DelayedNetworkedControlled}(c)-(d), the training error and the inference error are non-monotonic functions of AoI when $D=5$ and $D=10$. When $D=0$, they become monotonic functions of AoI. This phenomenon validates our theoretical results in Theorem \ref{theorem1} and Theorem \ref{theorem2}.}} 

Besides the mechanical response delay and the communication delay, due to long-range dependence of the target and the feature, the data sequence becomes non-Markov \cite{guo2019credibility}, \cite[pp. 61-62]{sutton2018reinforcement}. Hence, the training error and the inference error can be non-monotonic functions of AoI. This phenomenon is observed in temperature prediction \cite{technical_report}, channel state information prediction \cite{technical_report}, and solar power prediction \cite{shisher2021age}.}

\ignore{In our experiments, the testing and training errors exhibit similar (non-)monotonic behaviors as the AoI grows. {\red Due to space limitations, the plots of the testing error are omitted.}} 

\ignore{\subsubsection{Inference Error under Random Inference AoI}
We generalize Theorem \ref{theorem3} to the case of random {\blue inference AoI $\Delta$} in the following theorem.
\begin{theorem}\label{theorem4}
If the conditions of Theorem \ref{theorem3} (a)-(b) hold and the inference AoIs in two experiments $1$ and $2$ satisfy 
\begin{align}\label{eq_stochastic_cond}
  \Delta_{1} \leq_{st} \Delta_{2},
\end{align}
then 
\begin{align}\label{RandomTesting}
    \!\!\!H_L(\tilde Y_{t}; Y_{t} | \tilde X_{t-\Delta_1}, \Delta_1) \!\! \leq \!H_L(\tilde Y_{t}; Y_{t} | \tilde X_{t-\Delta_2}, \Delta_2)+O\big(\max\{\epsilon, \beta\}\big). 
\end{align}
\ignore{\violet provided that the $L$-conditional cross entropies in \eqref{RandomTesting} are finite.} 
\end{theorem}
\ifreport
\begin{proof}
See Appendix \ref{Ptheorem4}.
\end{proof}
\else
\fi
If $H_L(Y_t)$ is twice differentiable in $P_{Y_t}$, then the approximation errors in \eqref{eq_theorem3} and \eqref{RandomTesting} become $O\big(\max\{\epsilon^2, \beta\}\big)$. In Theorem 3 of \cite{shisher2021age}, we analyzed how $H_L(\tilde Y_{t}; Y_{t} | \tilde X_{t-\Delta}, \Delta)$ varies with respect to the AoI,  under a condition $D_{\chi^2}(P_{\tilde Y_t, \tilde X_{t-\Delta}, \Delta}||P_{Y_t, X_{t-\Theta}, \Theta})\leq \beta^2$ that is stronger than \eqref{T3condition2} and \eqref{T3condition2_diff}. {\violet This condition requires the offline training AoI $\Theta$ and the online inference AoI $\Delta$ to have similar distributions, which is unnecessary and difficult to fulfill. Hence, the results} in Theorems \ref{theorem3} and \ref{theorem4} are more practical than that in Theorem 3 of \cite{shisher2021age}.}  

\ignore{\subsection{How to Choose the Training AoI $\Theta$} 
{\violet We now discuss how to choose the training AoI $\Theta$ for reducing the inference error in supervised learning based real-time forecasting. One can observe from \eqref{Decomposed_Cross_entropy} that $H_L(\tilde Y_t;Y_t; X_{t-\Delta}, \Delta)$ is determined by $P_\Delta$, $P_{\tilde Y_t|\tilde X_{t-\delta}}$, and $P_{ Y_t| X_{t-\delta}}$, but it is not affected by $P_\Theta$. Hence, the distribution $P_\Theta$ of the training AoI appears to be non-essential from an information-theoretic perspective.} 

{\violet In the practice of supervised learning based forecasting, we suggest preparing the training dataset with two considerations: (i) The set of training AoI should cover as many values in the set of inference AoI $\mathcal D$ as possible. If $P_{\Theta}(\delta)=0$, then no training data is available for the AoI value $\Delta=\delta$ and the corresponding inference error would be higher. In this case, the inference error depends on the generalization ability of the neural networks. (ii) Sufficient training data should be kept for each AoI value $\delta \in \mathcal D$ such that $P_{\tilde Y_t|\tilde X_{t-\delta}}$ and $P_{ Y_t| X_{t-\delta}}$ are close to each other, which is helpful to reduce the inference error.}}


\section{Single-source Scheduling for Inference Error Minimization}\label{Scheduling}
As shown in Section \ref{InformationAnalysis}, the inference error is a function of the AoI $\Delta(t)$, whereas the function is not necessarily monotonic. To reduce the inference error, we devise a new scheduling algorithm that can minimize general functions of the AoI, no matter whether the function is monotonic or not. 

\ifreport

\begin{figure}[t]
  \centering
\begin{subfigure}[b]{0.20\textwidth}
\includegraphics[width=1\linewidth]{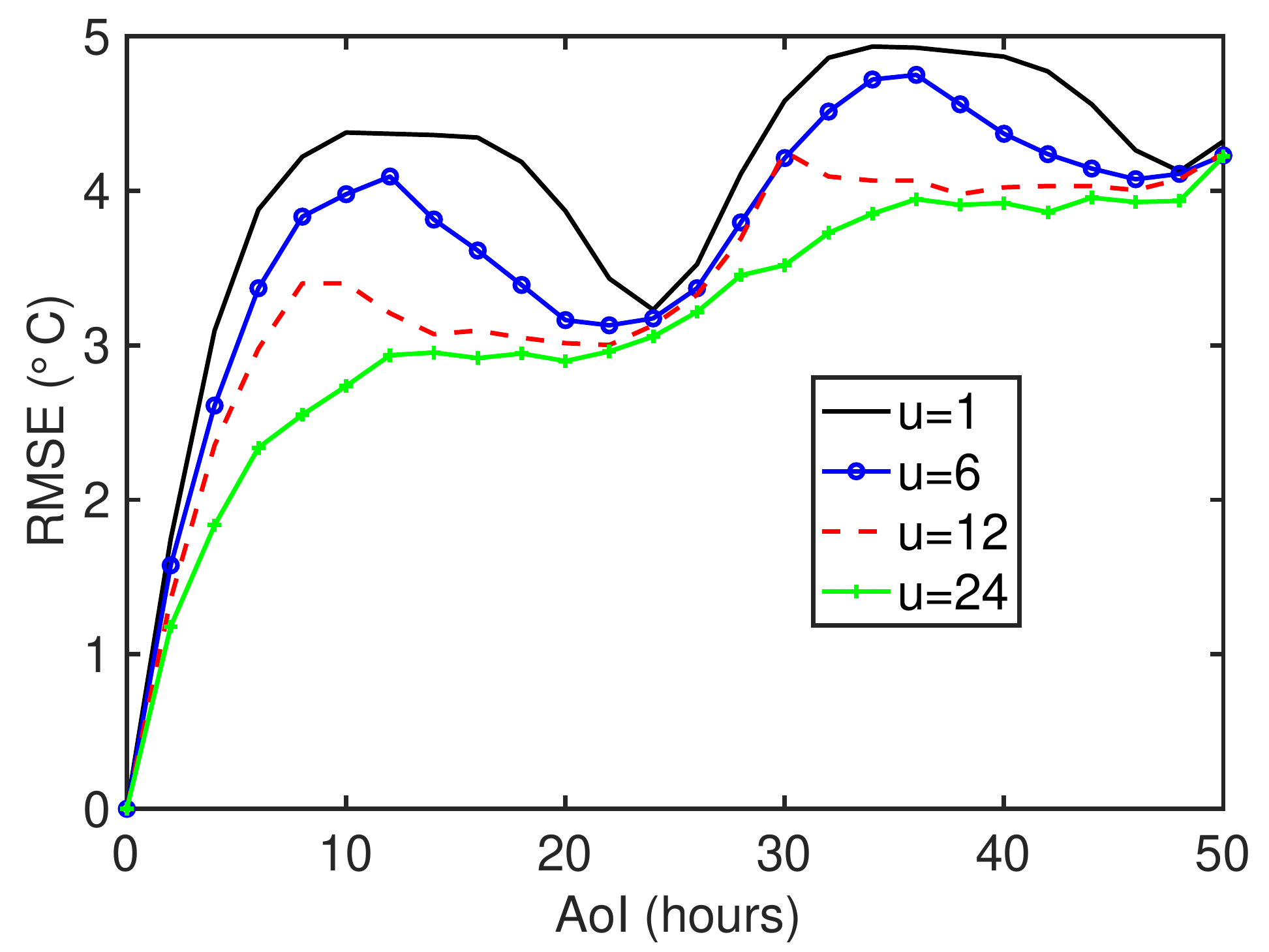}
  \subcaption{Training Error vs. AoI}
\end{subfigure}
\begin{subfigure}[b]{0.20\textwidth}
\includegraphics[width=1\linewidth]{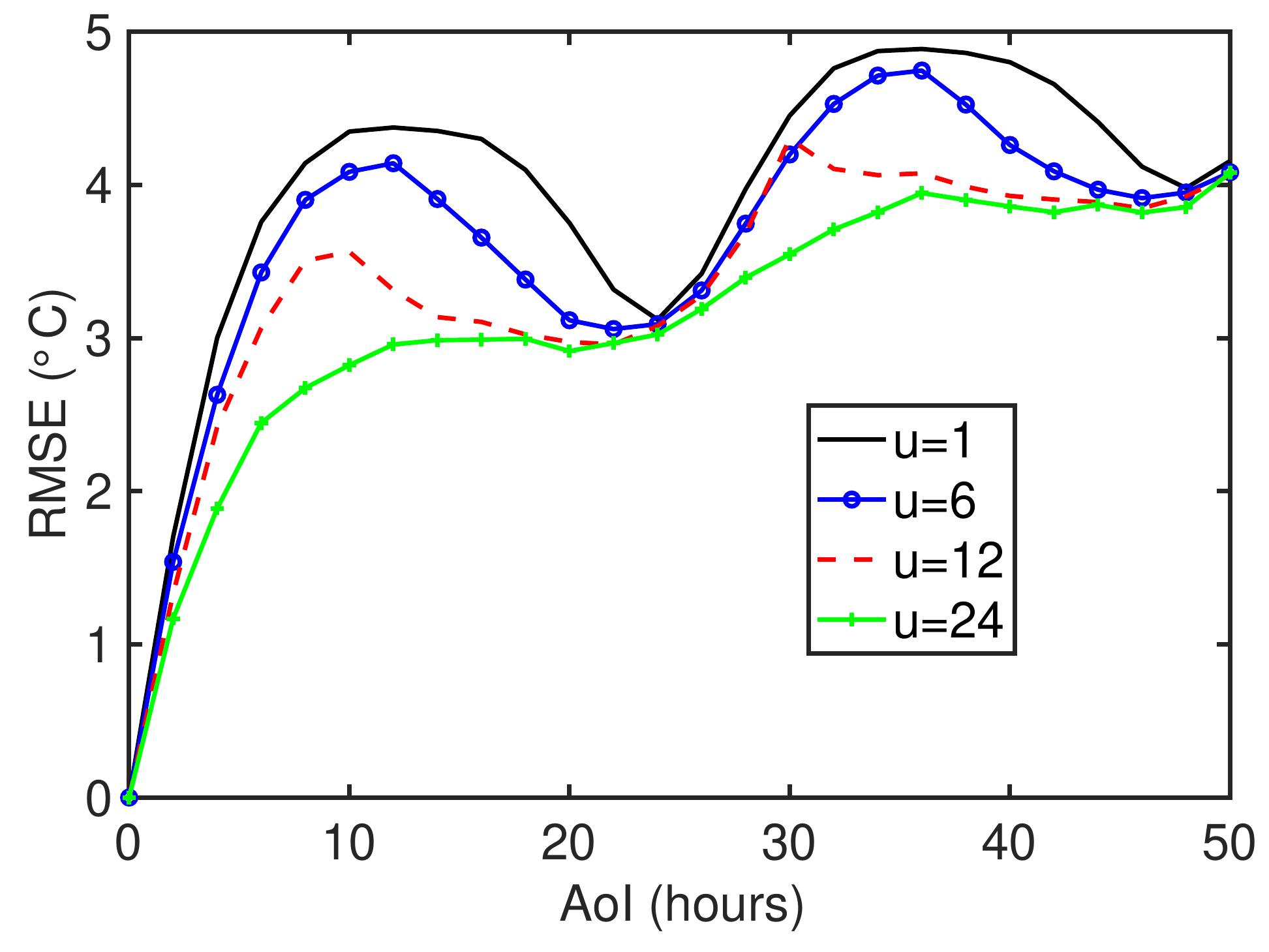}
  \subcaption{Inference Error vs. AoI}
\end{subfigure}
%
\caption{Performance of temperature Prediction. The training error and inference error are non-monotonic in AoI. As $u$ increases, the errors tend closer to non-decreasing functions of the AoI.}
\label{fig:Training}
\end{figure}

\begin{figure}[t]
  \centering
\begin{subfigure}[b]{0.20\textwidth}
\includegraphics[width=1\linewidth]{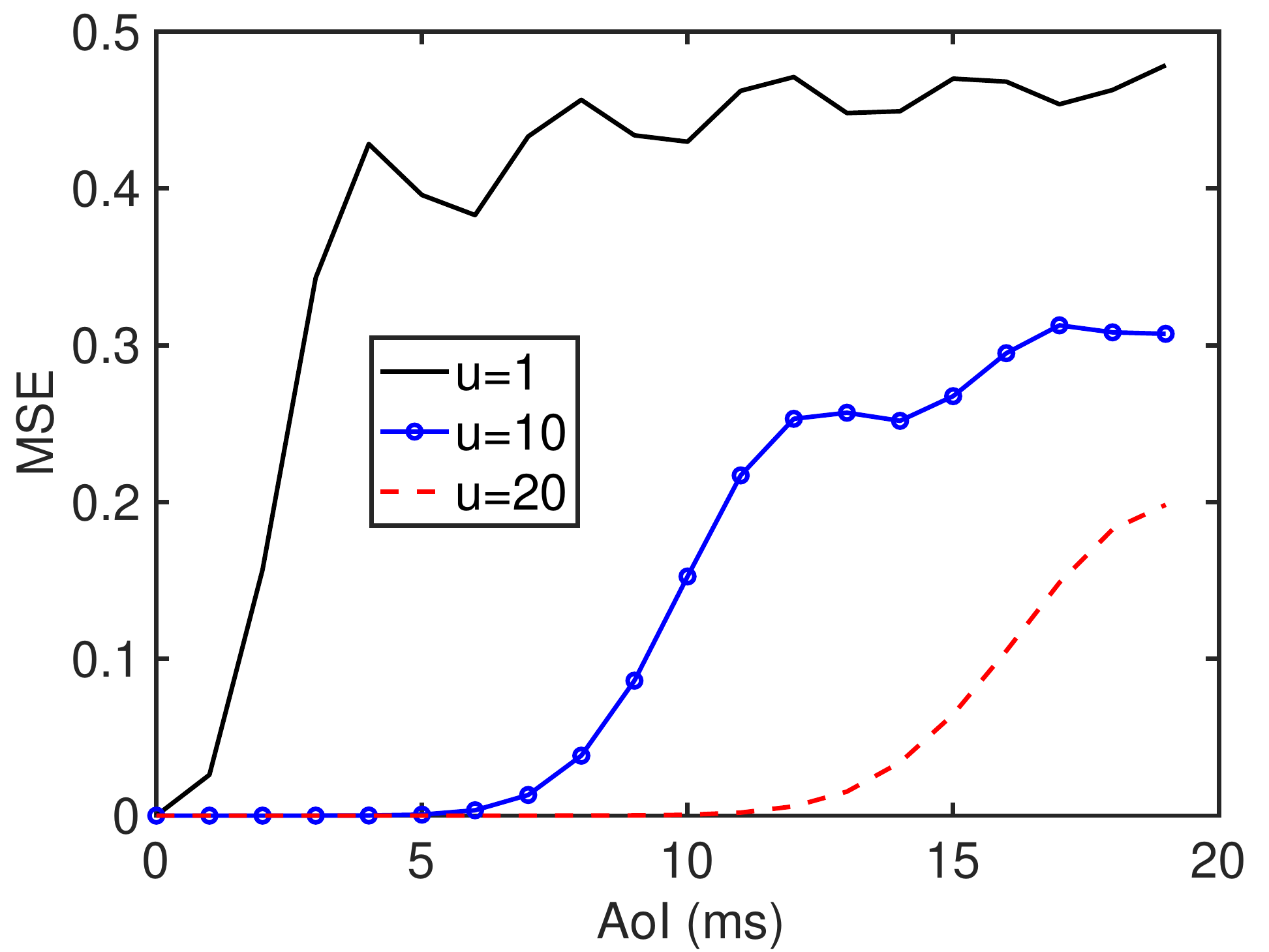}
  \subcaption{Training Error vs. AoI}
\end{subfigure}
\begin{subfigure}[b]{0.20\textwidth}
\includegraphics[width=1\linewidth]{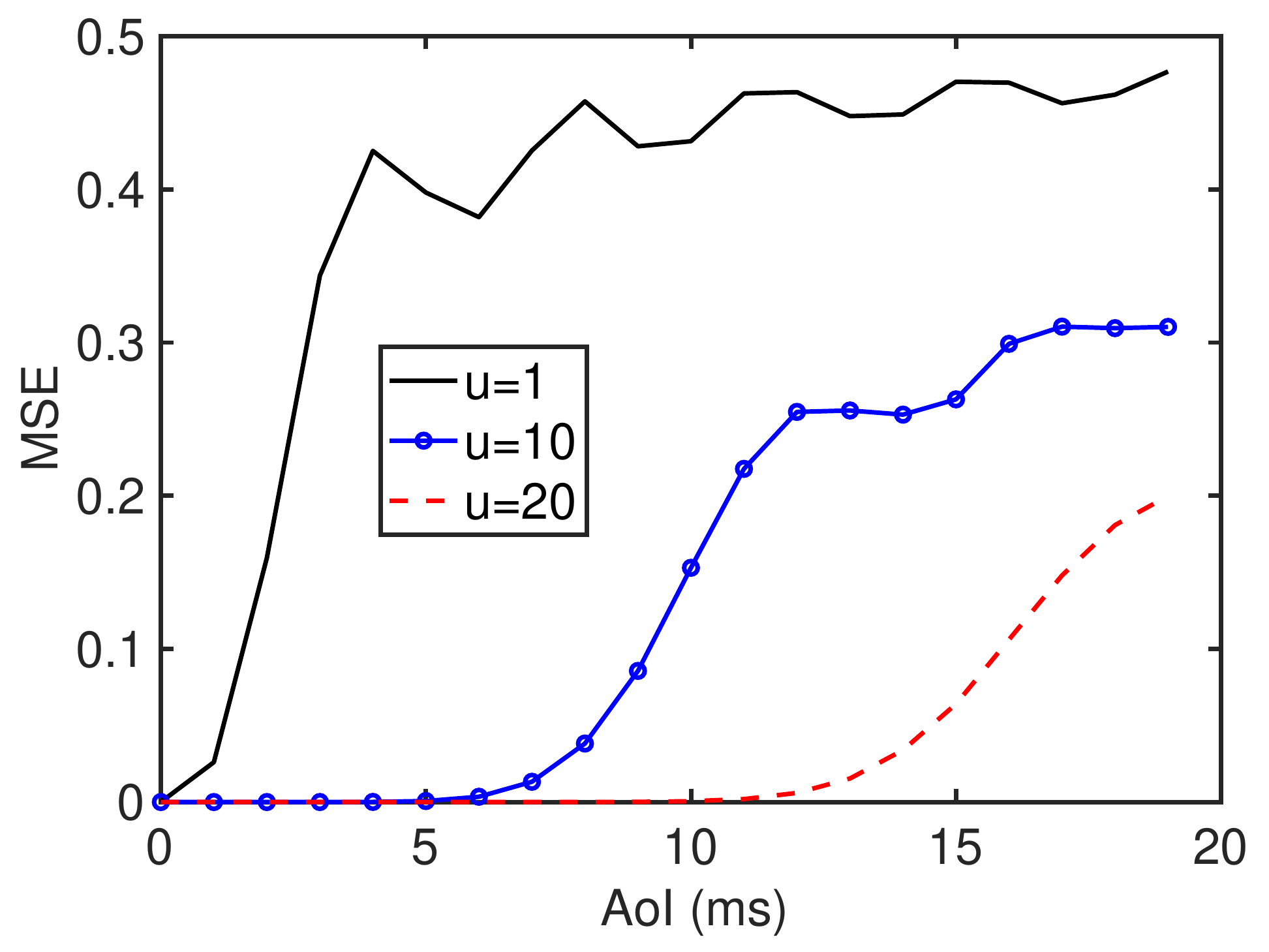}
  \subcaption{Inference Error vs. AoI}
\end{subfigure}
%
\caption{Performance of channel state information (CSI) prediction. The training error and inference error are non-monotonic in AoI. As $u$ increases, the errors tend closer to non-decreasing functions of the AoI.}
\label{fig:Trainingcsi}
\end{figure}

\else
\fi

\ignore{In this section, we will design a new scheduling algorithm for minimizing the long-term average inference error in supervised learning based real-time forecasting. Following the discussions in Section \ref{InformationAnalysis}, we can show that the inference error is a function of the inference AoI, which is not necessarily monotonic. Hence, existing scheduling algorithms for optimizing monotonic AoI metrics \cite{bedewy2020optimizing,SunTIT2020,bedewy2021optimal,SunNonlinear2019, OrneeTON2021,klugel2019aoi,  pan2020minimizing, jiang2018can, Maatouk2021, abd2020aoi,sun2017update, yates2021age, li2021age} may not be appropriate for inference error minimization. We will provide a semi-analytical optimal solution for minimizing the time-average of non-monotonic AoI functions, where the solution is expressed by using the Gittins index of an AoI bandit process. To the extent of our knowledge, the connection between AoI-based scheduling and Gittins index that we discover has not been reported before.}

\subsection{System Model}
We consider the networked supervised learning system in Fig. \ref{fig:scheduling}, where a source progressively sends features through a channel to a receiver. The channel is modeled as a non-preemptive server with i.i.d. service times. At any time $t$, the receiver uses the latest received feature to predict the current label $Y_t$. To minimize the inference error, we propose a new ``selection-from-buffer'' model for feature transmissions, which is more general than the ``generate-at-will'' model \cite{yates2015lazy}. Specifically, at the beginning of time slot $t$, the source generates a fresh feature $X_t$ and appends it to a buffer that stores the $B$ most recent features $(X_t, X_{t-1}, \ldots, X_{t-B+1})$; meanwhile, the oldest feature $X_{t-B}$ is removed from the buffer. The transmitter can pick any feature from the buffer and submit it to the channel when the channel is idle. A transmission scheduler determines (i) when to submit features to the channel and (ii) which feature in the buffer to submit. When $B=1$, the ``selection-from-buffer'' model reduces to the ``generate-at-will'' model. 
\ignore{We consider the networked real-time forecasting system illustrated in Fig. \ref{fig:scheduling}. At the beginning of each time slot $t$, the transmitter generates a new feature $X_t$ and adds it into a buffer that stores $B$ most recent features $(X_t, X_{t-1}, \ldots, X_{t-B+1})$. Meanwhile, the oldest feature $X_{t-B}$ is moved out of the buffer. A transmission scheduler determines when to send the next feature and which feature in the buffer should be chosen {\violet as the next feature} sent over the channel to the receiver. The channel is modeled as a non-preemptive server with i.i.d. service times. Hence, once the channel starts to send a feature, it must complete the on-going feature transmission before switching to send another feature. At the beginning of each time slot $t$, the receiver feeds the freshest delivered feature to the neural network for predicting the current label $Y_t$.} 

We assume that the system starts to operate in time slot $t = 0$ with $B$ features $(X_0, X_{-1}, \ldots, X_{-B+1})$ in the buffer. Hence, the feature buffer is full at all time $t\geq 0$.  The $i$-th feature sent over the channel is generated in time slot $G_i$, is submitted to the channel in time slot $S_i$, is delivered and available for inference in time slot $D_i = S_i +T_i$, where $T_i\geq 1$ is the feature transmission time, $G_i \leq S_i<D_i$, and $D_i \leq S_{i+1}< D_{i+1}$. The feature transmission times $T_i$ could be random due to time-varying channel conditions, congestion,  random packet sizes, etc. We assume that the $T_i$'s are i.i.d. with a finite mean $1 \leq \mathbb E[T_i ] < \infty$. In time slot $t=S_i$, the $(b_i + 1)$-th freshest feature in the buffer is submitted to the channel, where $b_i \in \{0, 1, \ldots, B-1\}$. Hence, the submitted feature is $X_{S_i-b_i}$ that was generated at time $G_i=S_i-b_i$. Once a feature is delivered, an acknowledgment (ACK) is fed back to the transmitter in the same time slot. Thus, the idle/busy state of the channel is known at the transmitter.
\ignore{The system starts to operate at time slot $t=0$ and the feature buffer is kept full at all time. The $i$-th feature sent over the channel is generated at the beginning of time slot $G_i$ and its transmission starts at the beginning of time slot $S_i$. The channel server takes $T_i$ time slots to complete the transmission. The $i$-th feature is delivered and is available for usage at the beginning of time slot $D_i=S_i+T_i$. The system satisfy $G_i\leq S_i< D_i$, $S_i < S_{i+1}$, and $D_i< D_{i+1}$. The feature transmission times $T_i = D_i- S_i > 0$ are i.i.d. with a finite mean $0<E[T_i]<\infty$. If a feature is delivered at time slot $t$, the transmitter is fed back an acknowledgment (ACK) in the same time slot $t$. Thus, the idle/busy state of the channel is known at the transmitter at every time slot $t$.}

\ignore{The $i$-th feature sent over the channel is generated at time slot $G_i$, its transmission starts at time slot $S_i$ and completes at time slot $D_i$ such that $G_i\leq S_i< D_i$, $S_i < S_{i+1}$, and $D_i< D_{i+1}$. The feature transmission times $T_i = D_i- S_i > 0$ are i.i.d. with a finite mean $0<E[T_i]<\infty$. Once a feature is delivered, the receiver sends an acknowledgment (ACK) to the transmitter. The transmitter receives the feedback at the same time slot the receiver sends it. Thus, the idle/busy state of the channel is known at the transmitter at every time slot $t$.} 

\ignore{At time slot $S_i$, suppose that the scheduler selects to send the $(b_{i}+1)$-th freshest feature among the $B$ features $(X_{S_i}, X_{S_i-1}, \ldots,$ $X_{S_i-B+1})$ stored in the buffer, where $b_i \in\{0,1,\ldots, B-1\}$. Hence, the selected feature is $X_{G_i} = X_{S_i-b_i}$ that was generated at time slot $G_i = S_i- b_i$. A scheduling policy is denoted by a 2-tuple $(g,f)$, where $g=(S_1, S_2,\ldots)$ are the transmission starting times of the features and {\violet $f=(b_1,b_2,\ldots)$ determines which feature from the buffer to send out.} We consider a class of causal scheduling policies in which each decision is made by using the current and history information available at the transmitter. Let $\mathcal F \times \mathcal G$ be the set of all the causal scheduling policies, where $f \in \mathcal F$ and $g \in \mathcal G$. {\violet We assume that the scheduler has access to the distribution of the feature process $\{(\tilde Y_t , \tilde X_t), t\in \mathbb Z \}$ but not its realization, and the feature transmission times $T_i$'s are not affected by the adopted scheduling policy.}}

\subsection{Scheduling Problem}
Let $U(t)= \max_i\{G_i : D_i \leq t \}$ be the generation time of the latest received feature in time slot $t$. The age of information (AoI) at time $t$ is given by \cite{kaul2012real}
\begin{align}\label{age}
\Delta(t) = t-U(t)=t-\max_i \{G_i: D_i \leq t\}.
\end{align}
Because $D_i < D_{i+1}$, $\Delta(t)$ can be also written as
\begin{align}
\Delta(t) = t - G_i=t-S_i+b_i,~~ \mathrm{if}~~D_i \leq t < D_{i+1}.
\end{align}
The initial state of the system is assumed to be $S_0 = 0, D_0 = T_0$, and $\Delta(0)$ is a finite constant.

A scheduling policy is denoted by a 2-tuple $(f, g)$, where $g = (S_1, S_2, \ldots)$ determines when to submit the features and \ $f= (b_1, b_2,$ $ \ldots)$ specifies which feature in the buffer to submit. We consider the class of \emph{causal scheduling policies} in which each decision is made by using the current and historical information available at the transmitter. Let $\Pi$ denote the set of all causal scheduling policies. We assume that the scheduler has access to the distribution of  $\{(Y_t , X_t), t \in \mathbb Z\}$ but not its realization, and the $T_i$’s are not affected by the adopted scheduling policy.

Our goal is to find an optimal scheduling policy that minimizes the time-average expected inference error among all causal scheduling policies in $\Pi$:
\begin{align}\label{scheduling_problem}
\bar p_{opt}=\inf_{(f, g) \in \Pi}  \limsup_{T\rightarrow \infty}\frac{1}{T} \mathbb{E}_{(f, g)} \left[ \sum_{t=0}^{T-1} p(\Delta(t))\right].
\end{align}
where $p(\Delta(t))$ is the inference error at time slot $t$, defined in \eqref{instantaneous_err1}, and $\bar p_{opt}$ is the optimum value of \eqref{scheduling_problem}. Because $p(\cdot)$ is not necessarily a non-decreasing function, \eqref{scheduling_problem} is more challenging than the scheduling problems in \cite{SunNonlinear2019, sun2017update}. 

\ifreport

\ignore{{\blue In order to evaluate the average inference error during time slots $0, 1, \ldots, T-1,$ let the random variable $\Delta$ in \eqref{eq_inferenceerror} follow the empirical distribution of the inference AoI process $\{\Delta(t), t = 0,1,\ldots, T-1\}$ during the first $T$ time slots. Then, the cumulative distribution function of $\Delta$ is given by \cite[Example 2.4.8]{durrett2019probability} 
\begin{align}\label{distribution_function_AoI}
F_{\Delta}(x) =  \frac{1}{T} \sum_{t=0}^{T-1} \mathbf 1(\Delta(t) \leq x), 
\end{align}
where $\mathbf 1(\cdot)$ is the indicator function. Because $\{(\tilde Y_t , \tilde X_t), t\in \mathbb Z \}$ is a stationary process that is independent of $\Delta(t)$, by substituting \eqref{distribution_function_AoI} into \eqref{eq_inferenceerror}, the average inference error during the first $T$ time slots can be expressed as 
\begin{align}\label{eq_inferenceerror1}
\mathrm{err}_{\mathrm{inference}}(T) = \frac{1}{T} \mathbb E_{(f,g)} \left [ \sum_{t=0}^{T-1} p(\Delta(t))\right],
\end{align}
where $E_{(f,g)}[\cdot]$ represents a conditional expectation for given scheduling policy $(f,g)$, $p(\Delta(t))$ is the inference error in time slot $t$, with the AoI function $p(\cdot)$ defined by 
\begin{align}\label{instantaneous_err} 
p(\delta)=\mathbb E_{Y, X \sim P_{\tilde Y_t, \tilde X_{t-\delta}}}\left[L\left(Y,\phi^*_{P_{Y_t, X_{t-\Theta},\Theta}}(X,\delta)\right)\right].
\end{align}
Let $T\rightarrow \infty$, we obtain the long-term average inference error over an infinite time horizon, which is determined by 
\begin{align}\label{eq_inferenceerror2}
\mathrm{err}_{\mathrm{inference}} =  \limsup_{T\rightarrow \infty}\frac{1}{T} \mathbb{E}_{(f, g)} \left[ \sum_{t=0}^{T-1} p(\Delta(t))\right].
\end{align}
}}

\ignore{The instantaneous error $p(\delta)$ is not completely known beforehand because the distribution $P_{\tilde Y_t, \tilde X_{t-\delta}}$ of inference data is unknown. However, we can estimate instantaneous inference error by evaluating the trained predictor $\phi^*_{P_{Y_t, X_{t-\Theta},\Theta}}$ on the testing data $(\hat Y_t, \hat X_t)$.}

\ignore{Our goal is to find an optimal scheduling policy that minimizes the long-term average inference error among all causal scheduling policies in $\mathcal F \times \mathcal G$:
\begin{align}\label{scheduling_problem}
\bar p_{opt}=\inf_{f \in \mathcal F, g \in \mathcal G}  \limsup_{T\rightarrow \infty}\frac{1}{T} \mathbb{E}_{(f, g)} \left[ \sum_{t=0}^{T-1} p(\Delta(t))\right].
\end{align}
where $p(\Delta(t))$ is the inference error at time slot $t$ and $\bar p_{opt}$ is the optimum value of \eqref{scheduling_problem}, i.e., the minimum time-average expected inference error. 

{\violet The new transmission model in Fig. \ref{fig:scheduling} is different from the widely used ``generate-at-will'' model \cite{yates2015lazy}, due to the additional feature buffer. As we will see in the next subsection, if $p(\cdot)$ is a non-decreasing function, then it is optimal to choose $b_i=0$ for all $i$ and the feature buffer is not needed. However, if $p(\cdot)$ is highly non-monotonic or even has a periodic pattern, it is better to send an old feature that was generated some time ago than sending a new feature that has just been created.}}

\begin{figure}[t]
\centering
\includegraphics[width=0.45\textwidth]{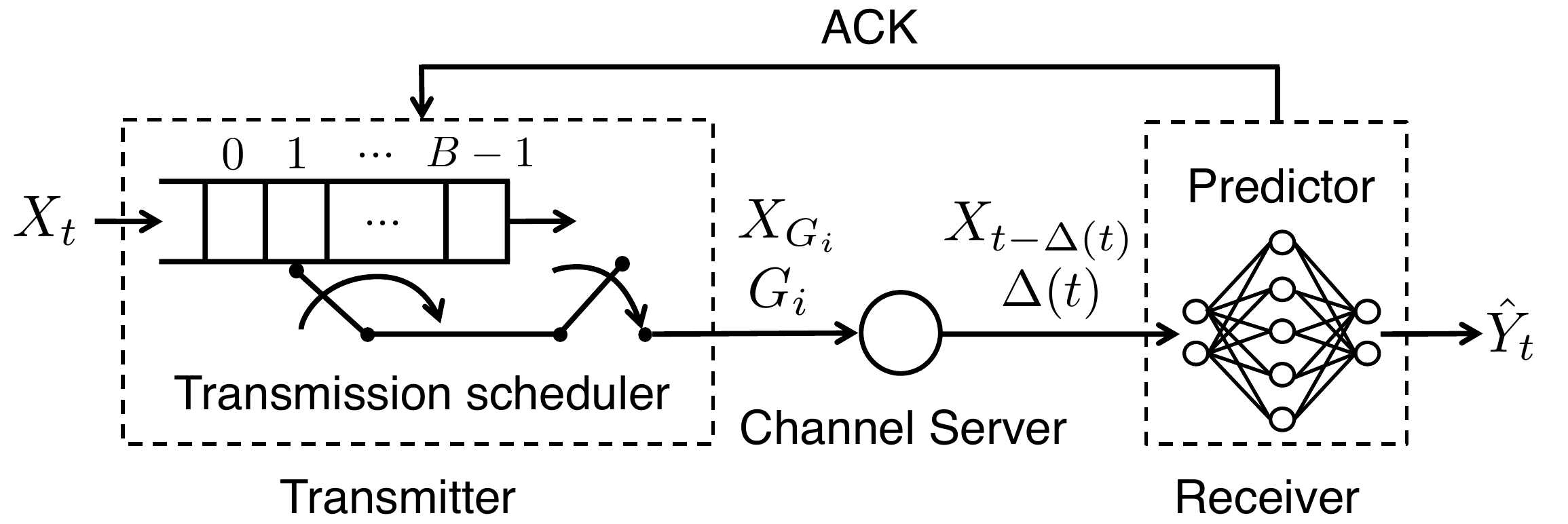}
\caption{\small  A networked real-time supervised learning system. At each time slot $t$, the transmitter generates a feature $X_t$ and keeps it in a buffer that stores $B$ most recent features $(X_t, X_{t-1}, \ldots, X_{t-B+1}$). The scheduler decides when to submit features to the channel and which feature in the buffer to submit. \label{fig:scheduling}
}
\vspace{-3mm}
\end{figure}
\subsection{Optimal Single-source Scheduling}
We solve \eqref{scheduling_problem} in two steps: (i) Given a fixed feature selection policy $f_b = (b, b, \ldots)$ with $b_i=b$ for all $i$, find the optimal feature submission times $g= (S_1, S_2, \ldots)$ that solves 
\begin{align}\label{sub_scheduling_problem}
\bar p_{b}=\inf_{(f_b, g)\in\Pi}  \limsup_{T\rightarrow \infty}\frac{1}{T} \mathbb{E}_{(f_b, g)}\left [ \sum_{t=0}^{T-1} p(\Delta(t))\right],
\end{align}
(ii) Use the solution to \eqref{sub_scheduling_problem} to describe an optimal solution to \eqref{scheduling_problem}.

It turns out that optimal solution to \eqref{sub_scheduling_problem} can be obtained by using the Gittins index of the following \emph{AoI bandit process with a random termination delay $T_1$}: 
A bandit process $\Delta(t)$ is controlled by a decision-maker that chooses between two actions \textsc{Continue} and \textsc{Stop} in each time slot. If the bandit process is not terminated in time slot $t$, its state evolves according to 
\begin{align}\label{Bandit}
\Delta(t)=\Delta(t-1)+1,
\end{align}
and a reward $[r - p(\Delta(t))]$ is collected, where $p(\cdot)$ is defined in \eqref{instantaneous_err1} and $r$ is a constant reward. If the \textsc{Continue} action is selected, the bandit process continues to evolve. If the \textsc{Stop} action is selected, the bandit process will terminate after a random delay $T_1$ and no more action is taken. Once the bandit process terminates, its state and reward remain zero. The total profit of the bandit process starting from time $t$ is maximized by solving the following optimal stopping problem:
\begin{align}\sup_{\nu \in \mathfrak M} \mathbb E\left[ \sum_{k=0}^{\nu+T_1-1} [r-p(\Delta(t+k))]\bigg| \Delta(t)=\delta \right],\end{align} 
where $\nu \geq 0$ is a history-dependent stopping time and $\mathfrak M$ is the set of all stopping times of the bandit process $\{\Delta(t+k), k=0, 1, \ldots\}$. Following the derivation of the Gittins index in \cite[Chapter 2.5] {gittins2011multi}, the decision-maker should choose the \textsc{Stop} action at time 
\begin{align}
\min_{t \in \mathbb Z}\{t \geq 0  : \gamma(\Delta(t)) \geq r\}, 
\end{align} 
where
\begin{align}\label{gittins-expand}
\gamma(\delta)\!=\inf_{\nu \in \mathfrak M, \nu \neq 0} \frac{ \mathbb E \left [\sum_{k=0}^{\nu-1}p(\Delta(t+k+T_{1}))~\bigg |~ \Delta(t)=\delta \right]}{\mathbb E[\nu~ |~ \Delta(t)=\delta]}
\end{align} 
is the Gittins index, i.e., the value of reward $r$ for which the \textsc{Continue} and \textsc{Stop} actions are equally profitable at state $\Delta(t)=\delta$. 
\ifreport
As shown in Appendix \ref{GittinsDerivation}, \eqref{gittins-expand} can be simplified as 
\begin{align}\label{gittins}
\gamma(\delta)=\inf_{\tau \in \{1, 2, \ldots\}} \frac{1}{\tau} \sum_{k=0}^{\tau-1} \mathbb E \left [p(\delta+k+T_{1}) \right],
\end{align}
where $\tau$ is a positive integer.
\else
As shown in \cite{technical_report}, \eqref{gittins-expand} can be simplified as 
\begin{align}\label{gittins}
\gamma(\delta)=\inf_{\tau \in \{1, 2, \ldots\}} \frac{1}{\tau} \sum_{k=0}^{\tau-1} \mathbb E \left [p(\delta+k+T_{1}) \right],
\end{align}
where $\tau$ is a positive integer.
\fi
 
\ignore{{\violet An optimal solution to \eqref{scheduling_problem} is presented below in two steps:  We first fix a feature selection policy $f_b=(b, b, \ldots)$, i.e., $b_i=b$ for all $i$, and find an optimal sequence of transmission starting times $g=(S_1, S_2, \ldots)$ that solves for the following scheduling problem: 
\begin{align}\label{sub_scheduling_problem}
\bar p_{b}=\inf_{ g \in \mathcal G}  \limsup_{T\rightarrow \infty}\frac{1}{T} \mathbb{E}_{(f_b, g)}\left [ \sum_{t=0}^{T-1} p(\Delta(t))\right].
\end{align}
In the second step, the solution to \eqref{sub_scheduling_problem} is used to describe an optimal solution to \eqref{scheduling_problem}.}

{\violet We find that one solution to \eqref{sub_scheduling_problem} can be expressed by using the Gittins index \cite{gittins2011multi, sonin2008generalized,weber1992gittins,frostig1999four} of a bandit process $\Delta(t)$. At each time slot $t$, if the bandit process $\Delta(t)$ is running, it evolves as
\begin{align}\label{Bandit}
\Delta(t)=\Delta(t-1)+1,
\end{align}
and incurs a cost $p(\Delta(t))$. Meanwhile, a constant reward $\gamma$ is received in each time slot when the bandit process is running. To maximize the total profit, a controller determines when to stop the bandit process. Once a ``stop" decision is made, the process would take a random time duration $T_1$ to stop. After it is stopped, the cost and reward both become zero and never change again. Given $\Delta(t)=\delta$, if a ``stop" action is chosen at the beginning of time slot $t$, then the total profit received since time slot $t$ is 
\begin{align}\label{Stopped}
\mathbb E\left[ \sum_{k=0}^{T_1-1} [\gamma-p(\Delta(t+k))]\bigg| \Delta(t)=\delta \right].
\end{align}
If a ``continue" action is taken at the beginning of time slot $t$, then the highest possible total profit received since time slot $t$ is
\begin{align}\label{continue}
\sup_{\nu \in \mathfrak M} \mathbb E\left[ \sum_{k=0}^{\nu+T_1-1} [\gamma-p(\Delta(t+k))]\bigg| \Delta(t)=\delta \right],
\end{align}
where $\nu \geq 1$ is a history-dependent stopping time and $\mathfrak M$ is the set of all positive stopping times of the bandit process $\{\Delta(t+k), k=0, 1, \ldots\}$. Gittins index $\gamma(\delta)$ \cite{gittins2011multi} is the value of reward $\gamma$ for which the ``stop" and ``continue" actions are equally profitable at state $\Delta(t)=\delta$, i.e., \eqref{Stopped} is equal to \eqref{continue}.} 
\ifreport
As shown in Appendix \ref{GittinsDerivation}, $\gamma(\delta)$ can be simplified as  
\begin{align}\label{gittins}
\gamma(\delta)=\inf_{\tau \in \{1, 2, \ldots\}} \frac{1}{\tau} \sum_{k=0}^{\tau-1} \mathbb E \left [p(\delta+k+T_{1}) \right],
\end{align}
where $\tau$ is a positive integer.
\else
As shown in our technical report \cite{technical_report}, $\gamma(\delta)$ can be simplified as  
\begin{align}\label{gittins}
\gamma(\delta)=\inf_{\tau \in \{1, 2, \ldots\}} \frac{1}{\tau} \sum_{k=0}^{\tau-1} \mathbb E \left [p(\delta+k+T_{1}) \right],
\end{align}
where $\tau$ is a positive integer.
\fi}
\begin{theorem}\label{theorem5}
If $|p(\delta)| \leq M$ for all $\delta$ and the $T_i$'s are i.i.d. with a finite mean $\mathbb E[T_i]$, then $g=(S_1(\beta_b), S_2(\beta_b), \ldots)$ is an optimal solution to \eqref{sub_scheduling_problem}, where 
\begin{align}\label{OptimalPolicy_Sub}
S_{i+1}(\beta_b) = \min_{t \in \mathbb Z}\big\{ t \geq D_i(\beta_b): \gamma(\Delta(t)) \geq \beta_b\big\},
\end{align}
$D_i(\beta_b)=S_{i}(\beta_b)+T_i$ is the delivery time of the $i$-th feature submitted to the channel, $\Delta(t)=t-S_i(\beta_b)+b$ is the AoI at time $t$, $\gamma(\delta)$ is the Gittins index in \eqref{gittins}, and
$\beta_b$ is the unique root of 
\begin{align}\label{bisection}
\mathbb{E}\left[\sum_{t=D_i(\beta_b)}^{D_{i+1}(\beta_b)-1}  p\big(\Delta(t)\big)\right] - \beta_b~ \mathbb{E}\left[D_{i+1}(\beta_b)-D_i(\beta_b)\right]=0.
\end{align}
The optimal objective value to \eqref{sub_scheduling_problem} is given by 
\begin{align}\label{optimum}
\bar p_b=\frac{\mathbb{E}\left[\sum_{t=D_i(\beta_b)}^{D_{i+1}(\beta_b)-1}  p\big(\Delta(t)\big)\right]}{\mathbb{E}\left[D_{i+1}(\beta_b)-D_i(\beta_b)\right]}.
\end{align}
Furthermore, $\beta_b$ is exactly the optimal value to \eqref{sub_scheduling_problem}, i.e., $\beta_b= \bar p_b$.
\end{theorem}
\ifreport
\begin{proof}
See Appendix \ref{Ptheorem5}.
\end{proof}
\else
\fi

The optimal scheduling policy in Theorem \ref{theorem5} has an intuitive structure. Specifically, a feature is transmitted in time-slot $t$ if two conditions are satisfied: (i) The channel is idle in time-slot $t$, (ii) the Gittins index $\gamma(\Delta(t))$ exceeds a threshold $\beta_b$ (i.e., $\gamma(\Delta(t)) \geq \beta_b$), where the threshold $\beta_b$ is exactly equal to the minimum time-averaged inference error $\bar p_b$. The optimal objective value $\bar p_{b}$ is computed by solving \eqref{bisection}. Three low-complexity algorithms for solving \eqref{bisection} were provided in \cite[Algorithms 1-3]{orneeTON2021}. In practical supervised learning algorithms, the features are shifted, rescaled, and clipped during the data preprocessing step, which can improve the convergence speed. Because of these operations, the inference error is finite in practice (See Figures \ifreport \ref{fig:learning}-\ref{fig:Trainingcsi} \else \ref{fig:learning}-\ref{fig:DelayedNetworkedControlled} \fi for a few example), and the condition $|p(\delta)| \leq M$ for all $\delta$ in Theorem \ref{theorem5} is not restrictive in practice. 

Theorem \ref{theorem5} is proven by directly solving the Bellman optimality equation of the Markov decision process \eqref{sub_scheduling_problem}, whereas the techniques for minimizing non-decreasing AoI functions in, e.g., \cite{SunNonlinear2019, sun2017update}, could not solve \eqref{sub_scheduling_problem}. We remark that if $p(\delta)$ is non-monotonic, then $\gamma(\delta)$ is not necessarily monotonic. Hence, \eqref{OptimalPolicy_Sub} in general could not be rewritten as a threshold policy of the AoI $\Delta(t)$ in the form of $\Delta(t) \geq \beta$. This is a key difference from the minimization of non-decreasing AoI functions, e.g., \cite[Eq. (48)]{SunNonlinear2019}. The adoption of the Gittins index $\gamma(\delta)$ as a tool for solving \eqref{sub_scheduling_problem} is motivated by a similarity between \eqref{sub_scheduling_problem} and the restart-in-state formulation of the Gittins index \cite[Chapter 2.6.4] {gittins2011multi}. This connection between the Gittins index theory and AoI minimization was unknown before.

Next, we present an optimal solution to \eqref{scheduling_problem}.

\begin{theorem}\label{theorem6}
If the conditions of Theorem \ref{theorem5} hold, then there exists an optimal solution $(f^*, g^*)$ to \eqref{scheduling_problem} that satisfies:
\begin{itemize}
\item[(a)] $f^*=(b^*, b^*, \ldots)$, where $b^*$ is obtained by solving 
\begin{align}\label{feature_Selects}
b^*= \arg \min_{b \in \{0, 1, \ldots, B-1\}} \beta_{b},
\end{align}
and $\beta_b$ is the unique root to \eqref{bisection}. 
\item[(b)] $g^* = (S_1^*,S_2^*,\ldots)$ , where 
\begin{align}\label{Optimal_Scheduler}
S_{i+1}^*=\min_{t \in \mathbb Z} \big\{ t \geq S_i^* + T_i: \gamma(\Delta(t)) \geq \bar p_{opt}\big\},
\end{align}
$S_i^*+T_i$ is the delivery time of the $i$-th feature, $\gamma(\delta)$ is the Gittins index in \eqref{gittins}, and $\bar p_{opt}$ is the optimal objective value of \eqref{scheduling_problem}, determined by 
\begin{align}\label{optimammain}
\bar p_{opt}= \min_{b \in \{0, 1, \ldots, B-1\}} \beta_b.
\end{align}
\end{itemize}
\end{theorem}
\ifreport
\begin{proof}
See Appendix \ref{Ptheorem6}.
\end{proof}
\else
\fi

Theorem \ref{theorem6} tells us that, to solve \eqref{scheduling_problem}, a feature is transmitted in time-slot $t$ if two conditions are satisfied: (i) The channel is idle in time-slot $t$, (ii) the Gittins index $\gamma(\Delta(t))$ exceeds a threshold $\bar p_{opt}$ (i.e., $\gamma(\Delta(t)) \geq \bar p_{opt}$), where the threshold $\bar p_{opt}$ is the optimal objective value of \eqref{scheduling_problem}. The optimal objective value $\bar p_{opt}$ is determined by \eqref{optimammain}.
\ignore{{\violet According to Theorem \ref{theorem6}, in an optimal solution to \eqref{scheduling_problem}, the $(i+1)$-th feature  should be sent in the earliest time slot $t$ that satisfies (i) the previously sent feature has already been delivered, and (ii) the Gittins index $\gamma(\Delta(t))$ is no less than the optimal objective value \emph{$\bar p_{opt}$}. Moreover, the scheduler should always send the $(b^*+1)$-th freshest feature in the buffer, where $b^*$ is the minimizer of \eqref{feature_Selects}.}}

In the special case of non-decreasing $p(\cdot)$ studied in \cite{SunNonlinear2019, sun2017update}, the Gittins index in \eqref{gittins} can be simplified as $\gamma(\delta) = \mathbb E[p(\delta+T_{1})]$ and the optimal solution to \eqref{feature_Selects} is $b^*=0$ such that it is optimal to always select the freshest feature from the buffer. Hence, Theorem 3 in \cite{SunNonlinear2019} is recovered from Theorem \ref{theorem6}, and the ``generate-at-will'' model can achieve the minimum inference error in this special case. 

If $p(\cdot)$ is non-monotonic, as in the cases of \ifreport Fig. \ref{fig:TrainingCartVelocity} and Fig. \ref{fig:DelayedNetworkedControlled} \else Fig. \ref{fig:DelayedNetworkedControlled}, \fi the ``selection-from-buffer'' model could achieve better performance than the ``generate-at-will'' model, and the optimal scheduler is provided by Theorem \ref{theorem6}.
\ignore{\subsubsection{Special Case: Non-decreasing $p(\cdot)$} 
{\violet In the special case of non-decreasing age function $p(\cdot)$, the Gittins index $\gamma(\delta)$ in \eqref{gittins} reduces to 
\begin{align}\label{gittins_monotonic}
\gamma(\delta) = \mathbb E[p(\Delta(t+T_{1}))].
\end{align}
In this case, the following result follows immediately from Theorem \ref{theorem6} and \eqref{gittins_monotonic}, and it coincides with some earlier results in \cite{SunNonlinear2019, SunTIT2020, OrneeTON2021}.}

\begin{corollary}\label{corollary1}
{\violet If $p(\cdot)$ is a non-decreasing function and the random transmission times $T_i$ is i.i.d. with a finite mean $\mathbb E[T_i]$, then there exists an optimal policy $\pi^*=(f^*, g^*)$ to \eqref{scheduling_problem} that satisfies
\begin{itemize}
\item[(a)] $f^*=(0, 0, \ldots)$, i.e., $b_i=0$ for all $i$.
\item[(b)]  $g^* = (S_1^*,S_2^*,\ldots)$ , where 
\begin{align}
S_{i+1}^*=\min_{t \in \mathbb Z} \big\{ t \geq S_i^* + T_i: \mathbb E[p(\Delta(t+T_{1}))] \geq \bar p_{opt}\big\},
\end{align}
$S_i^*+T_i$ is the delivery time of the $i$-th feature, and the optimal objective value $\bar p_{opt}$ is equal to $\beta_0$, i.e., the unique root to \eqref{bisection} with $b= 0$.   
\end{itemize}}
\end{corollary}
{\violet In \cite{SunNonlinear2019, SunTIT2020, OrneeTON2021}, the term $\mathbb E[p(\Delta(t+T_{1}))]$ occurred in the optimal scheduling policy, but it did not have a good interpretation. In this section, we have shown that $\mathbb E[p(\Delta(t+T_{1}))]$ is the Gittins index of an AoI bandit process $\Delta(t)$ with cost $p(\Delta(t))$. In addition, the results in \cite{SunNonlinear2019, SunTIT2020, OrneeTON2021} were obtained under an additional condition that inter-sending times $\{S_{i+1}-S_i, i=1,2, \ldots\}$ form a regenerative process. This condition is no longer required in our new proof techniques for Theorems \ref{theorem5} and \ref{theorem6}, which directly solve the Bellman optimality equations of \eqref{scheduling_problem} and \eqref{sub_scheduling_problem}. Finally, Corollary \ref{corollary1} implies that, if $p(\cdot)$ is non-decreasing, it is optimal to keep $b_i=0$ for all $i$ and hence no feature buffer is needed (as in the "generate-at-will" model). However, if $p(\cdot)$ is non-monotonic, it could be beneficial to add a feature buffer and send an old feature that has been stored in the buffer for some time, as indicated by Theorem \ref{theorem6}.}}

\section{Multiple-source Scheduling}\label{Multi-scheduling}

\subsection{System Model and Scheduling Problem}
Consider the networked intelligent system in Fig. \ref{fig:multi-scheduling}, where $m$ sources send features over a shared channel to the corresponding neural predictors at the receivers. At time slot $t$, each source $l$ maintains a buffer that stores the $B_l$ most recent features $(X_{l,t},\ldots,$ $X_{l,t-B_l+1})$. When the channel is free, at most one source can select a feature from its buffer and submit the selected feature to the channel.

A centralized scheduler makes two decisions in each time slot: (i) which source should submit a feature to the shared channel and (ii) which feature in the selected source's buffer to submit. A scheduling policy is denoted by $\pi=(\pi_{l, b_l})_{l=1,2,\ldots,m,b_l=0,1,\ldots, B_l-1}$, where $\pi_{l, b_l}=(d_{l, b_l}(0), d_{l, b_l}(1),\ldots)$ and $d_{l, b_l}(t)\in\{0,1\}$ represents the scheduling decision for the $(b_l+1)$-th freshest feature $X_{l,t-b_l}$ of source $l$ in time slot $t$. If source $l$ submits the feature $X_{l,t-b_l}$ in its buffer to the channel in time slot $t$, then $d_{l, b_l} (t)= 1$; otherwise, $d_{l, b_l}(t)=0$. Let $c_{l, b_l}(t)\in\{0,1\}$ denote the channel occupation status of the $(b_l+1)$-th freshest feature $X_{l,t-b_l}$ of source $l$ in time slot $t$. If source $l$ submits the feature $X_{l,t-b_l}$ in its buffer to the channel in time slot $t$, then the value of $c_{l, b_l}(t)$ becomes $1$ and remains $1$ until it is delivered; otherwise, $c_{l, b_l}(t)=0$. It is required that $\sum_{l=1}^m \sum_{b_l=0}^{B_l-1}c_{l, b_l}(t) \leq 1$ for all $t$. Let $\Pi$ denote the set of all causal scheduling policies.

\ignore{A centralized scheduler makes two decisions in each time slot: (i) which source should submit a feature to the shared channel and (ii) which feature in the selected source's buffer to submit. A scheduling policy is denoted by a 2-tuple $(h,\pi)$, where $h = (b_l(0), b_l(1),\ldots)_{l=1}^m$ describes the feature selection decision of the sources and $\pi = (d_l(0),$  $d_l(1),\ldots)_{l=1}^m$ represents the source scheduling decision in each time slot. If $d_l (t)= 1$, then source $l$ submits its $(b_l (t) + 1)$-th freshest feature to the channel in time slot $t$; otherwise $d_l(t)=0$ and source $l$ remains silence. {\blue Let $c_l(t)$ denote the channel occupation status of the source $l$ at time slot $t$. If $c_l(t)=1$, then the channel is occupied by the source $l$ at time slot $t$; otherwise $c_l(t)=0$ and the source $l$ is inactive.} It is required that $\sum_{l=1}^m c_l(t) \leq $1 for all $t$. Let $\mathcal H \times \Pi$ denote the set of all causal scheduling policies.}

\begin{figure}[t]
\centering
\includegraphics[width=0.45\textwidth]{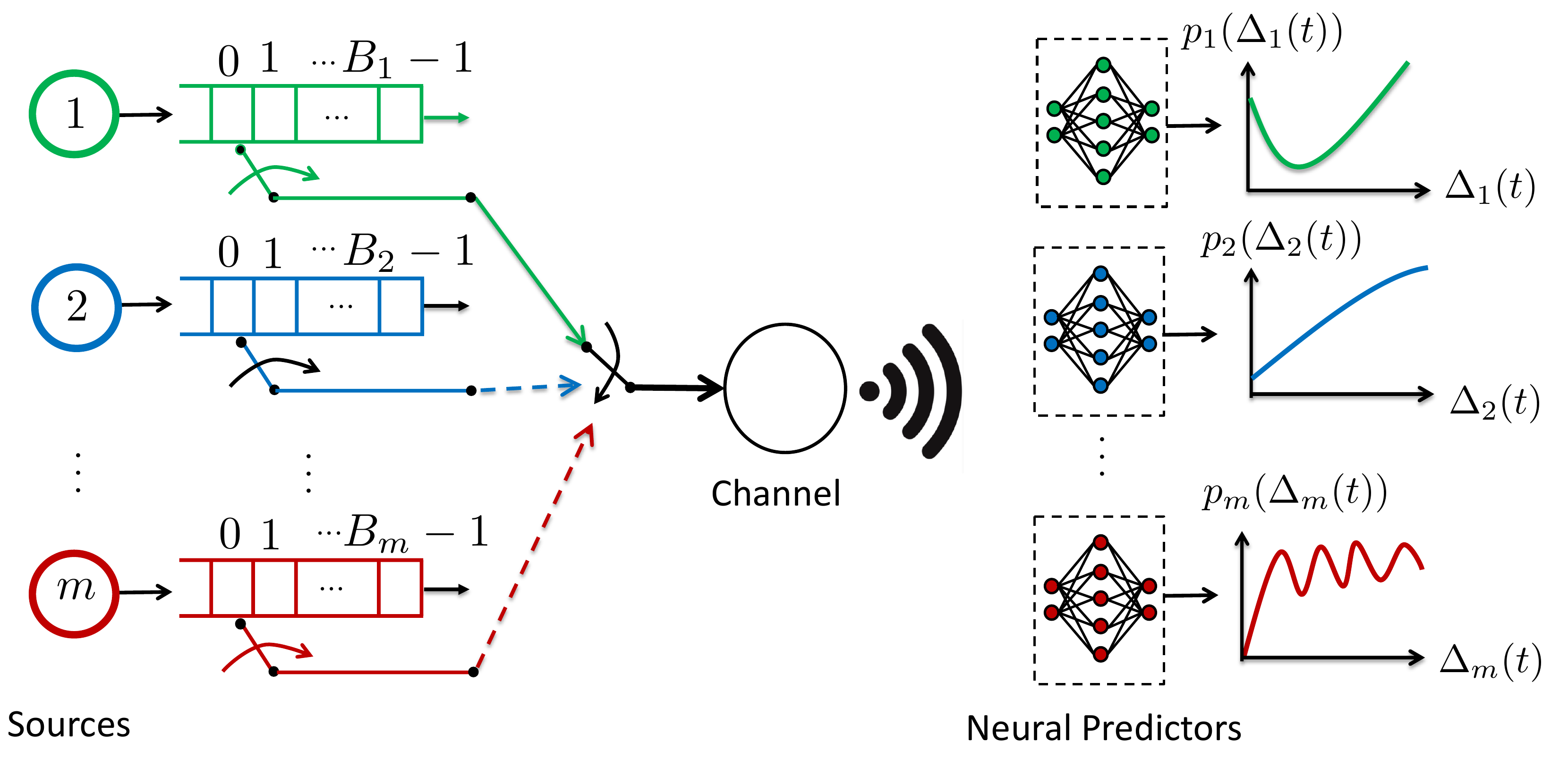}
\caption{\small A networked intelligent system, where $m$ sources send features over a shared channel to the corresponding neural predictors. At any time, at most one source can occupy the channel. \label{fig:multi-scheduling}
}
\vspace{-3mm}
\end{figure}

Let $G_{l, i}$, $S_{l,i}$, $D_{l, i}$, and $T_{l, i}$ denote the generation time, channel submission time, delivery time, and transmission time duration of the $i$-th feature sent by source $l$, respectively. The feature transmission times $T_{l, i}\geq 1$ are independent across the sources and i.i.d. among the features from the same source. We assume that the $T_{l, i}$'s are not affected by the adopted scheduling policy. The age of information (AoI) of source $l$ at time slot $t$ is given by 
\begin{align}\label{multi-source-Age}
\Delta_l(t)&=t-\max_i\{G_{l, i}: D_{l, i} \leq t\}. 
\end{align}
Our goal is to minimize the time-average weighted sum of the inference errors of the $m$ sources, which is formulated by
\begin{align}\label{Multi-scheduling_problem}
&\!\!\inf_{ \pi \in \Pi}  \limsup_{T\rightarrow \infty}\frac{1}{T} \sum_{l=1}^m w_l~\mathbb{E}_{\pi} \left[ \sum_{t=0}^{T-1} p_l(\Delta_l(t))\right], \\\label{Sceduling_constraint}
&\ \text{s.t.} \ \sum_{l=1}^m\sum_{b_l=0}^{B_l-1} c_{l, b_l}(t) \leq 1, ~~t=0, 1, 2, \ldots,
\end{align}
where $p_l(\Delta_l(t))$ is the inference error of source $l$ at time slot
$t$ and $w_l>0$ is the weight of source $l$. 

\subsection{Multiple-source Scheduling }
Problem \eqref{Multi-scheduling_problem} can be cast as a Restless Multi-arm Bandit (RMAB) problem by viewing the features stored in the source buffers as arms, where $(l, b_l)$ is an arm associated with the $(b_l+1)$-th freshest feature of the source $l$ and the state of the arm $(l, b_l)$ is the AoI $\Delta_{l}(t)$ in \eqref{multi-source-Age}. 
Finding the optimal solution for RMAB is generally PSPACE hard \cite{papadimitriou1994complexity}. Next, we develop a low-complexity scheduling policy by using both Gittins and Whittle indices. \ignore{According to \cite{weber1990index}, this policy is near-optimal for small number of arms and asymptotically optimal as the number of arms tends to infinite.}

\ignore{The following per-arm scheduling problem \eqref{decoupled_problem} is derived from \eqref{Multi-scheduling_problem}-\eqref{Sceduling_constraint} in two steps:} By relaxing the per-slot channel constraint \eqref{Sceduling_constraint} as the following time-average expected channel constraint
\begin{align}\label{Changed_constraint}
\limsup_{T \to \infty} \frac{1}{T} \sum_{t=0}^{T-1} \sum_{l=1}^m \sum_{b_l=0}^{B_l-1}\mathbb E[c_{l, b_l}(t)] \leq 1, 
\end{align}
and taking the Lagrangian dual decomposition of the relaxed scheduling problem \eqref{Multi-scheduling_problem} and \eqref{Changed_constraint}, we obtain following per-arm scheduling problem: 
\ignore{\begin{align}\label{decoupled_problem}
\!\!\!\!\bar J_{l}(\lambda)\!=\!\!\!\!\!\inf_{\substack{h_l \in \mathcal H_l\\ \pi_{l} \in \Pi_{l}}}  \limsup_{T\rightarrow \infty}\!\frac{1}{T} \!\mathbb{E}_{h_{l}, \pi_l}\!\left [ \sum_{t=0}^{T-1} \!w_l p_l(\Delta_l(t))\!+\!\lambda c_l(t)\right]\!,
\end{align}
where $h_l=(b_l(0), b_l(1), \ldots)$ denote the feature selection policy of the source $l$, $\pi_l=(d_l(0), d_l(1), \ldots)$ represents scheduling decision of source $l$, $\mathcal H_l \times \Pi_{l}$ is the set of all causal scheduling decisions for source $l$, where $h_l \in \mathcal H_l, \pi_l \in \Pi_l$, and $\lambda \geq 0$ is the Lagrangian dual variable associated to \eqref{Changed_constraint}. 

Now, by selecting feature selection policy $h_{l, b_l}=(b_l, b_l, \ldots) \in \mathcal H_l$, i.e., $b_l(t)=b_l$, we obtain following per-arm scheduling problem:} 
\begin{align}\label{decoupled_problem}
\!\!\!\!\inf_{\pi_{l, b_l} \in \Pi_{l, b_l}}  \limsup_{T\rightarrow \infty}\frac{1}{T} \mathbb{E}_{\pi_{l, b_l}}\left [ \sum_{t=0}^{T-1} \!w_l p_l(\Delta_l(t))\!+\!\lambda c_{l, b_l}(t)\right],\!
\end{align}
where $\Pi_{l, b_l}$ is the set of all causal scheduling policies of  arm $(l, b_l)$. 

\begin{definition}[\textbf{Indexability}]\cite{whittle1988restless}
Let $\Omega_{l, b_l}(\lambda)$ be the set of all AoI values $\delta$ such that if the channel is idle and $\Delta_l(t)=\delta$, the optimal action to \eqref{decoupled_problem} is $d_{l, b_l}(t)=0$. Then, the arm $(l, b_l)$ is \emph{indexable} if $\lambda_1 \leq \lambda_2$ implies $\Omega_{l, b_l}(\lambda_1) \subseteq \Omega_{l, b_l}(\lambda_2)$. 
\end{definition}

\begin{algorithm}[t]
\caption{Whittle Index Policy with Selection-from-Buffer}\label{alg:Whittle}
\begin{algorithmic}[1]
\State Do forever:
\State Update $\Delta_l(t)$ for all $l \in \{1, 2, \ldots m\}$.
\State Calculate the Whittle index $W_{l, b_l}(\Delta_l(t))$ for all $l \in \{1, 2, \ldots m\}$ and $b_l \in \{0, 1, \ldots, B_l-1\}$ using \eqref{Whittle_Index}-\eqref{waiting}.
\If{the channel is idle and $\max_{l, b_l} W_{l, b_l}(\Delta_l(t)) \geq 0$} 
    \State $(l^*, b_{l^*}^*) \gets  \arg\max_{l, b_l} W_{l, b_l}(\Delta_l(t))$.
    \State Source $l^*$ submits its feature $X_{l^*, t-b_{l^*}^*}$ to the channel.
\Else
   \State No source is scheduled, even if the channel is idle. 
\EndIf 
\end{algorithmic}
\end{algorithm}

\begin{theorem}\label{theorem7}
If $|p_l(\delta)| \leq M$ for all $\delta$ and the $T_{l, i}$'s are independent across the sources and i.i.d. among the features from the same source with a finite mean $\mathbb E[T_{l, i}]$, then all arms are indexable.
\end{theorem}
\ifreport
\begin{proof}
See Appendix \ref{ptheorem7}.
\end{proof}
\else
\fi

Given indexability, the Whittle index $W_{l, b_l}(\delta)$ \cite{whittle1988restless} of the arm $(l, b_l)$ at state $\delta$ is $W_{l, b_l}(\delta) = \inf \{\lambda \in\mathbb R: \delta\in \Omega_{l, b_l}(\lambda)\}$.  

\begin{theorem}\label{theorem8}
If the conditions of Theorem \ref{theorem7} hold, then the Whittle index $W_{l, b_l}(\delta)$ is given by
\begin{align}\label{Whittle_Index}
W_{l, b_l}(\delta)\!=&\frac{w_l}{\mathbb E[T_{l,1}]}~\mathbb{E}\left[ z(T_{l, 1}, b_l, \delta)+T_{l, 2}\right]~\gamma_l(\delta) \nonumber\\
&-\frac{w_l}{\mathbb E[T_{l,1}]}~\mathbb{E}\left[\sum_{t=T_{l, 1}}^{T_{l, 1}+z(T_{l, 1}, b_l, \delta)+T_{l, 2}-1}  \!\!\!\!\!\!p_l(t+b_l)\right],
\end{align}
where $\gamma_l(\delta)$ is the Gittins index of an AoI bandit process for source $l$, determined by 
\begin{align}
\gamma_l (\delta)=\inf_{\tau \in \{1, 2, \ldots\}} \frac{1}{\tau} \sum_{k=0}^{\tau-1} \mathbb E \left [p_l(\delta+k+T_{l, 2}) \right],
\end{align}
and 
\begin{align}\label{waiting}
z(T_{l, 1}, b_l, \delta)=\inf_{z \in \mathbb Z}\{ z \geq 0: \gamma_l(T_{l, 1}+b_l+z) \geq \gamma_l(\delta)\}.
\end{align}
\end{theorem}
\ifreport
\begin{proof}
See Appendix \ref{ptheorem8}.
\end{proof}
\else
\fi

Finding a (semi-)analytical expression of the Whittle index for minimizing non-monotonic AoI functions is in a challenging task. In Theorem \ref{theorem8}, this challenge is resolved by using the Gittins index $\gamma_l(\delta)$ to solve \eqref{decoupled_problem}, where the solution techniques of \eqref{sub_scheduling_problem} are employed. 
The Whittle index scheduling policy for reducing the weighted-sum inference error is described in Algorithm \ref{alg:Whittle}, where all sources remain silent when the channel is idle, if $W_{l, b_l}(\Delta_l(t)) < 0$ for all $l$ and $b_l$. 

In the special case that (i) the AoI function $p(\cdot)$ is non-decreasing and (ii) the transmission time is fixed as $T_{l,i} = 1$, it holds that $\gamma_l(\delta) = p_l(\delta+1)$ and $z(T_{l, 1}, b_l,\delta) =\max\{\delta-b_l - 1, 0\}$. Hence, 
\begin{align}
W_{l, 0}(\delta) = w_l\left[\delta~p_l(\delta+1)-\sum_{t=1}^{\delta} p_l(t)\right]
\end{align} for  $\delta\geq1$ and $b_l=0$. By this, the Whittle index in Section IV of \cite[Equation (7)]{Tripathi2019} is recovered from Theorem \ref{theorem8}. 

\ignore{{\red \subsubsection{Whittle Index}\label{Sec-whittle}
Now, we need to establish the indexability of the problem \eqref{decoupled_problem}. We denote $S_{l, b_l}(\lambda)$ as the set of AoI state $\delta$, where the optimal action to \eqref{decoupled_problem} is "do not transmit", i.e., if $\Delta(t)=\delta$, then $d_l(t)=0$. If the set $S_{l,b_l}(\lambda)$ is increasing in the cost $\lambda$, then the problem \eqref{decoupled_problem} is said to be indexable \cite{whittle1988restless}. Given indexability, Whittle index $W_{l, b_l}(\delta)$ \cite{whittle1988restless} is the infimum transmission cost $\lambda$ for which $\delta \in S_{l,b_l}(\lambda)$, given as  
\begin{align}\label{def-whittle}
W_{l, b_l}(\delta)=\inf_{\lambda}\{\lambda : \delta \in S_{l, b_l}(\lambda)\}.
\end{align}

\ignore{In other words, Whittle index $W_{l, b_l}(\delta)$ is the greatest transmission cost one willing to pay for transmitting $(b_l+1)$-th freshest feature from source $l$ at AoI state $\Delta_l(t)=\delta$.} 
\begin{theorem}\label{theorem7}
If there is an $M \geq 0$ such that $|p_l(\delta)| \leq M$ for all $\delta$ and the random transmission times $T_{l, i}$ are i.i.d. for all $i$ with a finite mean $\mathbb E[T_{l, i}]$, then the following assertions are true:
\begin{itemize}
\item[(a)] The indexability property holds for the problem \eqref{decoupled_problem}.
\item[(b_l)] Whittle index $W_{l, b_l}(\delta)$ is 
\begin{align}\label{Whittle_Index}
W_{l, b_l}(\delta)=&\mathbb{E}\left[D_{l, i+1}\big(\gamma_l(\delta)\big)-D_{l, i}\big(\gamma_l(\delta)\big)\right] \gamma_l(\delta)~ \nonumber\\
&-\mathbb{E}\left[\sum_{t=D_{l, i}\big(\gamma_l(\delta)\big)}^{D_{l, i+1}\big(\gamma_l(\delta)\big)-1}  p_l\big(\Delta_l(t)\big)\right],
\end{align}
where Gittins index $\gamma_l(\delta)$ is determined by
\begin{align}
\gamma_(\delta)=\inf_{\tau \in \{1, 2, \ldots\}} \frac{1}{\tau} \sum_{k=0}^{\tau-1} \mathbb E \left [p_l(\delta+k+T_{l, 1}) \right],
\end{align}
$D_{l, i}(\gamma_l(\delta))=S_{l, i}(\gamma_l(\delta))+T_{l, i}$ is the $i$-th delivery time, $\Delta_l(t)=t-S_{l, i}(\gamma_l(\delta))+b_l$ is the AoI at time $t$, and $S_{l, i+1}(\gamma_l(\delta))$ is the $(i+1)$-th sending time, given by
\begin{align}
S_{l, i+1}(\gamma_l(\delta))=\inf_{t \in \mathbb Z}\{ t \geq D_{l, i}: \gamma_l(\Delta_l(t)) \geq \gamma_l(\delta)\}.
\end{align}
\end{itemize} 
\end{theorem}
\ifreport
\begin{proof}
See Appendix \ref{ptheorem7}.
\end{proof}
\else
\fi
In Theorem \ref{theorem7}, we establish the indexability of the problem \eqref{decoupled_problem} and provide a semi-analytical expression of Whittle index $W_{l, b_l}(\delta)$ in \eqref{Whittle_Index}. Because the transmission times $T_{l, i}$ are i.i.d. for all $i$, Whittle index $W_{l, b_l}(\delta)$ in \eqref{Whittle_Index} is a function of Gittins index $\gamma_l(\delta)$. Given $\gamma_l(\delta)$, Whittle index $W_{l, b_l}(\delta)$ can be computed by Monte Carlo simulations. We also can find closed form expression of some simpler settings using \eqref{Whittle_Index}.

The following result follows from Theorem \ref{theorem7} and it fits with the result obtained in \cite{Tripathi2019}.
\begin{corollary}
If $p(\cdot)$ is a non-decreasing function and the transmission times $T_{l, i}$ is one for all $l$ and $i$, then Whittle index $W_{l, 0}(\delta)$ is 
\begin{align}\label{Whittle_Index}
W_{l, 0}(\delta)=\delta p_l(\delta+1)-\sum_{k=1}^\delta p_l(\delta).
\end{align}
\end{corollary}

\subsubsection{Scheduling Policy} 
Intuitively, Whittle index $W_{l, b_l}(\delta)$ is the measure of potential gain if the scheduler transmits $(b_l+1)$-th freshest feature from source $l$. The scheduler determines $b_l(t)$ at time slot $t$ by solving the following equation: 
\begin{align}
b_l(t)=\arg\max_{b_l \in \{0, 1, B_l-1\}} W_{l, b_l}(\Delta_l(t)).
\end{align}
Define 
\begin{align}
W_l(\delta)=\max_{b_l \in \{0, 1, B_l-1\}} W_{l, b_l}(\delta).
\end{align}
The scheduler selects source $l$ with the largest and non-negative Whittle index $W_l(\Delta_l(t))$ at time slot $t$. If $W_l(\Delta_l(t)) < 0$ for all $l$, then no source will be selected for the transmission. 

This section generalizes the Whittle index based policy in \cite{sombabu2020whittle, Tripathi2019, Kadota2018, Kadota2019} that depend heavily on the assumptions: (i) The AoI penalty functions are non-decreasing and (ii) the transmission times $T_{l, i}$ is one. }}

\section{Data Driven Evaluations}
In this section, we illustrate the performance of our scheduling policies, where the inference error function $p(\delta)$ is collected from the data driven experiments in Section \ref{Experimentation}.
\subsection{Single-source Scheduling Policies}
We evaluate the following four single-source scheduling policies:
\begin{itemize}
\item[1.] Generate-at-will, zero wait: The $(i+1)$-th feature sending time $S_{i+1}$ is given by 
$S_{i+1}=D_i = S_i+T_i$ and the feature selection policy is $f=(0, 0, \ldots)$, i.e., $b_i=0$ for all $i$. 

\item[2.] Generate-at-will, optimal scheduling: The policy is given by Theorem \ref{theorem5} with $b_i=0$ for all $i$.

\item[3.] Selection-from-buffer, optimal scheduling: The policy is given by Theorem \ref{theorem6}.

\item[4.] Periodic feature updating: Features are generated periodically with a period $T_p$ and appended to a queue with buffer size $B$. When the buffer is full, no new feature is admitted to the buffer. Features in the buffer are sent over the channel in a first-come, first-served order.  
\end{itemize}

\begin{figure}[t]
\centering
\includegraphics[width=0.35\textwidth]{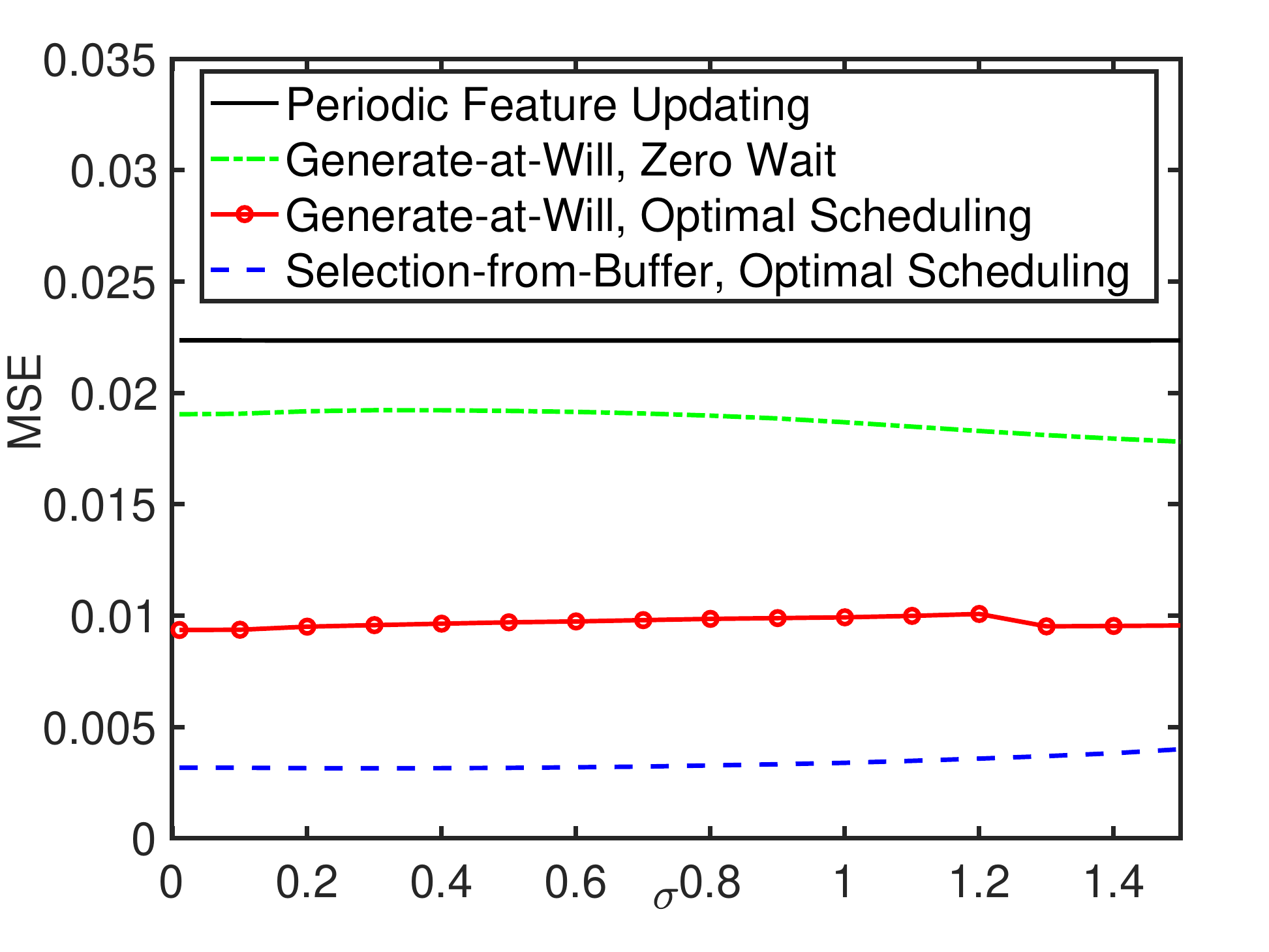}
\caption{\small Time average inference error (MSE) vs. the scale parameter $\sigma$ of discretized i.i.d. log-normal transmission time distribution for single-source scheduling (in robot state prediction task). \label{fig:singlesourcedifferentsigma}
}
\end{figure}

\ifreport

\begin{figure}[t]
\centering
\includegraphics[width=0.35\textwidth]{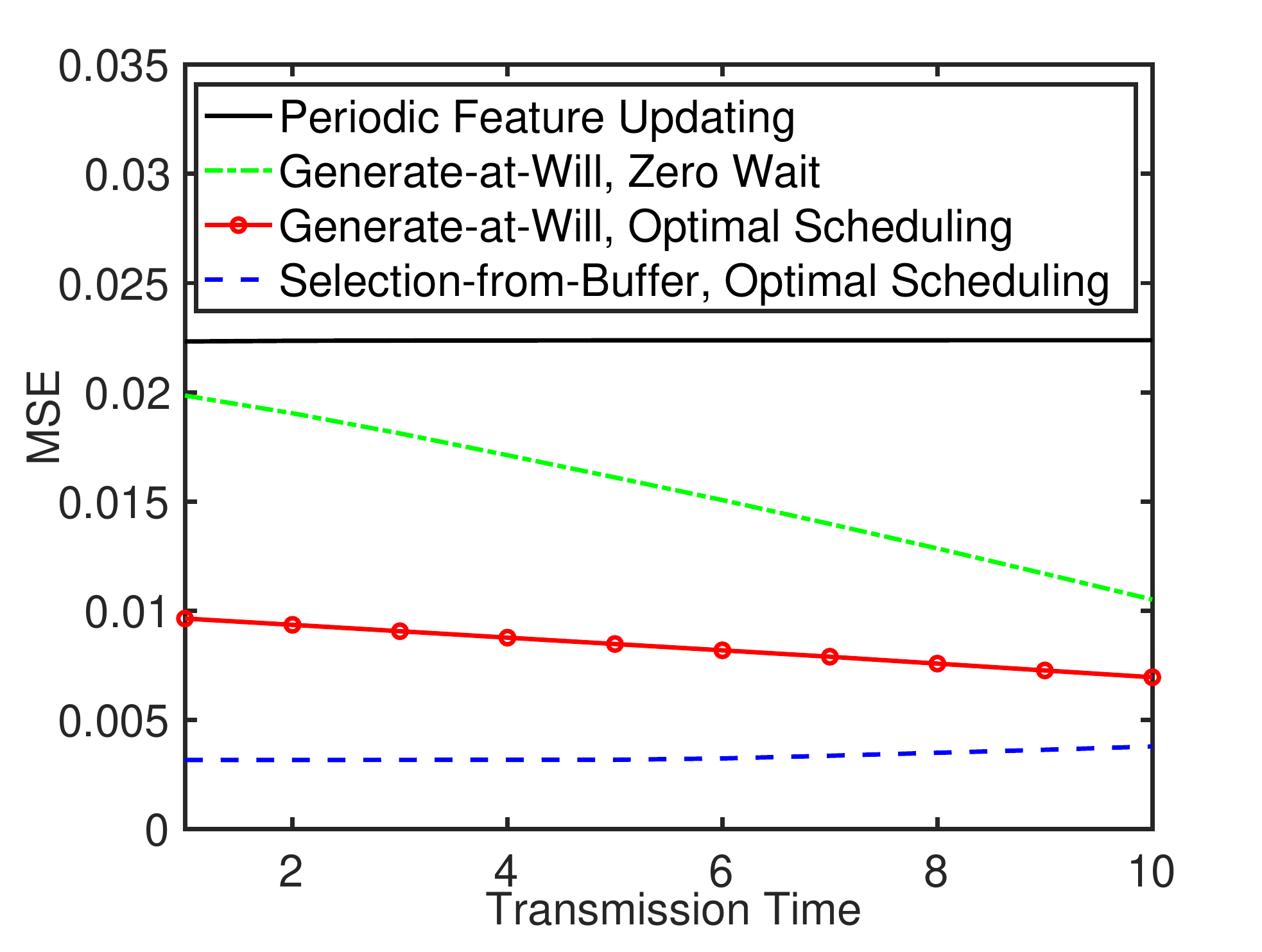}
\caption{\small Time average inference error (MSE) vs. Constant Transmission Time (in robot state prediction task). \label{fig:singlesourcedifferenttime}
}
\vspace{-3mm}
\end{figure}
\else
\fi

\begin{figure}[t]
\centering
\includegraphics[width=0.35\textwidth]{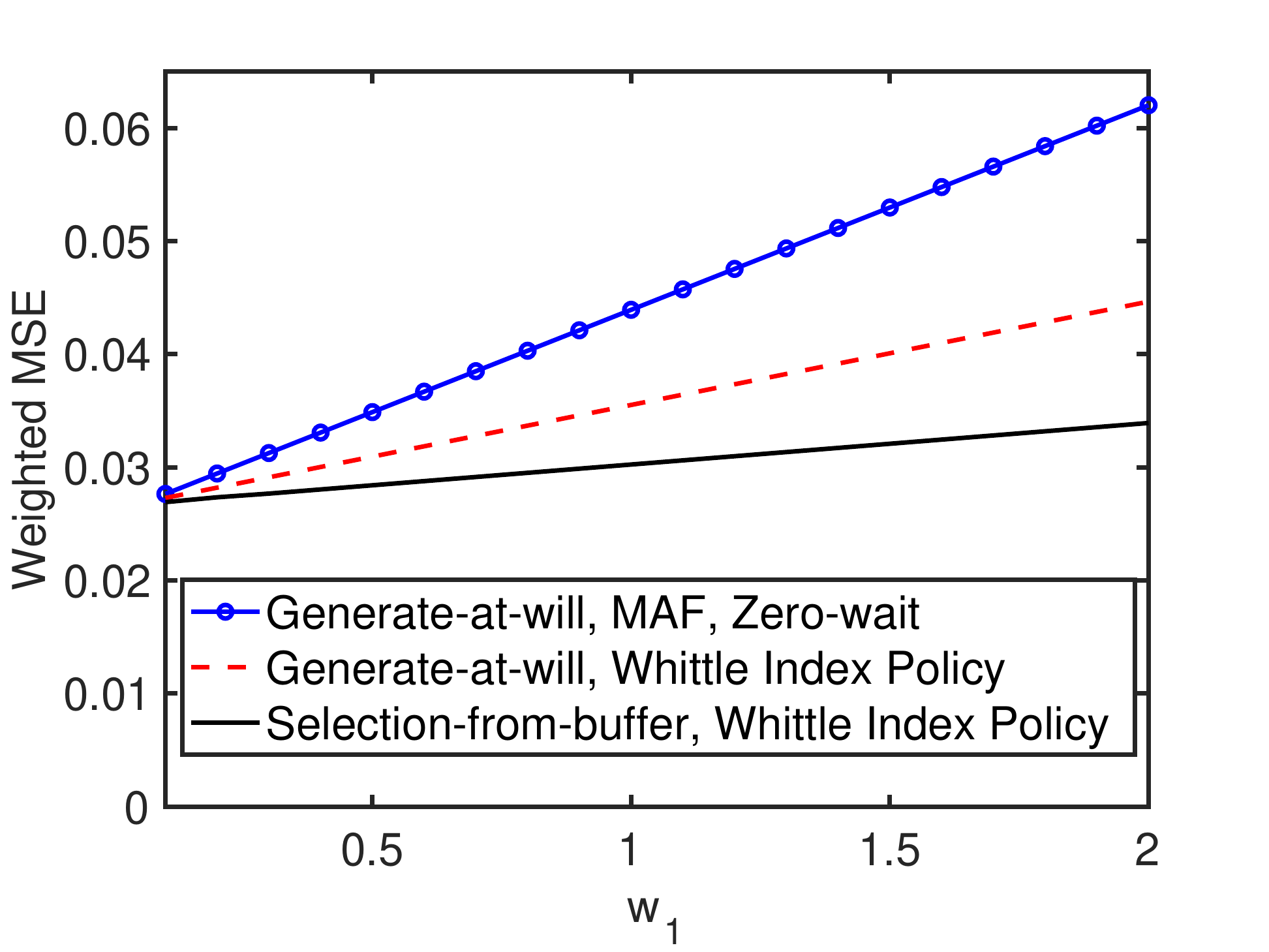}
\caption{\small Time-average weighted sum of the inference errors (Normalized MSE) vs. the weight $w_1$ of Source 1 for multi-source scheduling, where the number of sources is $m=2$ and the weight of Source 2 is $w_2=1$. \label{fig:multisourceweight}
}
\vspace{-3mm}
\end{figure}
\ifreport
\begin{figure}[t]
\centering
\includegraphics[width=0.35\textwidth]{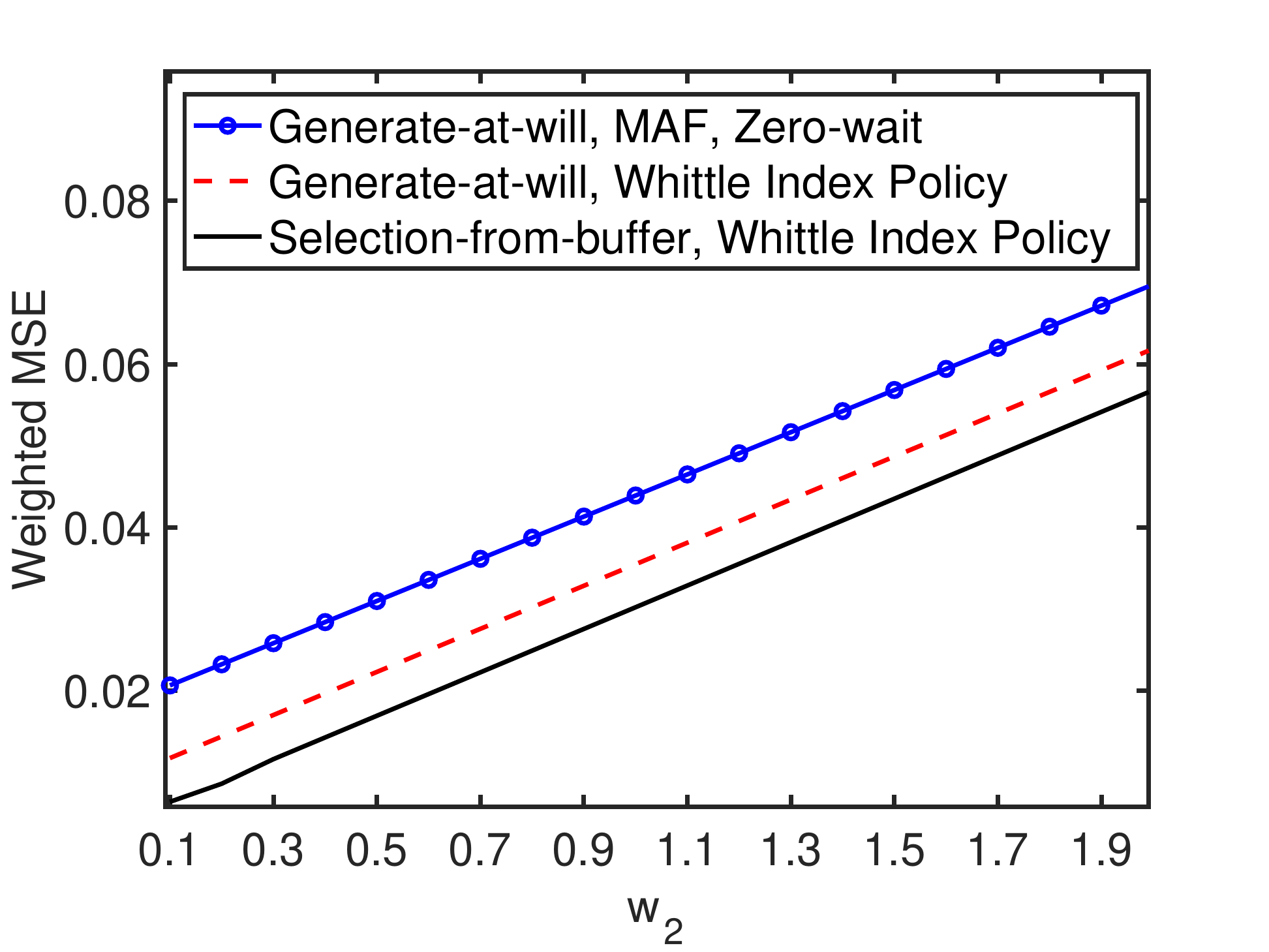}
\caption{\small Time-average weighted sum of the inference errors (MSE) vs. the weight $w_2$ of Source 2 for multi-source scheduling, where the number of sources is $m=2$ and the weight of Source 2 is $w_1=1$. \label{fig:multisourceweight1}
}
\vspace{-3mm}
\end{figure}

\begin{figure}[t]
\centering
\includegraphics[width=0.35\textwidth]{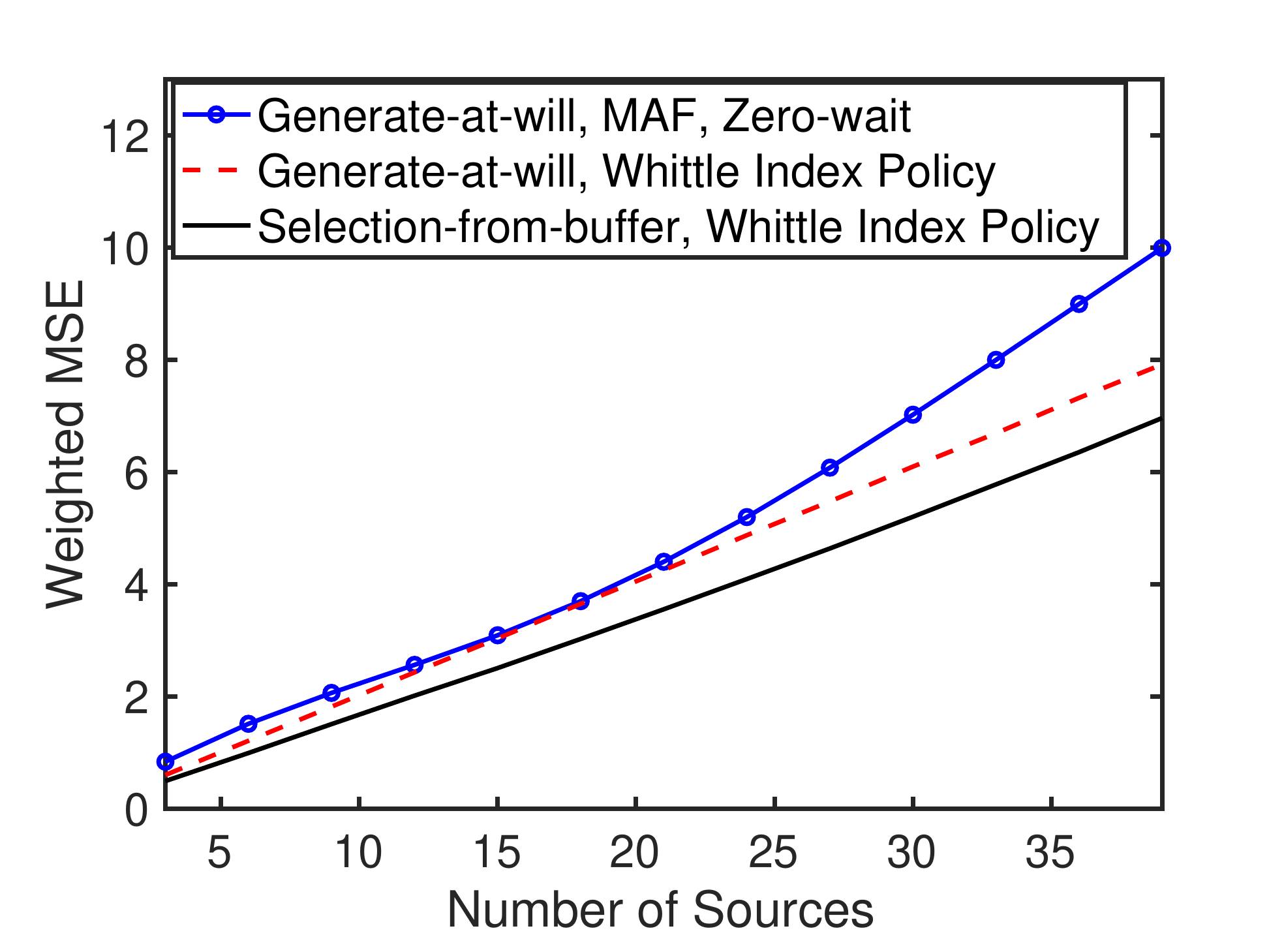}
\caption{\small Time-average weighted sum of the inference errors (MSE) vs. number of sources. \label{fig:multisourcediffnumber}
}
\vspace{-3mm}
\end{figure}

\else
\fi

Fig. \ref{fig:singlesourcedifferentsigma} illustrates the time-average inference error achieved by the four single-source scheduling policies defined above. The inference error function $p(\delta)$ used in this evaluation is illustrated in Fig. \ref{fig:DelayedNetworkedControlled}(c), which is generated by using the leader-follower robotic dataset and the trained neural network as explained in Section \ref{Experimentation}. The $i$-th feature transmission time $T_i$ is assumed to follow a discretized i.i.d. log-normal distribution. In particular, $T_i$ can be expressed as $T_i=\lceil \alpha e^{\sigma Z_i}/ \mathbb E[e^{\sigma Z_i}] \rceil$, where $Z_i$'s are i.i.d. Gaussian random variables with zero mean and unit variance. In Fig. \ref{fig:singlesourcedifferentsigma}, we plot the time average inference error versus the scale parameter $\sigma$ of discretized i.i.d. log-normal distribution, where $\alpha=1.2$, the buffer size is $B=30$, and the period of uniform sampling is $T_p=3$. The randomness of the transmission time increases with the growth of $\sigma$. Data-driven evaluations in Fig. \ref{fig:singlesourcedifferentsigma} show that ``selection-from-buffer” with optimal scheduler achieves $3$ times performance gain compared to ``generate-at-will,” and $8$ times performance gain compared to periodic feature updating.

\ifreport

Fig. \ref{fig:singlesourcedifferenttime} illustrates the performance of the four scheduling policies versus constant transmission time $T$. Similar to Fig. \ref{fig:singlesourcedifferentsigma}, the inference error function $p(\delta)$ is measured from leader-follower robotic dataset. This figure also shows that ``selection-from-buffer” with optimal scheduler can achieve $8$ time performance gain compared to periodic feature updating.

\else
\fi

\subsection{Multiple-source Scheduling Policies}

Now, we evaluate the following three multiple-source scheduling policies:
\begin{itemize}
 \item[1.] Generate-at-will, maximum age first (MAF), zero wait: At time slot $t$, if the channel is free, this policy will schedule the freshest generated feature from source $\arg\max_l \Delta_l(t)$; otherwise no source is scheduled. 
  \item[2.] Generate-at-will, Whittle index policy: Denote 
  \begin{align}
  W_0(t)=\max_l \!W_{l, 0}(\Delta_l(t)), ~  l_0^*=\arg\max_l \!W_{l, 0}(\Delta_l(t)). 
  \end{align}
If the channel is free and $W_0(t) \geq 0$, the freshest feature of the source $l_0^*$ is scheduled; otherwise no source is scheduled.  
 \item[3.] Selection-from-buffer, Whittle index policy: The policy is described in Algorithm \ref{alg:Whittle}.
\end{itemize}

In Fig. \ref{fig:multisourceweight}, we plot the time average weighted sum of inference errors versus weight $w_1$, where the number of sources is $m=2$ and weight $w_2=1$. The inference error function $p_1(\delta)$ is illustrated in Fig. \ref{fig:DelayedNetworkedControlled}(c). The inference error function $p_2(\delta)$ is illustrated in Fig. \ref{fig:learning}(c), which is generated by using the pre-trained neural network on ``BAIR" dataset from \cite{lee2018stochastic}. The transmission times for Source 1 and Source 2 are $T_{1, i}=1$ and $T_{2, i}=4$ for all $i$, respectively. The buffer sizes are $B_1=B_2=30$. The weight $w_1$ is associated with a non-monotonic AoI function. The performance gain of ``selection-from-buffer, Whittle index policy'' increases as $w_1$ grows. \ifreport  \else Due to limited space, more numerical results are provided in \cite{technical_report}.\fi

\ifreport
Fig. \ref{fig:multisourceweight1} shows the time average weighted sum of inference errors versus weight $w_2$, where the weight $w_1=1$. The other parameters are the same as in Fig. \ref{fig:multisourceweight}. The weight $w_2$ is associated with a monotonic AoI function. The difference among the average weighted sum of inference errors under policies ``Selection-from-buffer, Whittle index policy'', ``Generate-at-will, Whittle index policy'', and ``Generate-at-will, MAF, zero wait'' is fixed as $w_2$ grows, where ``Selection-from-buffer, Whittle index policy'' achieves the minimum inference errors.

Fig. \ref{fig:multisourcediffnumber} depicts the performance of the three scheduling policies as the number of sources $n$ increases. The number of sources is increased from $n=3$ to $n=39$. The number of sources $n$ increments by $3$ in which inference error functions are associated with Fig. \ref{fig:learning}(c), Fig. \ref{fig:TrainingCartVelocity}(c), and Fig. \ref{fig:DelayedNetworkedControlled}(c) with constant transmission times $4, 1,$ and $1$, respectively. From Fig. \ref{fig:multisourcediffnumber}, we observe that ``Selection-from-buffer, Whittle index policy'' achieves minimum inference error than the other two policies. 
\else
\fi

\section{Conclusions}
In this paper, we interpreted the impact of data freshness on the performance of real-time supervised learning. We showed that the training error and the inference error of real-time supervised learning could be non-monotonic AoI functions if the target and feature data sequence is far from a Markov model. Our experimental results suggested that the data sequence can be far from Markovian due to response delay, communication delay, and/or long-range dependence. To minimize the time-average inference error, we adopted a new feature transmission model called ``selection-from-buffer'' and designed an optimal single-source scheduling policy. The optimal single-source scheduling policy is found to be a threshold policy on the Gittins index. Moreover, we developed a Whittle index policy for multiple-source scheduling and provided a semi-analytical expression for the Whittle index. Our numerical results validated the efficacy of the proposed scheduling policies. 

\section*{Acknowledgement}
The authors are grateful to Vijay Subramanian for one suggestion, to John Hung for useful discussions on this work, and to Shaoyi Li for his help on Fig. \ref{fig:learning}(b)-(c).

\bibliographystyle{IEEEtran}
\bibliography{refshisher}
\section{Appendix}
\subsection{Relationship among $L$-divergence, Bregman divergence, and $f$-divergence} \label{InformationTheory2}
We provide a comparison among the $L$-divergence defined in \eqref{divergence}, the Bregman divergence \cite{dhillon2008matrix}, and the $f$-divergence \cite{csiszar2004information}.

Let $\mathcal{P^Y}$ denote the set of all probability distributions on the discrete set $\mathcal Y$. Any distribution $Q_Y \in \mathcal{P^Y}$ can be represented by a probability vector $\mathbf{q}_Y=(Q_Y(y_1), \ldots, Q_Y(y_{|\mathcal Y|}))^T$ that satisfies $\sum_{y\in \mathcal Y}Q_Y(y) = 1$ and $Q_Y(y) \geq 0$ for all $y \in \mathcal Y$. If $F: \mathcal {P^Y} \mapsto \mathbb R$ be a continuously differentiable and strictly convex function, then the Bregman divergence $B_F(P_Y||Q_Y)$ between two distributions $P_Y \in \mathcal {P^Y}$ and $Q_Y \in \mathcal {P^Y}$ associated with function $F$ is defined by \cite{Amari}
\begin{align}\label{Bregaman}
    B_{F}(P_Y||Q_Y)= F(\mathbf p_Y)-F(\mathbf q_Y)-\nabla F(\mathbf q_Y)^T (\mathbf p_Y-\mathbf q_Y),
\end{align}
where $\mathbf p_Y$ and $\mathbf q_Y$ are two probability vectors associated to the distributions $P_Y$ and $Q_Y$, respectively, and $\nabla F(\mathbf q_Y)$ is the gradient of function $F$ at $\mathbf q_Y$. Consider the loss function
\begin{align}\label{eq-L_F}
L_F(y, Q_Y)= - F(\mathbf q_Y)-\frac{\partial F(\mathbf q_Y)}{\partial P_Y(y)}+\nabla F(\mathbf q_Y)^T \mathbf q_Y,
\end{align}
where the action $a=Q_Y$ is a distribution in $\mathcal{P^Y}$.
\begin{lemma}\label{B-L}
Any Bregman divergence $B_{F}(P_Y||Q_Y)$ is an $L_F$-divergence $D_{L_F}(P_Y || Q_Y)$, where $L_F$ is defined in \eqref{eq-L_F}.  
\end{lemma}
\begin{proof}
The $L_F$-entropy associated with the loss function $L_F(y, Q_Y)$ in \eqref{eq-L_F} is 
\begin{align}\label{entropy_F}
H_{L_F}(Y)= \min_{Q_Y \in \mathcal{P^Y}} E_{Y \sim P_Y} \left [ L_F\left (Y, Q_Y\right)\right],
\end{align}
where $P_Y$ is the distribution of $Y$ and 
\begin{align}\label{expected_loss_F}
E_{Y \sim P_Y} \left [ L_F\left (Y, Q_Y\right)\right]=-F(\mathbf q_Y)-\nabla F(\mathbf q_Y)^T (\mathbf p_Y-\mathbf q_Y).
\end{align}
Because the function $F$ is convex, it follows from \eqref{expected_loss_F} that
\begin{align}
E_{Y \sim P_Y} \left [ L_F\left (Y, Q_Y\right)\right] \geq -F(\mathbf p_Y), \ \forall Q_Y \in \mathcal{P^Y}.
\end{align}
Moreover, if $Q_Y=P_Y$, then 
\begin{align}\label{lowerboundachieved}
E_{Y \sim P_Y} \left [ L_F\left (Y, P_Y\right)\right]=-F(\mathbf p_Y).
\end{align}
Combining \eqref{entropy_F}-\eqref{lowerboundachieved}, it follows that
\begin{align}
H_{L_F}(Y)=& E_{Y \sim P_Y} \left [ L_F\left (Y, P_Y\right)\right] ]=-F(\mathbf p_Y).
\end{align}
Due to the strict convexity of function $F$, $Q_Y=P_Y$ is the unique minimizer of \eqref{entropy_F}. Then, by the definition of $L$-divergence in \eqref{divergence}, we get
\begin{align}
&D_{L_F}(P_Y || Q_Y)\nonumber\\
=&E_{Y \sim P_Y} \left [ L_F\left (Y, Q_Y\right)\right]-H_{L_F}(Y) \nonumber\\
=&-F(\mathbf q_Y)-\nabla F(\mathbf q_Y)^T (\mathbf p_Y-\mathbf q_Y)+F(\mathbf p_Y),
\end{align}
which is equal to the Bregman divergence $B_{F}(P_Y||Q_Y)$ defined in \eqref{Bregaman}. This completes the proof.
\end{proof}
 

By Lemma \ref{B-L}, any Bregman divergence $B_{F}(P_Y||Q_Y)$ is an $L_F$-divergence $D_{L_F}(P_Y || Q_Y)$. However, the converse is not always true, which is explained below. If $H_{L}(Y)$ is strictly concave and continuously differentiable in $P_Y$, then the associated $L$-divergence $D_L(P_Y||Q_Y)$ can be expressed as \cite[Section 3.5.4]{Dawid2004}
\begin{align}\label{L-divergencetoBregman}
    &D_L(P_Y||Q_Y) \nonumber\\
    =&H_L(\mathbf q_Y)+\nabla H_L(\mathbf q_Y)^T (\mathbf p_Y-\mathbf q_Y)-H_L(\mathbf p_Y),
\end{align}
where the $L$-entropy $H_L(Y)$ is rewritten as $H_L(\mathbf q_Y)$ to emphasize that it is a function of vector $\mathbf q_Y$. By comparing \eqref{L-divergencetoBregman} with \eqref{Bregaman}, one can observe that the right hand side of \eqref{L-divergencetoBregman} is exactly the Bregman divergence $B_{-H_L} (P_Y||Q_Y)$ associated with function $-H_L$. If $H_{L}(Y)$ is not strictly concave or not continuously differentiable in $P_Y$, then the $L$-divergence $D_L(P_Y||Q_Y)$ may not be a Bregman divergence.

The $f$-divergence is defined by \cite{csiszar2004information} 
\begin{align}
   D_f(P_Y||Q_Y)=\sum_{y \in \mathcal Y} Q_Y(y) f\left(\frac{P_Y(y)}{Q(y)}\right),  
\end{align}
where $f:(0:\infty) \mapsto \mathbb R$ is a convex function satisfying $f(1)=0$. The $f$-mutual information can be expressed by using the $f$-divergence
\begin{align}
    I_f(Y; X)=\mathbb E_{X \sim P_X}[D_f(P_{Y|X}||P_Y)].
\end{align}
The $f$-mutual information is symmetric, i.e., $I_f(Y; X)=I_f(X;Y).$ However, the $L$-mutual information is non-symmetric in general, except for some special cases.  For example, Shannon's mutual information is defined by 
\begin{align}
 I_{\text{log}}(Y; X)=\mathbb E_{X \sim P_X}[D_{\text{log}}(P_{Y|X}||P_Y)],
\end{align}
where $D_{\text{log}}(P_{Y}||Q_Y)$ is the K-L divergence \cite{cover1999elements}. It is well-known that $I_{\log} (Y;X) = I_{\log} (X;Y)$. An $f$-divergence may not be $L$-divergence and an $L$ divergence may not be $f$-divergence. In fact, K-L divergence $D_{\text{log}}(P_{Y}||Q_Y)$ and its dual $D_{\text{log}}(Q_{Y}||P_Y)$ are unique divergences that belong to $f$-divergence and Bregman divergence \cite{Amari}. Hence, $D_{\text{log}}(P_{Y}||Q_Y)$ and $D_{\text{log}}(Q_{Y}||P_Y)$ are also the only divergences belonging to all the three classes of divergences. 

\subsection{Examples of Loss function $L$, $L$-entropy, and $L$-cross entropy}\label{InformationTheory1}
Several examples of the loss function $L$, $L$-entropy, and $L$-cross entropy are listed below. More examples can be found in \cite{Dawid2004, Dawid1998,farnia2016minimax}.
\subsubsection{Logarithmic~Loss (log-loss)} The log-loss function is given by $L_{\text{log}}(y,Q_Y) = -\log Q_Y(y)$, where the action $a=Q_Y$ is a distribution  in $\mathcal{P}^{\mathcal{Y}}$. The corresponding $L$-entropy is the well-known Shannon's entropy \cite{cover1999elements}, 
\begin{align}
H_{\text{log}}(Y) =& -\sum_{y \in \mathcal{Y}} P_Y(y)\ \text{log}\ P_Y(y), 
\end{align}
where $P_Y$ is the distribution of $Y$. The corresponding $L$-cross entropy is 
\begin{align}
    H_{\text{log}}(Y; \tilde Y)=&-\sum_{y \in \mathcal{Y}} P_{Y}(y)\ \text{log}\ P_{\tilde Y}(y).
\end{align}
The $L$-mutual information and $L$-divergence associated to the log-loss are Shannon's mutual information and the K-L divergence, respectively. 
\subsubsection{Brier~Loss} The Brier loss function is defined as $L_B(y, Q_Y)=\sum_{y' \in \mathcal{Y}} Q_Y(y')^2-2 \ Q_Y(y)+1$ \cite{Dawid2004}. The associated $L$-entropy is
\begin{align}
H_B(Y)=&1-\sum_{y \in \mathcal{Y}} P_Y(y)^2,
\end{align}
and the associated $L$-cross entropy is 
\begin{align}
    H_B(Y; \tilde Y)=&\sum_{y \in \mathcal Y} P_{\tilde Y}(y)^2-2 \sum_{y \in \mathcal Y}P_{\tilde Y}(y)P_{Y}(y)+1.
\end{align}
\subsubsection{0-1~Loss} The 0-1 loss function is given by $L_{\text{0-1}}(y, \hat{y}) = \mathbf{1}(y \neq \hat{y})$, where $\mathbf{1}(A)$ is the indicator function of event $A$. For this case, we have
\begin{align}
H_{\text{0-1}}(Y)~=&~ 1- \max_{y \in \mathcal{Y}} P_Y(y), \\
 H_{\text{0-1}}(Y; \tilde Y)~=&~ 1-P_{Y}\left(\arg\max_{y \in \mathcal{Y}} P_{\tilde Y}(y)\right).
\end{align}

\subsubsection{$\alpha$-Loss} The $\alpha$-loss function is defined by $L_{\alpha}(y, Q_Y)=\frac{\alpha}{\alpha-1} \left[1 - Q_Y(y)^{\frac{\alpha-1}{\alpha}}\right]$ for $\alpha>0$ and $\alpha \neq 1$  \cite[Eq. 14]{alpha}. It becomes the log-loss function at the limit $\alpha \rightarrow 1$ and the 0-1 loss function as the limit $\alpha \rightarrow \infty$. The $L$-entropy and $L$-cross entropy associated to the $\alpha$-loss function are given by 
\begin{align}
H_{\alpha \text{-loss}}(Y)= \frac{\alpha}{\alpha-1} \left[1 - \left(\sum_{y \in \mathcal{Y}} P_Y(y)^{\alpha}\right)^{\frac{1}{\alpha}}\right],
\end{align}
\begin{align}
   H_{\alpha \text{-loss}}(Y; \tilde Y)=\frac{\alpha}{\alpha-1} \left[1 - \left(\sum_{y \in \mathcal{Y}} P_{\tilde Y}(y)^{\alpha}\right)^{\frac{1}{\alpha}}\lambda\right],
\end{align}
where 
\begin{align}
    \lambda=\frac{\sum_{y \in \mathcal Y} \frac{P_{Y}(y)}{P_{\tilde Y}(y)}P_{\tilde Y}(y)^{\alpha}}{\sum_{y \in \mathcal Y} P_{\tilde Y}(y)^{\alpha}}.
\end{align}

\subsubsection{Quadratic~Loss} The quadratic loss function is  $L_2(y, \hat{y})= (y - \hat{y})^2$. The $L$-entropy function associated with the quadratic loss is the variance of $Y$, i.e., \begin{align}
H_2(Y)= \mathbb{E}[Y^2] - \mathbb{E} [Y]^2.
\end{align}
The corresponding $L$-cross entropy is  
\begin{align}
     H_{2}(Y;\tilde Y)=\mathbb E[ Y^2]-2\mathbb E[\tilde Y] \mathbb E[Y]+\mathbb E[\tilde Y]^2.
\end{align}

\subsection{Proof of Equation \eqref{freshness_aware_cond}}\label{pfreshness_aware_cond}
From \eqref{eq_TrainingErrorLB1}, we get 
\begin{align}\label{L_condEntropy1}
&H_L(\tilde Y_0| \tilde X_{-\Theta},\Theta)\nonumber\\
=&\sum_{\substack{x \in \mathcal X, \theta \in \mathcal D}} \!\!\!\! P_{\tilde X_{-\Theta}, \Theta}(x, \theta) H_L(\tilde Y_0|\tilde X_{-\Theta}=x, \Theta=\theta)\nonumber\\
\!\!\!\!=&\!\!\!\sum_{\theta \in \mathcal D} P_{\Theta}(\theta) \sum_{x \in \mathcal X} P_{\tilde X_{-\Theta}|\Theta=\theta}(x) H_L(\tilde Y_0|\tilde X_{-\Theta}=x, \Theta=\theta)\!\!\! \nonumber\\
\!\!\!=&\!\!\!\sum_{\theta \in \mathcal D} P_{\Theta}(\theta) \sum_{x \in \mathcal X} P_{\tilde X_{-\theta}|\Theta=\theta}(x) H_L(\tilde Y_0|\tilde X_{-\Theta}=x, \Theta=\theta).\!\!\!\!
\end{align}

Next, from \eqref{given_L_condentropy}, we obtain that for all $x \in \mathcal X$ and $\theta \in \mathcal D$
\begin{align}\label{given_L_condentropy2}
\!\!\!H_L(\tilde Y_0| \tilde X_{-\Theta}=x,\Theta=\theta)\!\!=& \min_{a\in\mathcal A}\! \mathbb E_{Y\sim P_{\tilde Y_0|\tilde X_{-\Theta}=x, \Theta=\theta}} [L(Y, a)]\!\!\! \nonumber\\
\!\!=& \min_{a\in\mathcal A}\! \mathbb E_{Y\sim P_{\tilde Y_0|\tilde X_{-\theta}=x, \Theta=\theta}} [L(Y, a)].\!\!\!
\end{align}
Because $\tilde Y_0$ and $\tilde X_{-\theta}$ are independent of $\Theta$, for all $x \in \mathcal X, y \in \mathcal Y$, and all $\theta =0, 1, \ldots$
\begin{align}\label{ind_X}
P_{\tilde Y_0|\tilde X_{-\theta}=x, \Theta=\theta}(y)=P_{\tilde Y_0| \tilde X_{-\theta}=x}(y).
\end{align}
Hence, \eqref{given_L_condentropy2} can be simplified as
\begin{align}\label{given_L_condentropy1}
H_L(\tilde Y_0| \tilde X_{-\Theta}=x,\Theta=\theta)=& \min_{a\in\mathcal A}\! \mathbb E_{Y\sim P_{\tilde Y_0|\tilde X_{-\theta}=x}} [L(Y, a)]\nonumber\\
=&H_L(\tilde Y_0| \tilde X_{-\theta}=x).
\end{align}
Substituting \eqref{given_L_condentropy1} and \eqref{ind_X} into \eqref{L_condEntropy1}, we observe that 
\begin{align}
H_L(\tilde Y_0| \tilde X_{-\Theta},\Theta)=&\sum_{\theta \in \mathcal D} P_{\Theta}(\theta) \sum_{x \in \mathcal X} P_{\tilde X_{-\theta}}(x) H_L(\tilde Y_0|\tilde X_{-\theta}=x)\nonumber\\
=&\sum_{\theta \in \mathcal D} P_{\Theta}(\theta)~H_L(\tilde Y_0| \tilde X_{-\theta}).
\end{align}

\subsection{Proof of Equation \eqref{eq_inferenceerror}}\label{peq_inferenceerror}


The expected inference error in time slots $0, 1, \ldots, T-1$ is expressed as 
\begin{align}\label{eq_inferenceerrorproof}
\!\!\!\!\!&\mathrm{err}_{\mathrm{inference}}(T) \nonumber\\
=&\frac{1}{T}\sum_{t=1}^T \!\mathbb{E}_{Y, X, \Delta \sim P_{Y_t, X_{t- \Delta(t)}, \Delta(t)}}\!\!\left[ \!L\left(Y,\phi^*_{P_{\tilde Y_0, \tilde X_{-\Theta},\Theta}}(X,\Delta)\right)\right], \nonumber\\
=&\frac{1}{T}\sum_{t=1}^T\bigg(\sum_{\delta \in \mathcal D} P(\Delta(t)=\delta)\nonumber\\
&\times \mathbb{E}_{Y, X \sim P_{Y_t, X_{t- \Delta(t)} | \Delta(t)=\delta}}\!\!\left[ \!L\left(Y,\phi^*_{P_{\tilde Y_0, \tilde X_{-\Theta},\Theta}}(X,\delta)\right)\right]\bigg),
\end{align} 
where the empirical distribution $P(\Delta(t)=\delta)$ is equal to $\mathbf 1(\Delta(t)=\delta)$, which is an indicator function.

Because $Y_t$ and $X_{t-\delta}$ are independent of $\Delta(t)$, for all $x \in \mathcal X, y \in \mathcal Y$, and all $\delta \in \mathcal D$, we have 
\begin{align}
P_{Y_t, X_{t- \Delta(t)} | \Delta(t)=\delta}(y, x)=&P_{Y_t, X_{t- \delta} | \Delta(t)=\delta}(y, x) \nonumber\\
=&P_{Y_t, X_{t- \delta}}(y, x).
\end{align}
Hence,
\begin{align}\label{e11-2}
&\mathbb{E}_{Y, X \sim P_{Y_t, X_{t- \Delta(t)} | \Delta(t)=\delta}}\!\!\left[ \!L\left(Y,\phi^*_{P_{\tilde Y_0, \tilde X_{-\Theta},\Theta}}(X,\delta)\right)\right]\nonumber\\
=&\mathbb{E}_{Y, X \sim P_{Y_t, X_{t- \delta} | \Delta(t)=\delta}}\!\!\left[ \!L\left(Y,\phi^*_{P_{\tilde Y_0, \tilde X_{-\Theta},\Theta}}(X,\delta)\right)\right]\nonumber\\
=&\mathbb{E}_{Y, X \sim P_{Y_t, X_{t- \delta}}}\!\!\left[ \!L\left(Y,\phi^*_{P_{\tilde Y_0, \tilde X_{-\Theta},\Theta}}(X,\delta)\right)\right].
\end{align}

Now, substituting \eqref{e11-2} into \eqref{eq_inferenceerrorproof}, we obtain \eqref{eq_inferenceerror}.


\ignore{\subsection{Proof of Equation \eqref{Decomposed_Cross_entropy}}\label{pDecomposed_Cross_entropy}
Because $Y_t$ and $X_{t-\delta}$ are independent of $\Theta$, from \eqref{L_condEntropy} we get
\begin{align}
\hat \phi_{P_{Y_t, X_{t-\Theta},\Theta}}(x,\delta)=&\arg\min_{\phi(x, \delta)\in\mathcal A}\!\! E_{Y\sim P_{Y_t|X_{t-\Theta}=x, \Theta=\delta}}[L(Y,\phi(x, \delta))] \nonumber\\
=&\arg\min_{a\in\mathcal A}\!\!~E_{Y\sim P_{Y_t|X_{t-\Theta}=x, \Theta=\delta}}[L(Y,a)] \nonumber\\
=&\arg\min_{a\in\mathcal A}\!\!~E_{Y\sim P_{Y_t|X_{t-\delta}=x}}[L(Y,a)].
\end{align}
From the above equation, we can say that $\hat \phi_{P_{Y_t, X_{t-\Theta},\Theta}}(x,\delta)$ is a Bayes action $a_{P_{Y_t|X_{t-\delta}=x}}$. 

Next, since $\tilde Y_t$ and $\tilde X_{t-\delta}$ are independent of $\Delta$, we obtain 
\begin{align}
&H_L(\tilde Y_{t}; Y_{t} | \tilde X_{t-\Delta}, \Delta)\nonumber\\
=&\sum_{y \in \mathcal Y, x \in \mathcal X, \delta \in \mathcal D} P_{\tilde Y_t, \tilde X_{t-\Delta}, \Delta}(y, x, \delta)~L\left(y, a_{P_{Y_t|X_{t-\delta}=x}}\right)\nonumber\\
=&\sum_{\delta \in \mathcal D} P_{\Delta}(\delta) \sum_{\substack{y \in \mathcal Y, x \in \mathcal X}}P_{\tilde Y_t, \tilde X_{t-\Delta}| \Delta=\delta}(y, x)~L\left(y, a_{P_{Y_t|X_{t-\delta}=x}}\right)\nonumber\\
=&\sum_{\delta \in \mathcal D} P_{\Delta}(\delta)\sum_{\substack{y \in \mathcal Y, x \in \mathcal X}}P_{\tilde Y_t, \tilde X_{t-\delta}}(y, x)~L\left(y, a_{P_{Y_t|X_{t-\delta}=x}}\right)\nonumber\\
=&\sum_{\delta \in \mathcal D} P_{\Delta}(\delta)~H_L(\tilde Y_{t}; Y_{t} | \tilde X_{t-\delta}).
\end{align}}
\subsection{Proof of Lemma \ref{Symmetric}}\label{PSymmetric}
By using the definition of $\chi^2$-conditional mutual information in \eqref{epsilon-Markov-def} and \cite{polyanskiy2014lecture}, we can get 
\begin{align}\label{proof3.21}
&I_{\chi^2}(Z;Y|X)\nonumber\\
&=\sum_{x \mathcal X, y \in \mathcal Y} P_{X, Y} (x, y) \sum_{z \in \mathcal Z} \frac{\left(P_{Z| X, Y}(z| x, y) - P_{Z|X}(z|x)\right)^2}{P_{Z|X}(z|x)}\nonumber\\
&=\!\!\!\!\!\!\!\!\sum_{x \in \mathcal X, y \in \mathcal Y, z \in \mathcal Z}\!\!\!\!\!P_{X}(x)\frac{\left(P_{Z, Y| X}(z, y|x) \!\!- \!P_{Z|X}(z|x) P_{Y|X}(y|x)\right)^2}{P_{Z|X}(z|x)P_{Y|X}(y|x)}\nonumber\\
&=\!\!\!\sum_{x \in \mathcal X, z \in \mathcal Z}\!\!\!P_{X, Z}(x, z)\sum_{y \in \mathcal Y} \frac{\left(P_{Y| X, Z}(y|x, z) - P_{Y|X}(y|x)\right)^2}{P_{Y|X}(y|x)}\nonumber\\
&=I_{\chi^2}(Y;Z|X).
\end{align}
If $Z \overset{\epsilon} \rightarrow X \overset{\epsilon} \rightarrow Y$, then by the definition of $\epsilon$-Markov chain in Section \ref{SecMinTrainingError} and \eqref{proof3.21}, 
\begin{align}
I_{\chi^2}(Z;Y|X)= I_{\chi^2}(Y;Z|X) \leq \epsilon^2.
\end{align}
Hence, we get $Y \overset{\epsilon}  \rightarrow X \overset{\epsilon} \rightarrow Z$. This completes the proof.

\subsection{Proof of Lemma \ref{Lemma_CMI}}\label{PLemma_CMI}
By the definition of $L$-conditional mutual information in \eqref{CMI}, we obtain
\begin{align}\label{cond_entropy_two}
    H_L(Y|X, Z)=&H_L(Y|X)-I_L(Y;Z|X) \nonumber\\
    =&H_L(Y|Z)-I_L(Y;X|Z).
\end{align}
From \eqref{CMI} and \eqref{cond_entropy_two}, we get
\begin{align}
    H_L(Y|X)=&H_L(Y|Z)+I_L(Y;Z|X)-I_L(Y;X|Z) \nonumber\\
    \leq& H_L(Y|Z)+I_L(Y;Z|X),
\end{align}
where the last inequality is due to $I_L(Y;X|Z) \geq 0$. Now, it remains to show that if $Y \overset{\epsilon}\leftrightarrow X \overset{\epsilon}\leftrightarrow  Z$, then
\begin{align}
  I_L(Y;Z|X)=O(\epsilon);
\end{align}
in addition, if $H_L(Y)$ is twice differentiable, then
\begin{align}
  I_L(Y;Z|X)=O(\epsilon^2).
\end{align}

By using the definition of $L$-conditional mutual information from \eqref{CMI}, we see that
\begin{align}\label{CMI1}
I_L(Y ; Z | X) =& \mathbb{E}_{X, Z} [D_L(P_{Y|X,Z} || P_{Y|X})].
\end{align}
If $Y \overset{\epsilon}\leftrightarrow X \overset{\epsilon}\leftrightarrow  Z$ is an $\epsilon$-Markov chain, then 
\begin{align}\label{eMarkovcondition}
  \sum_{(x,z) \in \mathcal X \times \mathcal Z} P_{X,Z}(x,z) D_{\chi^2}(P_{Y|X=x,Z=z} || P_{Y|X=x}) \leq \epsilon^2.
\end{align}
Because the left side of the above inequality is the summation of non-negative terms, the following holds
\begin{align}
  P_{X,Z}(x,z) D_{\chi^2}(P_{Y|X=x,Z=z} || P_{Y|X=x}) \leq \epsilon^2,
\end{align}
for all $(x, z) \in \mathcal X \times \mathcal Z$. 

If $P_{X, Z}(x,z) >0$, then
\begin{align}\label{Oepsilon}
  D_{\chi^2}(P_{Y|X=x,Z=z} || P_{Y|X=x}) \leq \dfrac{\epsilon^2}{P_{X,Z}(x,z)}.
\end{align}

Next, we need the following lemma.
\ignore{\begin{definition}[$\beta$-\textbf{neighborhood}]\cite{huang2019universal}
Given $\beta \geq 0$, the $\beta$- neighborhood of a reference distribution $Q_Y \in relint(\mathcal{P^Y})$ is the set of distributions that are in a Neyman's $\chi^2$-divergence ball of radius $\beta^2$ centered on $Q_Y,$ i.e.,
\begin{align}
\mathcal{N^Y_\beta}(Q_Y) = \left\{ P_Y \in \mathcal{P^Y} : D_{\chi^2}(P_Y || Q_Y) \leq \beta^2\right\}.
\end{align}
\end{definition}}
\begin{lemma}\label{divergenceL}
The following assertions are true:
\begin{itemize}
\item[(a)] If two distributions $Q_Y, P_Y \in \mathcal{P^Y}$ satisfy 
\begin{align}\label{condition_divergenceL}
D_{\chi^2}(P_Y || Q_Y) \leq \beta^2,
\end{align}
then 
\begin{align}
D_L(P_Y || Q_Y)=O(\beta).
\end{align}
\item[(b)] If, in addition, $H_L(Y)$ is twice differentiable in $P_Y$, then
\begin{align}
D_L(P_Y || Q_Y)=O(\beta^2).
\end{align}
\end{itemize}
\end{lemma}
\begin{proof}
See in Appendix \ref{PdivergenceL}.
\end{proof}

Define set $\hat{\mathcal X} \times \hat{\mathcal Z} =\{(x, z) \in \mathcal X \times \mathcal Z : P_{X, Z}(x, z) >0\}$. Then, using \eqref{Oepsilon} and Lemma \ref{divergenceL}(a) in \eqref{CMI1}, we obtain
\begin{align}\label{lemma1result1}
   &I_L(Y;Z|X) \nonumber\\
   =&\sum_{(x,z) \in \mathcal X \times \mathcal Z} P_{X,Z}(x,z)~D_{L}(P_{Y|X=x,Z=z} || P_{Y|X=x}) \nonumber\\
   =&\sum_{(x,z) \in \hat{\mathcal X} \times \hat{\mathcal Z}} P_{X,Z}(x,z)~D_{L}(P_{Y|X=x,Z=z} || P_{Y|X=x}) \nonumber\\
   =&\sum_{(x,z) \in \hat{\mathcal X} \times \hat{\mathcal Z}} P_{X, Z}(x,z)~O\left(\frac{\epsilon}{\sqrt{P_{X, Z}(x,z)}}\right) \nonumber\\
   =&~O(\epsilon).
 \end{align}
Similarly, when $H_L(Y)$ is differentiable in $P_Y$, by using Lemma \ref{divergenceL}(b) we obtain 
\begin{align}\label{lemma1result2}
    I_L(Y;Z|X)=O(\epsilon^2).
\end{align}
This completes the proof of Lemma \ref{Lemma_CMI}.

\subsection{Proof of Theorem \ref{theorem1}}\label{Ptheorem1}
By using the definition of the $L$-conditional mutual information \eqref{CMI}, we can show that
\begin{align}
&H_L(\tilde Y_0 | \tilde X_{-k}, \tilde X_{-k-1}) \nonumber\\
=&H_L(\tilde Y_0 | \tilde X_{-k-1})-I_L( \tilde Y_0; \tilde X_{-k} | \tilde X_{-k-1}) \nonumber\\
=&H_L(\tilde Y_0 | \tilde X_{-k})-I_L( \tilde Y_0; \tilde X_{-k-1} | \tilde X_{-k}).
\end{align}
We can expand $H_L(\tilde Y_0 | \tilde X_{-k})$ as 
\begin{align}
H_L(\tilde Y_0 | \tilde X_{-k})=& H_L(\tilde Y_0 | \tilde X_{-k-1})+ I_L(\tilde Y_0; \tilde X_{-k-1} | \tilde X_{-k})\nonumber\\
&-I_L(\tilde Y_0; \tilde X_{-k} | \tilde X_{-k-1}).
\end{align}
As the above equation is valid for all values of $k$, taking the summation of $H_L(\tilde Y_0 | \tilde X_{-k})$ from $k=0$ to $\theta-1$ yields
\begin{align} \label{PMarkovC1}
H_L(\tilde Y_0 | \tilde X_{-\theta})=& H_L(\tilde Y_0 | \tilde X_0) + \sum_{k=0}^{\theta-1} I_L(\tilde Y_0; \tilde X_{-k} | \tilde X_{-k-1})\nonumber\\
&- \sum_{k=0}^{\theta-1} I_L(\tilde Y_0; \tilde X_{-k-1} | \tilde X_{-k}).
\end{align}
Thus, we can write $H_L(\tilde Y_0 | \tilde X_{-\theta})$ as a function of $\theta$ as in \eqref{eMarkov} and \eqref{g12function}.
Moreover, the functions $g_1(\theta)$ and $g_2(\theta)$ defined in \eqref{g12function} are non-decreasing of $\theta$ as $I_L(\tilde Y_0; \tilde X_{-k} | \tilde X_{-k-1}) \geq 0$ and $I_L(\tilde Y_0; \tilde X_{-k-1} | \tilde X_{-k}) \geq 0$ for all values of $k$. 

To prove the next part, we use Lemma \ref{Lemma_CMI}. Because for every $\mu, \nu \geq 0$, $\tilde Y_{0} \overset{\epsilon} \leftrightarrow \tilde X_{-\mu} \overset{\epsilon} \leftrightarrow \tilde X_{-\mu-\nu}$ is an $\epsilon$-Markov chain, we can write
\begin{align}
I_L(\tilde Y_0; \tilde X_{-k-1} | \tilde X_{-k})=O(\epsilon).
\end{align}
This implies
\begin{align}
g_2(\theta)=\sum_{k=0}^{\theta-1} O(\epsilon)= O(\epsilon).
\end{align}
The last equality is due to the summation property of big-O-notation. 
\ignore{Similarly, if $H_L(\tilde Y)$ is twice differentiable in $P_Y$, we obtain
\begin{align}
g_2(\theta)=\sum_{k=0}^{\theta-1} O(\epsilon^2)=O(\epsilon^2).
\end{align}
This completes the proof.}

\subsection{Proof of Theorem \ref{theorem2}}\label{Ptheorem2}
Using \eqref{freshness_aware_cond} and Theorem \ref{theorem1}, we obtain
\begin{align}\label{Th2_1}
\!\!\!\!\!\!&H_L(\tilde Y_0 | \tilde X_{-\Theta}, \Theta) \nonumber\\
=& \sum_{\theta \in \mathcal D} \!P_{\Theta}(\theta) \big(H_L(\tilde Y_0 | \tilde X_0)+\!\hat g_1(\theta)-g_2(\theta)\big)\nonumber\\
=&H_L(\tilde Y_0 | \tilde X_0)\!+\mathbb E_{\Theta \sim P_{\Theta}}\! [\hat g_1(\Theta)]\!-\mathbb E_{\Theta \sim P_{\Theta}}\! [g_2(\Theta)],
\end{align}
where 
\begin{align}
\hat g_1(\theta)&=g_1(\theta)-H_L(\tilde Y_0 | \tilde X_0)\nonumber\\
&=\sum_{k=0}^{\theta-1} I_L(\tilde Y_0; \tilde X_{-k} | \tilde X_{-k-1}). 
\end{align}
Because mutual information $I_L(\tilde Y_0; \tilde X_{-k} | \tilde X_{-k-1})$ is non-negative,
\begin{align} 
\hat g_1(\theta)=\sum_{k=0}^{\theta-1} I_L(\tilde Y_0; \tilde X_{-k} | \tilde X_{-k-1})\geq 0.
\end{align}
Because $\hat g_1(\theta)$ is non-negative for all $\theta$, the function $\hat g_1(\cdot)$ is Lebesgue integrable with respect to all probability measure $P_\Theta$ \cite{durrett2019probability}. Hence, the expectation $\mathbb E_{\Theta \sim P_{\Theta}} [\hat g_1(\Theta)]$ exist. Note that $E_{\Theta \sim P_{\Theta}} [\hat g_1(\Theta)]$ can be infinite ($+\infty$). By using the same argument, we get that $E_{\Theta \sim P_{\Theta}} [g_2(\Theta)]$ exists, but can be infinite. Moreover, the function $\hat g_1(\theta)$ and $g_2(\theta)$ is non-decreasing in $\theta$.

Because (i) the function $\hat g_1(\theta)$ is non-decreasing in $\theta$, (ii) the expectation $\mathbb E_{\Theta \sim P_{\Theta}} [\hat g_1(\Theta)]$ exist, and (iii) $\Theta_1 \leq_{st} \Theta_2$, we get \cite{stochasticOrder}
\begin{align}\label{Sorder}
\mathbb E_{\Theta \sim P_{\Theta_1}} [\hat g_1(\Theta)] \leq \mathbb E_{\Theta \sim P_{\Theta_2}} [\hat g_1(\Theta)].
\end{align}

Next, by using \eqref{Th2_1}, \eqref{Sorder}, and Theorem \ref{theorem1}(b), we obtain \eqref{dynamicsoln}:
\begin{align}
&H_L(\tilde Y_0 | \tilde X_{-\Theta_1}, \Theta_1) \nonumber\\
=& H_L(\tilde Y_0 | \tilde X_0)+\mathbb E_{\Theta \sim P_{\Theta_1}} [\hat g_1(\Theta)]-\mathbb E_{\Theta \sim P_{\Theta_1}} [g_2(\Theta)] \nonumber\\
\leq& H_L(\tilde Y_0 | \tilde X_0)+\mathbb E_{\Theta \sim P_{\Theta_2}} [\hat g_1(\Theta)]-\mathbb E_{\Theta \sim P_{\Theta_1}} [g_2(\Theta)]\nonumber\\
=&H_L(\tilde Y_0 | \tilde X_{-\Theta_2}, \Theta_2) \nonumber\\
&+\mathbb E_{\Theta \sim P_{\Theta_2}} [g_2(\Theta)] -\mathbb E_{\Theta \sim P_{\Theta_ 1}} [g_2(\Theta)] \nonumber\\
=&H_L(\tilde Y_0 | \tilde X_{-\Theta_2}, \Theta_2)+O(\epsilon).
\end{align} 

\ignore{Similarly, if $H_L(Y)$ is twice differentiable, then \eqref{dynamic_eqn} holds.
This completes the proof.}

\subsection{Proof of Lemma \ref{lemma_inference}}\label{Plemma_inference}
By using condition \eqref{T3condition2} and Lemma \ref{divergenceL}(a), we obtain for all $x \in \mathcal X$
\begin{align}\label{E1}
D_L\left(P_{Y_{t}| X_{t-\delta}=x} || P_{\tilde Y_{0}| \tilde X_{-\delta}=x}\right)=O(\beta).
\end{align}
Next, by using \eqref{L-CondCrossEntropy} and \eqref{E1}, we get 
\begin{align}
&H_L(Y_{t}; \tilde Y_{0} | X_{t-\delta})\nonumber\\
=& H_L(Y_{t} | X_{t-\delta}) \nonumber\\
&+\sum_{x \in \mathcal X} P_{X_{t-\delta}}(x)~D_L\left(P_{Y_{t}| X_{t-\delta}=x} || P_{\tilde Y_{0}| \tilde X_{-\delta}=x}\right)\nonumber\\
=& H_L(Y_{t} | X_{t-\delta})+O(\beta).
\end{align}  

\subsection{Proof of Theorem \ref{theorem3}}\label{Ptheorem3}
Part (a): 
By the definition of $L$-conditional cross entropy \eqref{L-CondCrossEntropy}, we get 
\begin{align}\label{T3_1}
\!\!\!\!\!&H_L(Y_{t}; \tilde Y_0 | X_{t-\delta}) \nonumber\\
=&\!\!\!\sum_{x \in \mathcal X} \!\!P_{X_{t-\delta}}(x)\mathbb{E}_{Y \sim P_{Y_t| X_{t- \delta}=x}}\!\!\left[ \!L\left(Y,a_{\tilde Y_0|\tilde X_{-\delta}=x}\right)\right] ,\!\!
\end{align} 
where the Bayes predictor $a_{\tilde Y_0|\tilde X_{-\delta}=x}$ is fixed in the inference phase for every time slot $t$. Because $\{(\tilde Y_t, \tilde X_t),t \in \mathbb Z\}$ is a stationary process, \eqref{T3_1} is a function of the inference AoI $\delta$.

Part (b): By using Lemma \ref{lemma_inference} and Theorem \ref{theorem1}, we get 
\begin{align}
H_L(Y_{t}; \tilde Y_{0} | X_{t-\delta_1})=& H_L(Y_{t} | X_{t-\delta_1})+O(\beta) \nonumber\\
\leq& H_L( Y_{t} | X_{t-\delta_2})+O(\epsilon)+O(\beta) \nonumber\\
=&H_L(Y_{t}; \tilde Y_{0} | X_{t-\delta_2})+O(\beta)+O(\epsilon)+O(\beta) \nonumber\\
=&H_L(Y_{t}; \tilde Y_{0} | X_{t-\delta_2})+O(\max\{\epsilon, \beta\}).
\end{align}
This completes the proof.

\ignore{\subsection{Proof of Theorem \ref{theorem4}}\label{Ptheorem4}
By using \eqref{Decomposed_Cross_entropy}, Lemma \ref{lemma_inference}, and Theorem \ref{theorem2}, we obtain
\begin{align}
H_L(\tilde Y_{t}; Y_{t} | \tilde X_{t-\Delta_1}, \Delta_1)=&\sum_{\delta \in \mathcal D} P_{\Delta_1}(\delta)H_L(\tilde Y_{t}; Y_{t} | \tilde X_{t-\delta}) \nonumber\\
=& \sum_{\delta \in \mathcal D} P_{\Delta_1}(\delta)~H_L(\tilde Y_{t} | \tilde X_{t-\delta})+O(\beta) \nonumber\\
=& H_L(\tilde Y_{t} | \tilde X_{t-\Delta_1}, \Delta_1)+O(\beta) \nonumber\\
\leq& H_L(\tilde Y_{t} | \tilde X_{t-\Delta_2}, \Delta_2)+O(\epsilon)+O(\beta) \nonumber\\
=&H_L(\tilde Y_{t}; Y_{t} | \tilde X_{t-\Delta_2}, \Delta_2)+O(\beta)+O(\epsilon)+O(\beta) \nonumber\\
=&H_L(\tilde Y_{t}; Y_{t} | \tilde X_{t-\Delta_2}, \Delta_2)+O(\max\{\epsilon, \beta\}).
\end{align}}
\subsection{Simplification of the Gittins Index in \eqref{gittins} }\label{GittinsDerivation}
For the bandit process $\{\Delta(t) : t \geq 0\}$ in \eqref{Bandit}, define the $\sigma$-field 
\begin{align}\label{sigma-field}
\mathcal F_s^t= \sigma(\Delta(t+k): k \in \{0, 1, \ldots, s\}),
\end{align} which is the set of events whose occurrence are determined by the realization of the process $\{\Delta(t+k) :  k \in \{0, 1, \ldots, s\}\}$ from time slot $t$ up to time slot $t+s$. Then, $\{ \mathcal F_k^t, k \in \{0, 1, \ldots\}\}$ is the filtration of the time shifted process $\{\Delta(t+k) : k \in \{0, 1, \ldots\} \}$. We define $\mathfrak M$ as the set of all stopping times by 
\begin{align}
\mathfrak M=\{ \nu \geq 0 : \{\nu=k\} \in \mathcal F_k^t, k \in \{0, 1, 2, \ldots\}\}.
\end{align}

The Gittins index $\gamma(\delta)$ \cite{gittins2011multi} is the value of reward $r$ for which the \textsc{stop} and \textsc{continue} actions are equally profitable at state $\Delta(t)=\delta$. Hence, $\gamma(\delta)$ is the root of the following equation of $r$:  
\begin{align}\label{gittins_subtract}
&\sup_{\nu \in \mathfrak M, \nu \neq 0} \mathbb E\left[ \sum_{k=0}^{\nu+T_1-1} [r-p(\Delta(t+k))]\bigg| \Delta(t)=\delta \right]\nonumber\\
&~~~~~~=\mathbb E\left[ \sum_{k=0}^{T_1-1} [r-p(\Delta(t+k))]\bigg| \Delta(t)=\delta \right].
\end{align}
where the left hand side of \eqref{gittins_subtract} is the maximum total expected profit under \textsc{continue} action and the right hand side of \eqref{gittins_subtract} is the total expected profit under \textsc{stop} action. By re-arranging \eqref{gittins_subtract}, it can be expressed as
\begin{align}\label{gittins_subtract1}
&\sup_{\nu \in \mathfrak M, \nu \neq 0} \mathbb E\left[ \sum_{k=0}^{\nu-1} [r-p(\Delta(t+k+T_1))]\bigg| \Delta(t)=\delta \right]=0.
\end{align}
Because the left hand side of \eqref{gittins_subtract1} is the supremum of strictly increasing and linear functions of $r$, it is convex, continuous, and strictly increasing in $r$. Thus, the fixed-point equation \eqref{gittins_subtract1} has a unique root. The root can also be expressed as
\begin{align}\label{gittins1}
&\gamma(\delta) \nonumber\\
=&\bigg\{r:\!\!\!\!\sup_{\nu \in \mathfrak M, \nu \neq 0}\!\!\!\! \mathbb E\left[ \sum_{k=0}^{\nu-1} [r-p(\Delta(t+k+T_1))]\bigg| \Delta(t)\!=\!\delta \right]\!\!=\!0\bigg\}.
\end{align}
Let $\nu^* \in \mathfrak M$ be the optimal stopping time that solves \eqref{gittins1}. Because of \eqref{Bandit} and $T_1 \geq 1$, for any $k \leq \nu^*$, $\Delta(t+k)=\Delta(t)+k$. Hence, $\{\Delta(t+k): 1 \leq k \leq \nu^*\}$ is completely determined by the initial value $\Delta(t)$ and for all $k \leq \nu^*$, the $\sigma$-field $\mathcal F_k^t$ can be simplified as $\mathcal F_k^t=\sigma(\Delta(t))$. Thus, any stopping time in $\mathfrak M$ is a deterministic time, given $\Delta(t)=\delta$. By this, \eqref{gittins1} can be simplified as
\begin{align}\label{gittins2}
\!\!&\gamma(\delta)\nonumber\\
=&\bigg\{r:\!\!\!\!\! \sup_{\tau \in \{1, 2, \ldots\}}\!\!\!\!\! \mathbb E\left[ \sum_{k=0}^{\tau-1} [r-p(\Delta(t+k+T_1))]\bigg| \Delta(t)\!=\!\delta \right]\!\!=\!0\bigg\}\nonumber\\
=&\bigg\{r:\!\!\!\!\! \inf_{\tau \in \{1, 2, \ldots\}} \mathbb E\left[ \sum_{k=0}^{\tau-1} [p(\Delta(t+k+T_1))-r]\bigg| \Delta(t)\!=\!\delta \right]\!\!=\!0\bigg\}\nonumber\\
=&\!\bigg\{r: \!\!\!\!\! \inf_{\tau \in \{1, 2, \ldots\}} \sum_{k=0}^{\tau-1}\!\mathbb E\left[p(\Delta(t+k+T_1))-r\bigg| \!\Delta(t)\!=\delta \right]\!\!=\!0\bigg\},
\end{align}
where $\tau$ is a deterministic positive integer.

Define 
\begin{align}
f(r)=\inf_{\tau \in \{1, 2, \ldots\}} \sum_{k=0}^{\tau-1}\mathbb E\left[p(\delta+k+T_1)-r\right].
\end{align}
Similar to Lemma 7 in \cite{orneeTON2021}, the following lemma holds.
\begin{lemma}\label{fraction_programming}
$f(r) \lesseqgtr 0$ if and only if
\begin{align}
\inf_{\tau \in \{1, 2, \ldots\}} \frac{1}{\tau} \sum_{k=0}^{\tau-1} \mathbb E \left [p(\delta+k+T_{1}) \right] \lesseqgtr r.
\end{align}
\end{lemma}
By \eqref{gittins1}, \eqref{gittins2}, and Lemma \ref{fraction_programming}, the root of equation $f(r) = 0$ is given by \eqref{gittins}. This completes the proof. 

\subsection{Proof of Theorem \ref{theorem5}}\label{Ptheorem5}
The infinite horizon average AoI penalty problem \eqref{sub_scheduling_problem} can be cast as a Markov decision problem (MDP). \ignore{We solve the problem by using dynamic programming (DP) \cite{bertsekas2012dynamic}.} To describe the MDP, we define the action, state, state transition, and penalty function.  

\begin{itemize}
\item Action: If the channel server is idle, the possible actions taken by the scheduler at time slot $t$ are $d_b(t)=0$ \textsc{``do not send''} and $d_b(t)=1$ \textsc{``send''}. Then, $S_{i+1}$ is determined by
\begin{align}\label{modified}
S_{i+1} = \inf_{t \in \mathbb Z}\{t \geq D_i : d_b(t)=1\}.
\end{align}
Hence, one can also use $g=(d_b(0), d_b(1), d_b(2), \ldots)$ to represent a policy in $\mathcal G$. 

\item AoI Penalty: The penalty at every time slot $t$ is $p(\Delta(t))$.

\item State: \ignore{Because the transmission times are not affected by the adopted scheduling policies in $\mathcal F \times \mathcal G$ and transmission time $T_i$'s are i.i.d., the only available information useful for the scheduling decision is the current age value $\Delta(t)$. Hence,} The state of the MDP is the age value $\Delta(t)$. \ignore{{\red Consider the $i$-th feature is the latest feature delivered before time slot $t$. The information available at the scheduler are the previous AoI values and actions $\{\big(\Delta(s), d_b(s)\big): 0 \leq s < t\}$, previous transmission times $\{T_j: j=0, \ldots, i \}\}$, and the current AoI $\Delta(t)$. The AoI penalty $p(\Delta(t+k))$ for $k > 0$ depends on the current AoI $\Delta(t)$, future transmission times $\{T_{i+k}: k=1, 2, \ldots\}$, and the actions $\{d_b(t+k): k=0, 1, \ldots\}$. Because (i) the future transmission times $\{T_{i+k}: k=1, 2, \ldots\}$ are not affected by the previous actions $\{d_b(s): s=0, 1, \ldots, t-1\}$ and (ii) the transmission times $T_i$ are mutually independent, the AoI penalty $p(\Delta(t+k))$ for $k > 0$ is independent of the available information at the scheduler given the current AoI $\Delta(t)$. Hence, the only relevant information available at the scheduler to make decision $d_b(t)$ is the current age value $\Delta(t)$. So, the state of the MDP is the current age value $\Delta(t)$.}}

\item State Transition: Because $b_i=b$ for all $i$, the state $\Delta(t)$ evolves as follows
\begin{align}\label{AoIProcess}
\Delta(t)=
  \begin{cases}
        T_{i}+b, &\text{if}~t=D_{i}, i=0, 1, \ldots, \\
        \Delta(t-1)+1, &\text{otherwise}.
    \end{cases}
    \end{align}
\end{itemize}

\ignore{In order to solve this MDP, we consider an associated $\alpha$-discounted MDP for $0< \alpha <1$. If $\Delta(t)=\delta$, the Bellman optimality equation of the $\alpha$-discounted MDP is given by
\begin{align}\label{alpha-discounted-problem}
J_{\alpha}(\delta)=\min_{d_b(t) \in \{0, 1\}} Q^{\alpha}_b(\delta, d_b(t)),
\end{align}
where 
\begin{align}\label{alpha-discounted-problem1}
Q^{\alpha}_b(\delta, 1)&= \mathbb E \left [ \sum_{k=0}^{T_{i+1}-1} \alpha^k p(\delta+k)\right]+ \mathbb E[\alpha^{T_{i+1}} J_{\alpha}(T_{i+1}+b)], \nonumber\\
Q^{\alpha}_b(\delta, 0)&=\inf_{\tau \in \mathfrak M} \mathbb E \left [ \sum_{k=0}^{\tau+T_{i+1}-1} \alpha^k p(\delta+k)\right]+ \mathbb E[\alpha^{T_{i+1}} J_{\alpha}(T_{i+1}+b)].
\end{align}

Without loss of generality, we can add $\min_{\delta} p(\delta)$ to $p(\delta)$ and get
\begin{align}
\tilde p(\delta)&=p(\delta)+\min_{\delta} p(\delta), \text{for all}~\delta=0, 1, \ldots.
\end{align}
Because the AoI penalty function $p(\delta)$ satisfies 
\begin{align}
|p(\delta)| \leq M, \text{\normalfont{for all}}~ \delta \in \mathcal D,
\end{align}
$\tilde p(\delta)$ is non-negative and satisfies
\begin{align}
0 \leq \tilde p(\delta) \leq \tilde M=M+\min_{\delta} p(\delta), \text{for all}~\delta=0, 1, \ldots.
\end{align}
By replacing $\tilde p(\delta)$ with $p(\delta)$, \eqref{alpha-discounted-problem} and \eqref{alpha-discounted-problem1} become
\begin{align}\label{alpha-discounted-problem2}
\tilde J_{\alpha}(\delta)=\min_{d_b(t) \in \{0, 1\}} \tilde Q^{\alpha}_b(\delta, d_b(t)),
\end{align}
where 
\begin{align}\label{alpha-discounted-problem3}
\tilde Q^{\alpha}_b(\delta, 1)&= \mathbb E \left [ \sum_{k=0}^{T_{i+1}-1} \alpha^k \tilde p(\delta+k)\right]+ \mathbb E[\alpha^{T_{i+1}} \tilde J_{\alpha}(T_{i+1}+b)], \nonumber\\
\tilde Q^{\alpha}_b(\delta, 0)&=\inf_{\tau \in \mathfrak M} \mathbb E \left [ \sum_{k=0}^{\tau+T_{i+1}-1} \alpha^k \tilde p(\delta+k)\right]+ \mathbb E[\alpha^{T_{i+1}} \tilde J_{\alpha}(T_{i+1}+b)].
\end{align}
and the action space is finite, similar to \cite[Proposition 5.6.4] {bertsekas2012dynamic}, it is proven in Appendix \ref{pbounded_alpha} that
\begin{align}\label{bounded_alpha}
\bigg| \tilde J_{\alpha}(\delta)-\mathbb E[\tilde J_{\alpha}(T_{i+1}+b)]\bigg| \leq \tilde M~\mathbb E[T_{i+1}],
\end{align}
where $\mathbb E[T_{i+1}]$ is finite. 
Now, we consider 
\begin{align}
V_{\alpha, b}(\delta)=\tilde J_{\alpha}(\delta)-\mathbb E[\tilde J_{\alpha}(T_{i+1}+b)]-(1-\alpha) \mathbb E[\tilde J_{\alpha}(T_{i+1}+b)]
\end{align}
By taking $\lim_{\alpha \to 1} V_{\alpha, b}(\delta)$, we obtain the Bellman optimality equation for the average cost MDP \eqref{sub_scheduling_problem}:
\begin{align}
\lim_{\alpha \to 1} V_{\alpha, b}(\delta)=V_{b}(\delta)
\end{align}
with
\begin{align}
\lim_{\alpha \to 1} (1-\alpha) \mathbb E[\tilde J_{\alpha}(T_{i+1}+b)]=\tilde p_b+\min_{\delta} p(\delta).
\end{align}}
Because the state space is countable, action space is finite, and $p(\delta)$ is bounded,  \ignore{it is proven in Appendix \ref{pbounded_alpha} that} if $\Delta(t)=\delta$, then the optimal decision $d_b(t)$ in time slot $t$ satisfies the following Bellman optimality equation \cite[Section 5.6.3] {bertsekas2012dynamic}
\begin{align}\label{Bellman_optimality}
V_b(\delta)=\min_{d_b(t) \in \{0, 1\}} Q_b(\delta, d_b(t)),
\end{align}
where the function $V_b$ is the relative value function, the relative value $V_b(\delta)$ is the expected total cost relative to the optimal average cost $\bar p_b$ of the problem \eqref{sub_scheduling_problem} when starting from state $\delta$ and following an optimal policy, and $Q_b(\delta, d_b(t))$ is given by
\begin{align}
&Q_b(\delta, 1)\nonumber\\
&= \mathbb E \left [ \sum_{k=0}^{T_{i+1}-1} [p(\delta+k)-\bar p_b]\right]+\mathbb E[V_b(T_{i+1}+b)], \\
&Q_b(\delta, 0)\nonumber\\
&=\inf_{\nu \in \mathfrak M, \nu \neq 0} \mathbb E \left [ \sum_{k=0}^{\nu+T_{i+1}-1} [p(\delta+k)-\bar p_b]\right]+\mathbb E[V_b(T_{i+1}+b)] \nonumber\\
&\overset{(a)}{=}\inf_{\tau \in \{1, 2, \ldots\}} \mathbb E \left [ \sum_{k=0}^{\tau+T_{i+1}-1} [p(\delta+k)-\bar p_b]\right]+\mathbb E[V_b(T_{i+1}+b)] \nonumber\\
&\overset{(b)}{=}\inf_{\tau \in \{1, 2, \ldots\}} \mathbb E \left [ \sum_{k=0}^{\tau-1} [p(\delta+k+T_{i+1})-\bar p_b]\right] \nonumber\\
&+\mathbb E \left [ \sum_{k=0}^{T_{i+1}-1} [p(\delta+k)-\bar p_b]\right]+\mathbb E[V_b(T_{i+1}+b)].
\end{align}
where $\mathfrak M$ is the set of all stopping times defined in Appendix \ref{GittinsDerivation}.  As deduced in Appendix \ref{GittinsDerivation}, given $\Delta(t)=\delta$, the set of all stopping times is $\mathfrak M=\{0, 1, 2, \ldots\}$. We obtain equality (a) because $\mathfrak M$ is $\{0, 1, 2, \ldots\}$ given $\Delta(t)=\delta$. By re-arranging (a), we get equality (b). 

By \eqref{Bellman_optimality}, $d_b(t)=1$ is optimal, if 
\begin{align}\label{Bellman-inequality}
Q_b(\delta, 0) -Q_b(\delta, 1) \geq 0.
\end{align}
The inequality \eqref{Bellman-inequality} can also be expressed as 
\begin{align}\label{eq_fractional}
\inf_{\tau \in \{1, 2, \ldots\}} \mathbb E \left [ \sum_{k=0}^{\tau-1} [p(\delta+k+T_{i+1})-\bar p_b]\right] \geq 0.
\end{align} 
Next, Lemma \ref{fraction_programming} implies that  the inequality \eqref{eq_fractional} holds if and only if 
\begin{align}\label{eq_fractional1}
\inf_{\tau \in \{1, 2, \ldots\}} \frac{1}{\tau} \sum_{k=0}^{\tau-1} \mathbb E \left [p(\delta+k+T_{i+1}) \right] \geq \bar p_b.
\end{align}
Now, using \eqref{modified} and \eqref{Bellman-inequality}-\eqref{eq_fractional1}, we get \eqref{OptimalPolicy_Sub}. 

Since the transmission times $T_i$ are i.i.d., similar to \cite[Appendix F]{SunNonlinear2019}, we get that the optimal objective value to \eqref{sub_scheduling_problem} is 
\begin{align}\label{getoptimalvalue}
\bar p_b=\frac{\mathbb{E}\left[\sum_{t=D_i(\bar p_b)}^{D_{i+1}(\bar p_b)-1}  p\big(\Delta(t)\big)\right]}{\mathbb{E}\left[D_{i+1}(\bar p_b)-D_i(\bar p_b)\right]}.
\end{align}
Hence, $\bar p_b$ is equal to the root of \eqref{bisection}. The left hand side of \eqref{bisection} is concave, continuous, and strictly decreasing in $\beta_b$ \cite{orneeTON2021}. Hence, the root of \eqref{bisection} is unique. This completes the proof. 

\subsection{Proof of Theorem \ref{theorem6}}\label{Ptheorem6}
\ignore{Theorem \ref{theorem6} consists of part (a) and part (b). If part (a) holds, then part (b) is directly obtained using Theorem \ref{theorem5}. Hence, It is enough to prove part (a).}

The problem \eqref{scheduling_problem} can be cast as an MDP problem. The State and the penalty of the MDP are the same as the MDP discussed in Appendix \ref{Ptheorem5}. The action space is different: If the channel is idle, the scheduler sends $(b+1)$-th freshest feature or does not send any feature. Let $d(t)=\textsc{idle}$ means the scheduler does not send feature at time $t$ and $d(t)=b$ means the scheduler sends $(b+1)$-th freshest feature at time $t$. Then, $S_{i+1}$ and $b_{i+1}$ are determined by 
\begin{align}\label{sendingtimecal}
S_{i+1}&=\inf_{t \in \mathbb Z} \{t \geq S_i+T_i: d(t)\neq \textsc{idle}\}, \\
b_{i+1}&=d(S_{i+1}).
\end{align}

If the channel is idle and $\Delta(t)=\delta$, then the optimal decision $d(t)$ in time slot $t$ satisfies the following Bellman optimality equation 
\begin{align}\label{Bellman_optimality1}
V(\delta)=\min_{d(t) \in \{\textsc{idle}, 0, 1, \ldots, B-1\}} Q(\delta, d(t)),
\end{align}
where the function $V$ is the relative value function and $Q(\delta, d(t))$ is given by
\begin{align}
&Q(\delta, b) \nonumber\\
&= \mathbb E \left [ \sum_{k=0}^{T_{i+1}-1} [p(\delta+k)-\bar p_{\text{opt}}]\right]+\mathbb E[V(T_{i+1}+b)], \\
&Q(\delta, \textsc{idle}) \nonumber\\
&=\inf_{\tau \in \{1, 2, \ldots\}} \mathbb E \left [ \sum_{k=0}^{\tau-1} [p(\delta+k+T_{i+1})-\bar p_{\text{opt}}]\right] \nonumber\\
&+\mathbb E \left [ \sum_{k=0}^{T_{i+1}-1} [p(\delta+k)-\bar p_{\text{opt}}]\right]+\min_{b \in \mathcal B} \mathbb E[V(T_{i+1}+b)].
\end{align}
where $\mathcal B=\{0, 1, \ldots, B-1\}$ and $\bar p_{\text{opt}}$ is the optimal value of \eqref{scheduling_problem}.

By \eqref{Bellman_optimality1}, $d(t)= \textsc{idle}$ is not an optimal choice if
\begin{align}\label{proofinequality}
Q(\delta, \textsc{idle})-\min_{b \in \mathcal B} Q(\delta, b) \geq 0.
\end{align}
By using similar steps \eqref{Bellman-inequality}-\eqref{eq_fractional1}, we can get that the inequality \eqref{proofinequality} holds if and only if 
\begin{align}\label{proofinequality1}
\gamma(\delta) \geq \bar p_{\text{opt}}.
\end{align}
Then, by using \eqref{sendingtimecal}, \eqref{proofinequality}, and \eqref{proofinequality1}, we get the optimal sending time $S^*_{i+1}$ in \eqref{Optimal_Scheduler}. 

Next, we need to get the optimal $b^*_{i+1}$ and $\bar p_{\text{opt}}$. When $\Delta(S^*_{i+1})=\delta$, $b^*_{i+1}=b^*$ is optimal if
\begin{align}
b^*=&\arg\min_{b \in \mathcal B} Q(\delta, b) \nonumber\\
=&\arg\min_{b \in \mathcal B}\mathbb E \left [ \sum_{k=0}^{T_{i+1}-1} [p(\delta+k)-\bar p_{\text{opt}}]\right]+\mathbb E[V(T_{i+1}+b)] \nonumber\\
=&\arg\min_{b \in \mathcal B}\mathbb E[V(T_{i+1}+b)].
\end{align}
Observe that the optimal $b^*$ is independent of the state $\Delta(S_{i+1})=\delta$. Moreover, because $T_i$ is identically distributed, 
\begin{align}
\mathbb E[V(T_{0}+b)]=\mathbb E[V(T_{1}+b)]=\ldots= \mathbb E[V(T_{i}+b)], \forall i.
\end{align}
Thus, the optimal $b^*$ is constant for all $i$. 
If $b_{i+1}=b$ for all $i$, then $\bar p_b$ is the average inference error. Hence, the optimal choice $b^*$ satisfies 
\begin{align}
b^*=\arg\min_{b \in \mathcal B} \bar p_b,
\end{align}
where $\bar p_b$ is $\beta_b$, which is the root of \eqref{bisection}. The optimal value is 
\begin{align}
\bar p_{\text{opt}}=\min_{b \in \mathcal B} \bar p_b.
\end{align}

\subsection{Proof of Theorem \ref{theorem7}}\label{ptheorem7}
\ignore{Similar to Section \ref{Ptheorem5}, the problem \eqref{decoupled_problem} can be cast as an MDP and we solve the problem by using dynamic programming. Then, we prove the indexability of the arm $(l, b_l)$.}

If the channel is idle and $\Delta_l(t)=\delta$, the optimal decision $d_{l, b_l}(t)$ for the problem \eqref{decoupled_problem} in time slot $t$ satisfies the following Bellman optimality equation 
\begin{align}\label{Belmandecoupled}
V_{l, b_l}(\delta)=\min_{d_{l, b_l}(t) \in \{0, 1\}} Q_{l, b_l}(\delta, d_{l, b_l} (t)),
\end{align}
where the function $V_{l, b_l}$ is the relative value function and $ Q_{l, b_l}(\delta, d_{l, b_l} (t))$ is given by
\begin{align}
&Q_{l, b_l}(\delta, 1)\nonumber\\
&= \mathbb E \left [ \sum_{k=0}^{T_{l, i+1}-1} [w_l p_l(\delta+k)-\bar p_{l, b_l}(\lambda)]\right]\! \\
&+\!\mathbb E[V_{l, b_l}(T_{l, i+1}+b_l)]\!+\!\lambda, \nonumber\\
&Q_{l, b_l}(\delta, 0)\nonumber\\
&=\inf_{\tau \in \{1, 2, \ldots\}} \mathbb E \left [ \sum_{k=0}^{\tau-1} [w_l p_l(\delta+k+T_{l, i+1})-\bar p_{l, b_l}(\lambda)]\right] \nonumber\\
&+\mathbb E \left [ \sum_{k=0}^{T_{l, i+1}-1} [w_l p_l(\delta+k)-\bar p_{l, b_l}(\lambda)]\right]\!\nonumber\\
&+\! \mathbb E[V_{l, b_l}(T_{l, i+1}+b_l)]\!+\! \lambda.
\end{align}
where $\bar p_{l, b_l}(\lambda)$ is optimal objective value of the problem \eqref{decoupled_problem}. 

Similar to the proof of \eqref{eq_fractional1} and \eqref{getoptimalvalue} in Section \ref{Ptheorem5}, by solving \eqref{Belmandecoupled}, $d_{l, b_l}(t)=0$ is optimal if 
\begin{align}\label{solutiondecoupled}
w_l~\gamma_l(\delta) \leq \bar p_{l, b_l}(\lambda), 
\end{align} 
otherwise $d_{l, b_l}(t)=1$ is optimal, where $\bar p_{l, b_l}(\lambda)$ is given by
\begin{align}\label{getoptimalvaluearm}
\bar p_{l, b_l}(\lambda)=\frac{C^l_{i, i+1}(\lambda)}{N^l_{i, i+1}(\lambda)},
\end{align}
where $C^l_{i, i+1}(\lambda)$ is the expected penalty of source $l$ starting from $i$-th delivery time to $(i+1)$-th delivery time and $N^l_{i, i+1}(\lambda)$ is the expected number of time slots from $i$-th delivery time to $(i+1)$-th delivery time, given by
\begin{align}\label{C}
\!\!C^l_{i, i+1}(\lambda)\!\!=\!\mathbb{E}\left[\sum_{t=D_{l, i}(\bar p_{l, b_l}(\lambda))}^{D_{l, i+1}(\bar p_{l, b_l}(\lambda))-1}  \!\!\!\!\!w_l p_l\big(\Delta_l(t)\big)\right]+\lambda \mathbb E[T_{l, i+1}],
\end{align}
\begin{align}\label{N}
N^l_{i, i+1}(\lambda)=\mathbb{E}\left[D_{l, i+1}(\bar p_{l, b_l}(\lambda))-D_{l, i}(\bar p_{l, b_l}(\lambda))\right],
\end{align}
the $(i+1)$-th feature delivery time $D_{l, i+1}(\bar p_{l, b_l}(\lambda))$ from source $l$ is  
\begin{align}\label{deliverytimearm}
D_{l, i+1}(\bar p_{l, b_l}(\lambda))=S_{l, i+1}(\bar p_{l, b_l}(\lambda))+T_{l, i+1}
\end{align}
and the $(i+1)$-th sending time $S_{l, i+1}(\bar p_{l, b_l}(\lambda))$ is 
\begin{align}\label{sendingtimearm}
S_{l, i+1}(\bar p_{l, b_l}(\lambda))=\inf_{t \in \mathbb Z}\{t \geq D_{l, i}: w_l~\gamma_l(\delta) \leq \bar p_{l, b_l}(\lambda)\}.
\end{align}
The sending time $S_{l, i+1}(\bar p_{l, b_l}(\lambda))$ can also be expressed as 
\begin{align}\label{revisedsendingtime}
S_{l, i+1}(\bar p_{l, b_l}(\lambda))=D_{l, i}(\bar p_{l, b_l}(\lambda))+z(T_{l, i}, b_l, \bar p_{l, b_l}(\lambda)),
\end{align}
the waiting time $z(T_{l, i}, b_l, \bar p_{l, b_l}(\lambda))$ after the delivery time $D_{l, i}(\bar p_{l, b_l}(\lambda))$ is 
\begin{align}\label{waitingformulation}
&z(T_{l, i}, b_l, \bar p_{l, b_l}(\lambda))\nonumber\\
=&\inf_{z \in \mathbb Z}\{ z \geq 0: w_l~\gamma_l(\Delta_l(D_{l, i}(\bar p_{l, b_l}(\lambda)))+z) \geq \bar p_{l, b_l}(\lambda)\} \nonumber\\
=&\inf_{z \in \mathbb Z}\{ z \geq 0: w_l~\gamma_l(T_{l, i}+b_l+z) \geq \bar p_{l, b_l}(\lambda)\},
\end{align}
where the last equality holds because from \eqref{multi-source-Age}, we get
\begin{align}\label{MultiAoIProcess}
\Delta_l(t)=
  \begin{cases}
        T_{l, i}+b_l, &\text{if}~t=D_{l, i}(\bar p_{l, b_l}(\lambda)), i=0, 1, \ldots, \\
        \Delta_l(t-1)+1, &\text{otherwise}.
    \end{cases}
    \end{align}
By using \eqref{deliverytimearm}-\eqref{MultiAoIProcess}, \eqref{C} and \eqref{N} reduce to 
\begin{align}
C^l_{i, i+1}(\lambda)=&\mathbb{E}\left[\sum_{k=T_{l, i}}^{T_{l, i}+z(T_{l, i}, b_l, \bar p_{l, b_l}(\lambda))+T_{l, i+1}-1}  \!\!\!\!\!\!w_l~p_l(k+b_l)\right]\nonumber\\
&+\lambda \mathbb E[T_{l, i+1}],
\end{align}
\begin{align}
N^l_{i, i+1}(\lambda)=\mathbb{E}\left[ z(T_{l, i}, b_l, \bar p_{l, b_l}(\lambda))+T_{l, i+1}\right].
\end{align}
Thus, the optimal objective value $\bar p_{l, b_l}(\lambda)$ is exactly equal to 
\begin{align}\label{bisectionrevised}
\bar p_{l, b_l}(\lambda)=\{\beta: f(\beta)+\lambda \mathbb E[T_{l, i+1}]=0\},
\end{align}
\begin{align}\label{f-beta-lambda}
f(\beta)=&\mathbb{E}\left[\sum_{k=T_{l, i}}^{T_{l, i}+z(T_{l, i}, b_l, \beta)+T_{l, i+1}-1}  \!\!\!\!\!\!w_l~p_l(t+b_l)\right] \nonumber\\
&- \beta~ \mathbb{E}\left[ z(T_{l, i}, b_l, \beta)+T_{l, i+1}\right],
\end{align}

Now, we prove the indexability of the arm $(l, b_l)$ by using \eqref{solutiondecoupled} and \eqref{bisectionrevised}. Because $f(\beta)$ is continuous and strictly decreasing in $\beta$ \cite{orneeTON2021}, $\bar p_{l, b_l}(\lambda)$ defined in \eqref{bisectionrevised} is unique and continuous in $\lambda$. From \eqref{bisectionrevised}, we get 
\begin{align}\label{eq_f1}
f(\bar p_{l, b_l}(\lambda_1)) = -\lambda_1 \mathbb E[T_{l, i+1}],
f(\bar p_{l, b_l}(\lambda_2)) = -\lambda_2 \mathbb E[T_{l, i+1}],
\end{align}
Since $f(\beta)$ is continuous and strictly decreasing in $\beta$, if $\lambda_2 > \lambda_1$, then \eqref{eq_f1} yields
\begin{align}\label{eq_f2}
\bar p_{l, b_l}(\lambda_2) > \bar p_{l, b_l}(\lambda_1),
\end{align}
 i.e., $\bar p_{l, b_l}(\lambda)$ is continuous and strictly increasing function of $\lambda$. By using the definition of the set $\Omega_{l, b_l}(\lambda)$ in Section \ref{Multi-scheduling} and \eqref{solutiondecoupled}, we obtain 
\begin{align}\label{passive-set}
\Omega_{l, b_l}(\lambda)=\{\delta: \bar p_{l, b_l}(\lambda) \geq w_l~\gamma_l (\delta)\}.
\end{align}
\ignore{From \eqref{passive-set}, we observe that $\delta \in \Omega_{l, b_l}(\lambda)$ if $\beta_{b_l}(\lambda)$ is greater or equal to the value of weighted Gittins index $w_l~ \gamma(\delta)$.} For a given $\delta$, if $\delta \in \Omega_{l, b_l}(\lambda_1)$, then
\begin{align}\label{eq_f3}
\bar p_{l, b_l}(\lambda_1) \geq w_l~\gamma_l (\delta).
\end{align}
From \eqref{eq_f2}, \eqref{passive-set}, and \eqref{eq_f3}, we obtain $\delta \in \Omega_{l, b_l}(\lambda_2)$. Hence, we get $\Omega_{l, {b_l}}(\lambda_1) \subseteq \Omega_{l, b_l}(\lambda_2)$. Thus, by the definition of indexablity in Section \ref{Multi-scheduling}, the arm $(l, b_l)$ is indexable for all values of $l$ and $b_l$. This concludes the proof. 

\subsection{Proof of Theorem \ref{theorem8}}\label{ptheorem8}
By using the definition of Whittle index in Section \ref{Multi-scheduling}, the Whittle index $W_{l, b_l}(\delta)$ is
\begin{align}\label{def-whittle}
W_{l, b_l}(\delta) = \inf \{\lambda \in\mathbb R: \delta\in \Omega_{l, b_l}(\lambda)\}
\end{align}
Now, substituing \eqref{passive-set} into \eqref{def-whittle}, we obtain
\begin{align}\label{eq_whittle_1}
W_{l, b_l}(\delta)=\inf\{\lambda: w_l~\gamma_l(\delta) \leq \bar p_{l, b_l}(\lambda)\}.
\end{align}
Because $\bar p_{l, b_l}(\lambda)$ is continuous and strictly increasing function of $\lambda$, \eqref{eq_whittle_1} implies that the Whittle index $W_{l, b_l}(\delta)$ is unique and satisfies 
\begin{align}\label{eq_whittle_2}
w_l~\gamma_l(\delta) = \bar p_{l, b_l}(W_{l, b_l}(\delta)),
\end{align}
where 
\begin{align}\label{beta-lambda-2}
\bar p_{l, b_l}(\lambda)=\frac{C^l_{i, i+1}(\lambda)}{N^l_{i, i+1}(\lambda)}.
\end{align}
Substituting \eqref{eq_whittle_2} into \eqref{beta-lambda-2}, we obtain 
\begin{align}\label{beta-lambda-3}
w_l \gamma(\delta)=\frac{C^l_{i, i+1}(W_{l, b_l}(\delta))}{N^l_{i, i+1}(W_{l, b_l}(\delta))}.
\end{align}
Because $T_{l, i}$'s are i.i.d. for all $i$, we can write 
\begin{align}\label{C2}
&C^l_{i, i+1}(W_{l, b_l}(\delta))\nonumber\\
=&w_l~\mathbb{E}\left[\sum_{k=T_{l, 1}}^{T_{l, 1}+z(T_{l, 1}, b_l, \bar p_{l, b_l}(W_{l, b_l}(\delta))+T_{l, 2}-1}p_l(k+b_l)\right] \nonumber\\
&+W_{l, b_l}(\delta)  \mathbb E[T_{l, 1}],
\end{align}
\begin{align}\label{N2}
N^l_{i, i+1}(W_{l, b_l}(\delta))=\mathbb{E}\left[ z(T_{l, 1}, b_l, \bar p_{l, b_l}(W_{l, b_l}(\delta))+T_{l, 2}\right].
\end{align}

Equations \eqref{beta-lambda-3}-\eqref{N2} yield \eqref{Whittle_Index}. In the statement of Theorem \ref{theorem8}, we denoted the waiting time $z(T_{l, 1}, b_l, w_l \gamma_l(\delta))$ as $z(T_{l, 1}, b_l, \delta)$. This concludes the proof.

\subsection{Proof of Lemma \ref{divergenceL}}\label{PdivergenceL}
To prove Lemma \ref{divergenceL}, we will use the sub-gradient mean value theorem \cite{bertsekas2003convex}. When $H_L(Y)$ is twice differentiable in $P_Y$, we can use second order Taylor series expansion. 

By condition \eqref{condition_divergenceL}, we get 
\begin{align}
\sum_{y \in \mathcal Y} \dfrac{(P_Y(y)-Q_Y(y))^2}{Q_Y(y)} \leq \beta^2. 
\end{align}
The above condition can be expressed equivalently as for all $y \in \mathcal Y$,
\begin{align}\label{L1}
 \dfrac{P_Y(y)-Q_y(y)}{\sqrt{Q_Y(y)}} =& ~\beta~\psi(y),
 \end{align}
 where 
\begin{align}
\sum_{y \in \mathcal{Y}} \psi^2(y) \leq 1.
\end{align}
This yields
\begin{align}\label{L2}
 \sum_{y \in \mathcal{Y}} \left(P_Y(y)-Q_Y(y)\right)^2 = O(\beta^2).
\end{align}

Define a convex function $g: \mathbb R^{|\mathcal Y|} \mapsto \mathbb R$ as
\begin{align}\label{gfunction}
g(\mathbf z)=\sum_{i=1}^{|\mathcal Y|} z_i~L(y_i, a_{Q_Y})-\min_{a\in\mathcal A} \sum_{i=1}^{|\mathcal Y|}  z_i~ L(y_i, a),
\end{align}
where $a_{Q_Y}$ is a Bayes action associated with distribution $Q_Y$.

Because $g(\mathbf z)$ is a convex function and the set of sub-gradients of $g(\mathbf z)$ is bounded \cite[Proposition 4.2.3]{bertsekas2003convex}, by using sub-gradient mean value theorem \cite{bertsekas2003convex}, \eqref{divergence}, and \eqref{L1}, we get 
\begin{align}\label{Lemma3_e1}
    g(\mathbf{p}_Y)=&D_L(P_Y||Q_Y)\nonumber\\
    =&g(\mathbf{q}_Y)+O\left(\sum_{y \in \mathcal Y} \left|P_Y(y)-Q_Y(y)\right|\right) \nonumber\\
 =&D_L(Q_Y||Q_Y)+O\left(\sum_{y \in \mathcal Y} \left|P_Y(y)-Q_Y(y)\right|\right)\nonumber\\
 =&O\left(\sum_{y \in \mathcal Y} \left|\beta~\psi(y) \sqrt{Q_Y(y)}\right|\right) \nonumber\\
 =&O(\beta).
\end{align}

Now, we moved to the case that $H_L(Y)$ is assumed to be twice differentiable in $P_Y$. The function $g(\mathbf{p}_Y)$ can also be expressed by $H_L(Y)$ as
\begin{align}\label{gH}
g(\mathbf{p}_Y)=&\sum_{i=1}^{|\mathcal Y|} P_Y(y_i)~L(y_i, a_{Q_Y})-\min_{a\in\mathcal A} \sum_{i=1}^{|\mathcal Y|}  P_Y(y_i)~ L(y_i, a) \nonumber\\
=&\sum_{i=1}^{|\mathcal Y|} P_Y(y_i)~L(y_i, a_{Q_Y})-H_L(Y).
\end{align}
Because $H_L(Y)$ is assumed to be twice differentiable in $P_Y$, from \eqref{gH}, we get that $g(\mathbf{p}_Y)$ is twice differentiable in $\mathbf{p}_Y$. Moreover, 
\begin{align}
g(\mathbf{p}_Y) \geq 0, \forall \mathbf{p}_Y \in \mathbb R^{|\mathcal Y|}
\end{align} and 
\begin{align}\label{dqq}
g(\mathbf{q}_Y)=D_L(Q_Y||Q_Y)=0.
\end{align}
By using the first-order necessary condition for optimality, the gradient of $g(\mathbf{p}_Y)$ at point $\mathbf{p}_Y=\mathbf{q}_Y$ is zero, i.e.,
\begin{align}\label{first_order_optimality}
\nabla g(\mathbf{q}_Y)=0.
\end{align}
Next, by \eqref{dqq} and \eqref{first_order_optimality}, the second-order Taylor series expansion of $g(\mathbf p_Y)$ at $\mathbf p_Y=\mathbf q_Y$ is
\begin{align}\label{Taylor}
g(\mathbf{p}_Y)
=&g(\mathbf{q}_Y)+(\mathbf{p}_Y-\mathbf{q}_Y)^T \nabla g(\mathbf{q}_Y)\nonumber\\
&+ \frac{1}{2} (\mathbf{p}_Y-\mathbf{q}_Y)^T \mathcal H(\mathbf{q}_Y) (\mathbf{p}_Y-\mathbf{q}_Y)\nonumber\\
&+ o\left(\sum_{y \in \mathcal Y} (P_Y(y)-Q_Y(y))^2\right) \nonumber\\
=&\frac{1}{2}\! (\mathbf{p}_Y\!-\mathbf{q}_Y)^T\! \mathcal H(\mathbf{q}_Y) (\mathbf{p}_Y-\mathbf{q}_Y)\! \nonumber\\
&+ o\left(\sum_{y \in \mathcal Y} \!(P_Y(y)\!-Q_Y(y))^2\right),
\end{align}
where $\mathcal H(\mathbf{q}_Y)$ is the Hessian matrix of $g(\mathbf{p}_Y)$ at point $\mathbf{p}_Y=\mathbf{q}_Y$. 

Because $g(\mathbf{p}_Y)$ is a convex function, $$(\mathbf{p}_Y-\mathbf{q}_Y)^T \mathcal H(\mathbf{q}_Y) (\mathbf{p}_Y-\mathbf{q}_Y) \geq 0.$$ Moreover, we can write 
\begin{align}\label{big_O_H}
&\frac{1}{2}(\mathbf{p}_Y\!-\mathbf{q}_Y)^T\! \mathcal H(\mathbf{q}_Y) (\mathbf{p}_Y-\mathbf{q}_Y)\! \nonumber\\
&=\frac{1}{2}\sum_{y, y'} (P_{Y}(y)-Q_Y(y))\mathcal H(\mathbf{q}_Y)_{y, y'}(P_{Y}(y')-Q_Y(y'))\nonumber\\
&=O\left(\sum_{y, y'}\! (P_Y(y)-Q_Y(y))(P_{Y}(y')-Q_Y(y'))\right).
\end{align} 
Now, using \eqref{big_O_H} in \eqref{Taylor}, we get
\begin{align}\label{last}
g(\mathbf{p}_Y)=&O\left(\sum_{y, y'}\! (P_Y(y)-Q_Y(y))(P_{Y}(y')-Q_Y(y'))\right)\nonumber\\
&+ o\left(\sum_{y \in \mathcal Y} \!(P_Y(y)\!-Q_Y(y))^2\right).
\end{align}
Substituting \eqref{L1} and \eqref{L2} into \eqref{last}, we obtain
\begin{align}
g(\mathbf{p}_Y)=D_L(P_Y||Q_Y)=O(\beta^2)+o(\beta^2)=O(\beta^2).
\end{align}

This completes the proof.

\ignore{\subsection{Proof of Equation \eqref{bounded_alpha}}\label{pbounded_alpha}
By using \eqref{alpha-discounted-problem}, we get 
\begin{align}\label{bound1}
J_{\alpha}(\delta) &\leq Q^{\alpha}_b(\delta, 1) \nonumber\\
&=\mathbb E \left [ \sum_{k=0}^{T_{i+1}-1} \alpha^k p(\delta+k)\right]+ \mathbb E[\alpha^{T_{i+1}} J_{\alpha}(T_{i+1}+b)]\nonumber\\
& \leq M~\mathbb E \left [ \sum_{k=0}^{T_{i+1}-1} \alpha^k \right]+\mathbb E[ J_{\alpha}(T_{i+1}+b)] \nonumber\\
& \leq M~\mathbb E[ T_{i+1}]+\mathbb E[ J_{\alpha}(T_{i+1}+b)].
\end{align}

Moreover, by using \eqref{alpha-discounted-problem}, we obtain
\begin{align}
J_{\alpha}(\delta) \geq \mathbb E[\alpha^{T_{i+1}} J_{\alpha}(T_{i+1}+b)],
\end{align}
which yields
\begin{align}\label{bound2}
\mathbb E[J_{\alpha}(T_{i+1}+b)]-J_{\alpha}(\delta) & \leq  \mathbb E[(1-\alpha^{T_{i+1}}) J_{\alpha}(T_{i+1}+b)] \nonumber\\
& \leq \mathbb E[(1-\alpha^{T_{i+1}})] \frac{M}{1-\alpha} \nonumber\\
& \leq (1-\alpha^{\mathbb E[T_{i+1}]}) \frac{M}{1-\alpha}\nonumber\\
&=(1+\alpha+\cdots+\alpha^{\mathbb E[T_{i+1}]-1})~M \nonumber\\
&\leq M~\mathbb E[T_{i+1}].
\end{align}
The third inequality is due to the convexity of $\alpha^{k}$ as a function of $k$ and Jensen's inequality. 

Equations \eqref{bound1} and \eqref{bound2} implies \eqref{bounded_alpha}. This concludes the proof.}

\subsection{Toy Example}\label{ToyExample}
Let $X_t$ is a Markov chain and $Y_t = f(X_{t-d})$. One can view $X_t$ as the input of a causal system with delay $d \geq 0$, and $Y_t$ as the system output. We need to predict $Y_t$ based on the observation $X_{t-\delta}$. Then, we have the following lemma. 
 
\begin{lemma}\label{ToyExampleLemma1}
If $Y_t = f(X_{t-d})$, $X_t$ is a Markov chain, and the training and inference datasets have similar empirical distributions, then $H_L( \tilde Y_0 | \tilde X_{\delta})$ and $H_L(Y_t;\tilde Y_0 | X_{t-\delta})$ decrease with $\delta$ when $0 \leq \delta \leq d$ and increase with $\delta$ when $\delta \geq d$.
\end{lemma}
\begin{proof}
If the training and inference datasets have similar empirical distributions, by using Lemma \ref{lemma_inference} and definition of $L$-conditional entropy \eqref{eq_cond_entropy1}, we can show 
\begin{align}\label{Toy1}
H_L(Y_{t}; \tilde Y_0 | X_{-\delta})&= H_L(Y_{t} | X_{t-\delta}). \\ \label{Toy2}
H_L( \tilde Y_0 | \tilde X_{\delta})&= H_L(Y_{t} | X_{t-\delta}).
\end{align}
Now, we only need to prove that $H_L(Y_{t} | X_{t-\delta})$ decreases with $\delta$ when $0 \leq \delta \leq d$ and increases with $\delta$ when $\delta \geq d$.

Because $Y_t = f(X_{t-d})$ and $X_t$ is a Markov chain, $Y_t \leftrightarrow X_{t-\delta} \leftrightarrow X_{t-(\delta-1)}$ is a Markov chain for all $0 \leq \delta \leq d$. By the data processing inequality for $L$-conditional entropy \cite[Lemma 12.1] {Dawid1998}, one can show that  for all $0 \leq \delta \leq d$, 
\begin{align}
 H_L(Y_{t} | X_{t-\delta}) \leq H_L(Y_{t} | X_{t-(\delta-1)})
\end{align}
This proves that $H_L(Y_{t} | X_{t-\delta})$ decreases with $\delta$ when $0 \leq \delta \leq d$. 

Next, since $Y_t = f(X_{t-d})$ and $X_t$ is a Markov chain, $Y_t \leftrightarrow X_{t-\delta} \leftrightarrow X_{t-(\delta+1)}$ is a Markov chain for all $\delta \geq d$.  By the data processing inequality \cite[Lemma 12.1] {Dawid1998}, one can show that for all $ \delta \geq d$, 
\begin{align}
H_L(Y_{t} | X_{t-\delta}) \leq H_L(Y_{t} | X_{t-(\delta+1)}).
\end{align}
This proves that $H_L(Y_{t} | X_{t-\delta})$ increases with $\delta$ when $\delta \geq d$. 
\end{proof}

\end{document}